\newif\ifcomments

\commentsfalse

\documentclass[review=false]{jfp-epi}

\usepackage[utf8]{inputenc}
\usepackage[verbose]{newunicodechar}
\usepackage[capitalise]{cleveref} 
\Crefname{example}{Example}{Examples}
\Crefname{definition}{Definition}{Definitions}
\Crefname{law}{Law}{Laws}
\usepackage[textsize=tiny]{todonotes} 
\usepackage[misc]{ifsym}

\crefname{appendix}{Appendix}{Appendices}

\ifcomments
\newcommand\jp[1]{\todo[backgroundcolor=blue!10]{JP: #1}}
\newcommand\andreas[1]{\todo[backgroundcolor=green!10]{AA: #1}}
\else
\newcommand\jp[1]{}
\newcommand\andreas[1]{}
\fi

\usepackage{tikz}
\usetikzlibrary{matrix}
\newwrite\boxesfile\immediate\openout\boxesfile=main.boxes
\newsavebox{\marxupbox}
\DeclareMathAlphabet{\mymathds}{U}{BOONDOX-ds}{m}{n}
\newcommand\llbrace{\mathopen{{\{\kern-0.24em|}}}    
\newcommand\rrbrace{\mathclose{{|\kern-0.24em\}}}}   
\newcommand{\bidirmultimap}{\mathrel{\ooalign{$\mkern6mu\multimap$\cr$\reflectbox{\(\multimap\)}$\cr}}}
\newunicodechar{ }{{\,}}
\newunicodechar{¬}{\ensuremath{\neg}} %
\newunicodechar{²}{^2}
\newunicodechar{³}{^3}
\newunicodechar{·}{\ensuremath{\cdot}}
\newunicodechar{¹}{^1}
\newunicodechar{×}{\ensuremath{\times}} %
\newunicodechar{÷}{\ensuremath{\div}} %
\newunicodechar{ʳ}{^r}
\newunicodechar{ˡ}{^l}
\newunicodechar{Γ}{\Gamma}   %
\newunicodechar{Δ}{\ensuremath{\Delta}} %
\newunicodechar{Η}{\ensuremath{\textrm{H}}} %
\newunicodechar{Θ}{\ensuremath{\Theta}} %
\newunicodechar{Λ}{\ensuremath{\Lambda}} %
\newunicodechar{Ξ}{\ensuremath{\Xi}} %
\newunicodechar{Π}{\ensuremath{\Pi}}   %
\newunicodechar{Σ}{\ensuremath{\Sigma}} %
\newunicodechar{Φ}{\ensuremath{\Phi}} %
\newunicodechar{Ψ}{\ensuremath{\Psi}} %
\newunicodechar{Ω}{\ensuremath{\Omega}} %
\newunicodechar{α}{\ensuremath{\mathnormal{\alpha}}} %
\newunicodechar{β}{\ensuremath{\beta}} %
\newunicodechar{γ}{\ensuremath{\mathnormal{\gamma}}} %
\newunicodechar{δ}{\ensuremath{\mathnormal{\delta}}} %
\newunicodechar{ε}{\ensuremath{\mathnormal{\varepsilon}}} %
\newunicodechar{ζ}{\ensuremath{\mathnormal{\zeta}}} %
\newunicodechar{η}{\ensuremath{\mathnormal{\eta}}} %
\newunicodechar{θ}{\ensuremath{\mathnormal{\theta}}} %
\newunicodechar{ι}{\ensuremath{\mathnormal{\iota}}} %
\newunicodechar{κ}{\ensuremath{\mathnormal{\kappa}}} %
\newunicodechar{λ}{\ensuremath{\mathnormal{\lambda}}} %
\newunicodechar{μ}{\ensuremath{\mathnormal{\mu}}} %
\newunicodechar{ν}{\ensuremath{\mathnormal{\mu}}} %
\newunicodechar{ξ}{\ensuremath{\mathnormal{\xi}}} %
\newunicodechar{π}{\ensuremath{\mathnormal{\pi}}} %
\newunicodechar{ρ}{\ensuremath{\mathnormal{\rho}}} %
\newunicodechar{σ}{\ensuremath{\mathnormal{\sigma}}} %
\newunicodechar{τ}{\ensuremath{\mathnormal{\tau}}} %
\newunicodechar{φ}{\ensuremath{\mathnormal{\varphi}}} %
\newunicodechar{χ}{\ensuremath{\mathnormal{\chi}}} %
\newunicodechar{ψ}{\ensuremath{\mathnormal{\psi}}} %
\newunicodechar{ω}{\ensuremath{\mathnormal{\omega}}} %
\newunicodechar{ϕ}{\ensuremath{\mathnormal{\phi}}} %
\newunicodechar{ϵ}{\ensuremath{\mathnormal{\epsilon}}} %
\newunicodechar{ᵏ}{^k}
\newunicodechar{ᵢ}{_i} %
\newunicodechar{•}{\ensuremath{\cdot}} %
\newunicodechar{ }{{\,}}
\newunicodechar{′}{\ensuremath{^\prime}}  %
\newunicodechar{″}{\ensuremath{^\second}} %
\newunicodechar{‴}{\ensuremath{^\third}}  %
\newunicodechar{ⁱ}{^i}
\newunicodechar{⁵}{\ensuremath{^5}}
\newunicodechar{⁺}{\ensuremath{^+}} 
\newunicodechar{⁻}{\ensuremath{^-}} 
\newunicodechar{ⁿ}{^n}
\newunicodechar{₀}{\ensuremath{_0}} %
\newunicodechar{₁}{\ensuremath{_1}}
\newunicodechar{₂}{\ensuremath{_2}}
\newunicodechar{₃}{_3}
\newunicodechar{₊}{\ensuremath{_+}} 
\newunicodechar{₋}{\ensuremath{_-}} 
\newunicodechar{ₙ}{_n}
\newunicodechar{ₚ}{_p} %
\newunicodechar{ℂ}{\ensuremath{\mathds{C}}} %
\newunicodechar{ℕ}{\mathbb{N}} %
\newunicodechar{ℝ}{\ensuremath{\mathbb{R}}} %
\newunicodechar{ℤ}{\ensuremath{\mathbb{Z}}} %
\newunicodechar{ℳ}{\mathcal{M}}
\newunicodechar{⅋}{\ensuremath{\parr}} %
\newunicodechar{←}{\ensuremath{\leftarrow}} %
\newunicodechar{↑}{\ensuremath{\uparrow}} %
\newunicodechar{↓}{\ensuremath{\downarrow}} %
\newunicodechar{→}{\ensuremath{\rightarrow}} %
\newunicodechar{↔}{\ensuremath{\leftrightarrow}} %
\newunicodechar{↖}{\nwarrow} %
\newunicodechar{↗}{\nearrow} %
\newunicodechar{↝}{\ensuremath{\rightsquigarrow}}
\newunicodechar{↦}{\ensuremath{\mapsto}}
\newunicodechar{⇝}{\ensuremath{\rightsquigarrow}} %
\newunicodechar{↭}{\ensuremath{\leftrightsquigarrow}} %
\newunicodechar{↾}{\ensuremath{\upharpoonright}} %
\newunicodechar{↿}{\ensuremath{\upharpoonleft}} %
\newunicodechar{⇆}{\ensuremath{\leftrightarrows}} %
\newunicodechar{⇐}{\ensuremath{\Leftarrow}} %
\newunicodechar{⇒}{\ensuremath{\Rightarrow}} %
\newunicodechar{⇔}{\ensuremath{\Leftrightarrow}} %
\newunicodechar{∀}{\ensuremath{\forall}}   %
\newunicodechar{∃}{\ensuremath{\exists}} %
\newunicodechar{∅}{\ensuremath{\varnothing}} %
\newunicodechar{∈}{\ensuremath{\in}}
\newunicodechar{∉}{\ensuremath{\not\in}} %
\newunicodechar{∋}{\ensuremath{\ni}}  %
\newunicodechar{∎}{\ensuremath{\qed}}
\newunicodechar{∏}{\prod}
\newunicodechar{∑}{\sum}
\newunicodechar{∗}{\ensuremath{\ast}} %
\newunicodechar{∘}{\ensuremath{\circ}} %
\newunicodechar{∙}{\ensuremath{\bullet}} %
\newunicodechar{∙}{\ensuremath{\cdot}}
\newunicodechar{∞}{\ensuremath{\infty}} %
\newunicodechar{∣}{\ensuremath{\mid}} %
\newunicodechar{∧}{\ensuremath{\wedge}}%
\newunicodechar{∨}{\ensuremath{\vee}}%
\newunicodechar{∩}{\ensuremath{\cap}} %
\newunicodechar{∪}{\ensuremath{\cup}} %
\newunicodechar{∫}{\int}
\newunicodechar{∷}{::} %
\newunicodechar{∼}{\ensuremath{\sim}} %
\newunicodechar{≃}{\ensuremath{\simeq}} %
\newunicodechar{≅}{\ensuremath{\cong}} %
\newunicodechar{≈}{\ensuremath{\approx}} %
\newunicodechar{≔}{:=} %
\newunicodechar{≜}{\ensuremath{\stackrel{\scriptscriptstyle {\triangle}}{=}}}
\newunicodechar{≜}{\triangleq}
\newunicodechar{≝}{\ensuremath{\stackrel{\scriptscriptstyle {\text{def}}}{=}}}
\newunicodechar{≟}{\ensuremath{\stackrel {_\text{\textbf{?}}}{\text{\textbf{=}}\negthickspace\negthickspace\text{\textbf{=}}}}}
\newunicodechar{≠}{\ensuremath{\neq}}%
\newunicodechar{≡}{\ensuremath{\equiv}}%
\newunicodechar{≤}{\ensuremath{\leq}} %
\newunicodechar{≥}{\ensuremath{\geq}} %
\newunicodechar{≺}{\ensuremath{\prec}}
\newunicodechar{≼}{\ensuremath{\preceq}}
\newunicodechar{⊂}{\ensuremath{\subset}} %
\newunicodechar{⊃}{\ensuremath{\supset}} %
\newunicodechar{⊆}{\ensuremath{\subseteq}} %
\newunicodechar{⊇}{\ensuremath{\supseteq}} %
\newunicodechar{⊎}{\ensuremath{\uplus}} %
\newunicodechar{⊑}{\ensuremath{\sqsubseteq}} %
\newunicodechar{⊒}{\ensuremath{\sqsupseteq}} %
\newunicodechar{⊓}{\ensuremath{\sqcap}} %
\newunicodechar{⊔}{\ensuremath{\sqcup}} %
\newunicodechar{⊕}{\ensuremath{\oplus}} %
\newunicodechar{⊗}{\ensuremath{\otimes}} %
\newunicodechar{⊞}{\ensuremath{\boxplus}} %
\newunicodechar{⊟}{\ensuremath{\boxminus}} %
\newunicodechar{⊡}{\ensuremath{\boxdot}} %
\newunicodechar{⊢}{\ensuremath{\vdash}} %
\newunicodechar{⊣}{\ensuremath{\dashv}}
\newunicodechar{⊤}{\ensuremath{\top}}%
\newunicodechar{⊥}{\ensuremath{\bot}} 
\newunicodechar{⊧}{\vDash} %
\newunicodechar{⊨}{\vDash} %
\newunicodechar{⊩}{\Vdash}
\newunicodechar{⊸}{\ensuremath{\multimap}} %
\newunicodechar{⋁}{\ensuremath{\bigvee}}
\newunicodechar{⋃}{\ensuremath{\bigcup}} %
\newunicodechar{⋄}{\ensuremath{\diamond}} %
\newunicodechar{⋆}{\ensuremath{\star}} %
\newunicodechar{⋮}{\ensuremath{\vdots}} %
\newunicodechar{⋯}{\ensuremath{\cdots}} %
\newunicodechar{─}{---}
\newunicodechar{□}{\ensuremath{\Box}} %
\newunicodechar{▹}{\ensuremath{\rhd}} %
\newunicodechar{◂}{\ensuremath{\triangleleft}}
\newunicodechar{▴}{\ensuremath{\mathbin{\triangle}}}
\newunicodechar{◃}{\triangleleft{}} %
\newunicodechar{◆}{\blacklozenge} %
\newunicodechar{◇}{\ensuremath{\Diamond}} %
\newunicodechar{◽}{\ensuremath{\square}} %
\newunicodechar{★}{\ensuremath{\star}}   %
\newunicodechar{♢}{\ensuremath{\diamondsuit}} %
\newunicodechar{♭}{\ensuremath{\flat}} %
\newunicodechar{♯}{\ensuremath{\sharp}} %
\newunicodechar{✓}{\ensuremath{\checkmark}} %
\newunicodechar{⟂}{\ensuremath{\bot}} 
\newunicodechar{⟦}{\ensuremath{\llbracket}}
\newunicodechar{⟧}{\ensuremath{\rrbracket}}
\newunicodechar{⟨}{\ensuremath{\langle}} %
\newunicodechar{⟩}{\ensuremath{\rangle}} %
\newunicodechar{⟶}{\ensuremath{\longrightarrow}}
\newunicodechar{⟷}{\longleftrightarrow}
\newunicodechar{⟹}{\ensuremath{\Longrightarrow}} %
\newunicodechar{⦃}{\ensuremath{\llbrace}} %
\newunicodechar{⦄}{\ensuremath{\rrbrace}} %
\newunicodechar{⧟}{\ensuremath{\bidirmultimap}} %
\newunicodechar{⨁}{\bigoplus} %
\newunicodechar{⨂}{\bigotimes} %
\newunicodechar{⫪}{\ensuremath{\top\kern-1.3ex\top}}
\newunicodechar{ⱼ}{_j} %
\newunicodechar{𝒟}{\ensuremath{\mathcal{D}}} %
\newunicodechar{𝒢}{\ensuremath{\mathcal{G}}}
\newunicodechar{𝒦}{\ensuremath{\mathcal{K}}} %
\newunicodechar{𝔸}{\ensuremath{\mathds{A}}} %
\newunicodechar{𝔹}{\ensuremath{\mathds{B}}} %
\newunicodechar{𝟘}{\ensuremath{\mymathds{0}}} %
\newunicodechar{𝟙}{\ensuremath{\mymathds{1}}} %
\newunicodechar{≋}{\ensuremath{\approx}}
\newunicodechar{≷}{\ensuremath{\lessgtr}}

\ifdefined\LONGVERSION
\newcommand{\SHORTVERSION}[1]{}
\else
\newcommand{\LONGVERSION}[1]{}
\newcommand{\SHORTVERSION}[1]{#1}
\fi

\newcommand{\EXTENDED}[1]{}
\newwrite\agdarefs
\immediate\openout\agdarefs=References.agda

\write\agdarefs{{-\string# OPTIONS --warning=error \string#-}}
\write\agdarefs{module References where}
\write\agdarefs{}
\newcommand\agdacomment[1]{\write\agdarefs{-- #1}}
\newcommand\agdablank{\write\agdarefs{}}
\NewDocumentCommand\agdaref{o m o m}{%
  \href{html/#2.html\#\IfValueT{#3}{#3.}#4}{✓}%
  \IfValueT{#1}{\agdacomment{#1}}%
  \write\agdarefs{import #2}%
  \write\agdarefs{open #2\IfValueT{#3}{.#3} using (#4)} %
  \write\agdarefs{}%
}

\newcommand{\NOVERSION}[1]{}  
\newcommand\seeExtendedMaterial{(See extended material)}
\newcommand\category[1]{\ensuremath{\mathsf{#1}}}
\newcommand\structure[1]{\ensuremath{\mathcal{#1}}}
\newcommand\Funct[1]{\structure{E}(\category{#1})} 
\NewDocumentCommand\Lamp{o}{\ensuremath{λ^{\mathsf{≼}}\IfNoValueTF{#1}{}{_{\structure{#1}}}}\xspace}
\NewDocumentCommand\PLamp{o}{\ensuremath{λ^{◃}\IfNoValueTF{#1}{}{_{\structure{#1}}}}\xspace}

\SHORTVERSION{
\newcommand\proofInSupp{The proof can be found in the long version of
  the paper, provided as supplementary material.
  \renewcommand\proofInSupp{\hfill\seeExtendedMaterial} }
}
\LONGVERSION{
\newcommand\proofInSupp{}
}

\newenvironment{definition*}{}{} 

\usepackage{dsfont}
\usepackage{mathpartir}
\usepackage{mathtools} 
\usepackage{calc} 

\usepackage{booktabs}   
\usepackage{subcaption} 

\newtheorem{remark}{Remark}

\setcopyright{none} 
\acmDOI{10.1145/3408972}
\acmYear{2026}
\copyrightyear{2026}
\acmJournal{PACMPL}
\acmVolume{10}
\acmNumber{ICFP}
\acmArticle{90}
\acmMonth{8}

\input{syntaxMacros.tex}

\newcommand{\abbrev}[1]{#1} 

\newcommand{\mgoal}[1][]{\mbox{goal\ifthenelse{\equal{#1}{}}{}{~#1}}}

\newcommand{\rulename}[1]{\ensuremath{\mbox{\sc#1}}}
\newcommand{\ru}[2]{\dfrac{\begin{array}[b]{@{}c@{}} #1 \end{array}}{#2}}

\newcommand{\nru}[3]{#1\ \ru{#2}{#3}}
\newcommand{\nrux}[4]{#1\ \ru{#2}{#3}\ #4}
\newcommand{\dstack}[2]{\begin{array}[b]{c}#1\\#2\end{array}}

\begin{document}

\title{%
  Arbitrary Functors as Modalities in a Linear Calculus
  \LONGVERSION{ (Extended~Version)}}


\author{Jean-Philippe Bernardy}
\orcid{0000-0002-8469-5617}             
\affiliation{
  \department{Department of Computer Science and Engineering}
  \institution{University of Gothenburg and Chalmers University of Technology}
  \city{Göteborg}
  \country{Sweden}
}
\email{jean-philippe.bernardy@gu.se}          

\author{Andreas Abel}
\orcid{0000-0003-0420-4492}             
\affiliation{
  \department{Department of Computer Science and Engineering}
  \institution{University of Gothenburg and Chalmers University of Technology}
  \city{Göteborg}
  \country{Sweden}
}
\email{andreas.abel@gu.se}          


\begin{abstract}
  We study modalities in the linear lambda calculus, by interpreting
  them as endofunctors in monoidal closed categories. We support
  functors that do not distribute over products, including (co)monads.
  We also support arbitrary natural transformations as conversion between modalities.
  We provide a generic categorical semantics, which we
  then specialise to a sound and complete resource semantics. We use
  this resource semantics to derive normalisation by evaluation.
  Along the way, we describe several applications of such models.
  This work aims at providing a foundation for future programming
  languages with modalities.  The paper is backed by a mechanisation
  in Agda.
\end{abstract}



\maketitle



%

\section{Introduction}\label{sec:intro}

Modalities, rooted in modal logic, provide a formal framework for
expressing notions of necessity, possibility, and other modal
operators.  In programming languages they can be used to model
resource usage, concurrency, effects and coeffects, and security
constraints~\citep{murphy_distributed_2005,dal_lago_bounded_2009,girard_new_1991,orchard_quantitative_2019,petricek_coeffects_2014}.
While the term \emph{grades} is more common in the programming
community, we use the word \emph{modalities} consistently throughout
this paper.

Modalities are typically assumed to form a preordered, ring‑like
structure~\citep{abel_unified_2020,brunel_core_2014}.
Such a structure is sufficiently general to accommodate
all of the applications mentioned above, which makes it tempting to
unify them within a single system.  However, the ring structure
imposes (1) distributivity \(p(q+r)=pq+pr\)
and (2) that modality conversion is computationally irrelevant, i.e. there
is a single interpretation of the preorder relation \(p ≼ q\).  Our
starting observation is that these two assumptions break down when one
pushes the generic system to its limits.

Consider first
distributivity. As we see below (\cref{lem:distr-lax}), distributivity
corresponds to modalities having the structure of (op)lax monoidal
functors; meaning that there are these conversions:
\begin{center}
  \begin{tikzpicture}
    \node(factored)         at (-1,0)                   {
      \(\begin{array}{c}
        \BOX p (A \Tensor B)
      \end{array}\)
    };
    \node(expanded)      at (2,0)  {\(
      \begin{array}{c}
        \BOX p A \Tensor \BOX p B
      \end{array}
      \)};
    \path [->] (factored) edge [bend left]
      node[above,font=\scriptsize] {$\mathit{oplax}$}
      (expanded);
    \path [->] (expanded) edge [bend left]
      node[below,font=\scriptsize] {$\mathit{lax}$}
      (factored);
  \end{tikzpicture}
\end{center}
But even commonplace effects such as the writer monad
are not lax unless the payload is itself a monoid. Furthermore, if one
simultaneously cares about resource usage, the underlying language may
be linear, and then even the reader monad
is not always lax.\footnote{Indeed, laxity of the reader monad
  \(\BOX{\mathsf{reader}_R}A = R \multimap A\) requires
  \(ψ_2 : (R ⊸ A) ⊗ (R ⊸ B) ⊸ R ⊸ (A ⊗ B)\), which cannot be realised
  unless \(R\) is a comonoid; and not all types are comonoids in a
  linear language.}

Regarding the second point, assuming that there is a single possible
conversion \(p ≼ q\) between modalities \(p\) and \(q\) may be reasonable if there is
a single kind of effect, or if they do not interact.~\citep{abel_unified_2020,orchard_quantitative_2019} This breaks down
also if one considers data structures as modalities. In such a case a
modality conversion is a reshaping of the structure, and such conversions are not
determined by their type.

Our goal is to study type-systems where modalities are not necessarily
distributive and where conversion carries computational content.
\subsection{Roadmap}

The idea that guided this work is that it would suffice to take an
existing core language with modalities and relax it appropriately so that:
\begin{itemize}
\item not all modalities automatically distribute, and
\item modality conversions are made explicit.
\end{itemize}
(See \cref{sec:modalities} for the precise description of the new modality structure.)
Then these relaxations would be applied everywhere in the meta-theory
in a straightforward way. However, we found that it is not so
simple. Following a standard method of semantic analysis, we found
that the calculus must be thoroughly modified in a pervasive way.

In this article, we could just present the final calculus, but we believe that
doing so would be less pedagogical, and we prefer to guide the reader
through the true path which led us to it. This means that our
methodology sometimes departs from a standard description of core
calculi. We explain the detours taken in the rest of this section.

First, we define our starting point core language. It is meant to
resemble usual calculi with modalities (\cref{sec:calculus}). We make
the minimal modifications necessary to relax the assumptions that all
modalities distribute. We do so by adding side conditions to the rules
which would normally just assume distributivity. Unfortunately this
calculus does not seem to have the desirable properties. In particular
it is unclear how to reduce some redexes.

Rather than extending the syntax haphazardly, we investigate the
implications of the design by semantic means.  We do so by providing a
generic categorical semantics (\cref{sec:semantics}); then a (sound
and complete) resource semantics based on presheaves
(\cref{sec:forgetful-model,sec:computational-model}). Along the way,
we consider usual and novel applications of modalities and grades,
sometimes via the relational model (\cref{sec:relmodel}) arising from
the generic semantics.

Generally, presheaf models can be utilized in normalisation
by evaluation techniques: one evaluates the syntax in a presheaf model
and then reifies back to syntax. The reification outputs forms which
do not context redexes, and are in this sense normal forms.
Furthermore, the completeness proof strongly suggests what normal
forms one should have.  In our case, it provides us with the design of the
stable calculus---one with all the desired properties, see
\cref{sec:normalisation}. This final language has the same semantics
as the original calculus, but its syntax departs significantly from
usual lambda calculi. Our conclusion is that one can not relax
modalities to be non-distributive without paying a significant cost in
terms of syntactic complications. Yet, we remark that the simplified
calculus (the one we started with) appear to be sufficient to describe
most applications, even if not all syntactic reasoning can be done
using it.

We suggest that the reader eager to see this final version of the
calculus does not skip directly to \cref{sec:normalisation}, but first
goes through \cref{sec:modalities,sec:semantics} and
\cref{thm:completeness}.

\subsection{Methodological notes}

We have completely mechanised our results in the form of an Agda
library.  Because of the modular structure of the paper, readers may
need to navigate between concepts, hence it is heavily hyperlinked,
both within the document and to the Agda formalisation. In particular,
whenever a checkmark (✓) appears, it links to the corresponding
formalisation.\footnote{To reviewers: to access the formalisation
  in this way, the extended material must be unpacked so that the
  \texttt{html} directory resides in the same directory as the PDF
  file. Additionally, the PDF reader must support local links, as many
  do not. Google Chrome works on all tested platforms, and Evince
  works on Linux.}

Thanks to its generality and methodology, we hope that this
paper will provide a solid foundation for the development of future
languages with modalities.

\section{The Structure of Modalities}\label{sec:modalities}
\agdacomment{Section: Modality Structure}
\agdacomment{---------------------------}
\agdablank

\providecommand\modSet[1]{M_{\structure #1}}
\begin{definition}[modality structure]\label[definition]{def:ringoid}\agdaref[Modality structure]{Modalities}{S}
  A modality structure is a category augmented with:
   \begin{itemize}
     \item Two monoidal structures: $(·,1)$ and $(+,0)$.
     \item Proof-relevant predicates $\distributive$ and $\codistributive$ that single out distributive and codistributive modalities.
        I.e., the propositions \(\distributes p \) and \(\codistributes p\) are actually sets containing witnesses of (co)distributivity.
     \item A distinguished element of \(\distributes 1\) and \(\codistributes 1\).
     \end{itemize}
   \end{definition}
   To distinguish modality structures we will denote them with a script font.
    The object set of a modality structure \structure F, written \(\modSet F\), contains the modalities, and
   the homsets of the above category represents convertibility between modalities.
   To our knowledge, programming language literature~\citep{abel_unified_2020,orchard_quantitative_2019,brunel_core_2014} assumes that each such homset has at most one morphism.
   In other words, modalities form a poset. (See also \cref{sec:poset-modalities}.)
   Even though we do not assume such, it remains an important motivating situation, and hence we write \(p ≼ q\) for the set of conversions from \(p\) to \(q\).
   %

Modalities $p$ act on types $A$ via qualification \(\BOX p A\), producing a modal type.
The calculus is designed such that the following principles of interaction between modalities and other type formers hold.
First, the additive monoid interacts with \(\Tensor\) types like so:
\begin{align}
\BOX {0} A & \funOp 1 \tag{zero principle} \label{principle:zero}
  & \text{\agdaref[Zero principle]{Natural.Pats.Examples}{tmThrow}} \\
\BOX {p+q} A & \funOp \BOX p A \Tensor \BOX q A \tag{sum principle} \label{principle:sum}
  & \text{\agdaref[Sum principle]{Natural.Pats.Examples}{tmSplit}}
\end{align}
Hence the modality \(0\) corresponds to no use, and \(p+q\) corresponds to
combined \(p\)-use and \(q\)-use.
Multiplication arises from composition, and is governed by the
following principles (in this case, bimaps, unlike the addition principles above):
\begin{align}
\BOX 1 A & \funOpBidir A  & \text{\agdaref[Unit principle]{Natural.Pats.Examples}{gradedMult}} \tag{unit principle} \label{principle:unit}  \\
\BOX p {(\BOX q A)} & \funOpBidir \BOX {p · q} A & \text{\agdaref[Product principle]{Natural.Pats.Examples}{gradedUnit}} \tag{product principle} \label{principle:product}
\end{align}
The unit qualification is always equivalent to no qualification, and
combination behaves multiplicatively.
Because qualification also behaves functorially \agdaref[Functoriality]{Natural.Pats.Examples}{boxFmap},
\(\BOX p\) is always simultaneously a graded monad and a graded co-monad.
\begin{equation}
\text{If~} {⊢ A \funOp B}, \text{~then~} {⊢ \BOX p A \funOp \BOX p B} \tag{functoriality}\label{principle:functoriality}
\end{equation}
\noindent
A modality morphism can be lifted to a function:\agdaref[Subsumption]{Natural.Pats.Examples}{convertibility}
\begin{equation}
\text{If~} p ≼ q, \text{~then~} \BOX p A \funOp \BOX q A . \tag{lift} \label{principle:convertibility}
\end{equation}

Most of the literature on modalities posits that \((·)\) distributes
over \((+)\), and then \(\BOX p\) distributes over products.  As
mentioned in the introduction, we will not make this assumption in
general, and restrict distributivity to the modalities $p$ which satisfy
$\distributive$ or $\codistributive$, like so:
\providecommand\boththings[2]{\begin{cases}#1\\#2\end{cases}}
\begin{align}
\text{Given~} \distributes p & \text{~we have~}
  \boththings
    {\fun 1 \BOX p 1
       \text{\agdaref[Oplax nullary]{Natural.Pats.Examples}{oplax0}}}
    {\fun {(\BOX p A \Tensor \BOX p B)} \BOX p {(A \Tensor B)}
       \text{\agdaref[Oplax binary]{Natural.Pats.Examples}{oplax}}}
  \tag{factorisation}\label{principle:lax}
\\
\text{Given~} \codistributes p & \text{~we have~}
  \boththings
    {\fun {\BOX p 1} 1
      \text{\agdaref[Lax nullary]{Natural.Pats.Examples}{lax0}}}
    {\fun {\BOX p {(A \Tensor B)}} (\BOX p A \Tensor \BOX p B)
      \text{\agdaref[Lax binary]{Natural.Pats.Examples}{lax}}}
  \tag{distribution}\label{principle:oplax}
\end{align}
The purpose of \(\distributes p\) and \(\codistributes p\) is to carry
the computational content associated with the above types.
When modalities from a poset,
\(\distributes p\) and \(\codistributes p\) are just conditions.
The articulation between those sets (or conditions) and distributivity is explained
in detail in \cref{sec:semantics}.

\section{Lambda Calculus with General Modalities}\label{sec:calculus}
\agdacomment{Section: Lambda Calculus with General Modalities}
\agdacomment{------------------------------------------------}
\agdablank

In this section we introduce  $\Lamp[{F}]$, a
typed lambda calculus augmented with modalities in \(\structure
F\).
It is designed as a minimal extension of modern lambda-calculi \cite{abel_unified_2020,Orchard_2020} with modalities, to
support non-distributive modalities and explicit modality conversions.
The calculus can also be considered as a modal (linear) logic.
For most of the paper, we only consider simple types, but discuss polymorphism in \cref{sec:polymorphism}.



Types are given by the following grammar.\agdaref[Types]{Natural.Types}{Typ}
\begin{align*}
  A,B,C & ::= \basetype ∣ 1 ∣ A \Tensor B ∣ \Ttop ∣ A \With B ∣  \fun A B ∣ \BOX p A 
\end{align*}
The letter $\basetype$ ranges over a set $\TyConst$ of base
types. These types will be used in the semantics to freely
interpret type variables beyond the types formed from $1$, $\With$,
$\Tensor$, $\funOp$, and $\BOX p$, which have a fixed meaning.  Base types
$\basetype$ are generic unless equipped with primitive operations. We
have two product types $A \Tensor B$ (multiplicative) and $A \With B$
(additive), and they behave as in linear logic. Namely,
eliminating $A \Tensor B$ requires consuming both $A$ and $B$, whereas
$A \With B$ only allows either $A$ or $B$ to be consumed. The function
type $\fun A B$ is purely linear, and $\BOX p A$ is the qualification
of $A$ by an arbitrary modality $p$.

Consequently, $\Lamp$ is a
calculus for substructural typing~\citep{walker:substructural}. The
availability of structural rules is precisely controlled by
modalities.
There is no sum type, but bits of information can be exchanged using
the type $\fun {(X \With X)} X$, for some suitable continuation type
\(X\).  In general, we can define
$\Plus[X] A B = \fun {((\fun A X) \With (\fun B X))} X$. In the
presence of polymorphism we can even define
$\Plus A B = ∀α. \Plus[{α}] A B$, which is isomorphic to the usual sum
type.

\paragraph{Contexts}
In contexts, variable bindings are annotated with a modality:\agdaref[Contexts]{Contexts}{Ctx}
\begin{align*}
  Γ,Δ ::= \emptyContext \mid Γ, x{:}\moda q A \quad (x \not\in\Gamma)
\end{align*}

We can scale contexts by a modality $q$, pointwise:\agdaref[Context scaling]{Contexts}{_◃_}
\begin{align*}
  q·\emptyContext & = \emptyContext \\
  q·(x : \moda p A, Γ) & = x : \moda {q·p} A , q · Γ
\end{align*}
Thus contexts are a free module generated by simple bindings
of form \(x : A\).

\paragraph*{Terms and typing}
Our syntax follows that of typed usual lambda calculus with patterns:
\begin{align*}
  t,u & ::= x ∣ \lam m t \mid t \app u 
  & \text{variables; functions} \\
  & ∣ \ttop{Γ} ∣ \twith t u ∣ \tout 1 t ∣ \tout 2 t & \text{\(\With\)-tuples} \\
  & ∣ () ∣ (t,u)  & \text{\Tensor-tuples} \\
  & ∣ \TBOX p t & \text{qualification} \\
  & ∣ \tlet m t u & \text{pattern binding}
\end{align*}
But we allow pattern bindings in \(\ttlet\) and abstraction. Patterns admit the following syntax:
\begin{align*}
  l,m & ::= x ∣ \pthrow ∣ \panalyse p l q m ∣(\pconv {φ} m) ∣ (l,m) ∣ () ∣ \TBOX p m
\end{align*}
The typing judgement has the form $Γ ⊢ t : A$, meaning that $t$ has
type $A$ (with an implicit modality 1) in context $Γ$.
A variable occurrence always corresponds to usage
1.  This ensures stability of usage under variable
substitution. In the rest of the paper, any omitted modality annotation defaults to $1$.
The pattern typing judgement \(\typat {Γ} m p A\) means that with modality \(p\), the pattern \(m\) analyses the type \(A\) into the variable bindings \(Γ\).

%
  
\newcommand\figsepar[1]{	\hrulefill \\ \small #1 }
\begin{figure}
  \figsepar{Terms}
\begin{mathpar}
  {\rul{var}}\inferrule{{}}{x : \moda 1 A ⊢ x : A}


  {\rul{abs}}\inferrule{\typat {Δ} m 1 A \\ Γ,Δ ⊢ t : B}{Γ ⊢ \lam m t : \fun A B}
  \and
  {\rul{app}}\inferrule{Γ ⊢ t : \fun A B\\ Δ ⊢ u : A}{Γ,Δ  ⊢ t u : B}\\





  {\rul{\Unit-intro}}\inferrule{\bla}{\emptyContext ⊢ () : \Unit}

  {\rul{\Tensor-intro}}\inferrule{Γ ⊢ t : A\\ \Delta ⊢ u : B}{Γ,Δ ⊢ (t,u) : A \Tensor B}

  {\rul{\Ttop-intro}}\inferrule{\bla}{Γ ⊢ \ttop{Γ} : \Ttop}
  \and
  {\rul{\(\With{}\)-elim}}\inferrule{Γ ⊢ t : A₁ \With A₂}{Γ ⊢ \tout i t : Aᵢ}
  \and
  {\rul{\(\With{}\)-intro}}\inferrule{Γ ⊢ t_1 : A_1 \\ \Gamma \vdash t_2 : A_2}{Γ ⊢ \twith {t₁} {t₂} : A_1 \With A_2}

  {\BOXintro}\inferrule{Γ ⊢ t : A}{pΓ ⊢ \TBOX p t : \BOX p A}{\distributesMaybe{Γ}{p}}

  {\rlet{}}\inferrule{ Γ ⊢ t : A  \\ \typat {Ξ} m 1 A \\ Δ,Ξ ⊢ u : C}{Δ,Γ ⊢ \wlet {m = t} u : C}

\end{mathpar}
\figsepar{Patterns}
\begin{mathpar}
  \inferrule{ϕ : p ≼ q \\ \typat {Γ} m q A}{\typat {Γ} {\pconv {ϕ} m} p A}

  \inferrule{\typat {Γ} l p A \\ \typat {Δ} m q A}{ \typat {Γ,Δ} {\panalyse p l q m} {p+q} A}

  \inferrule{~}{\typat {\emptyContext} {\pthrow} 0 A}

  \inferrule{~}{ \typat {x : \moda p A} x p A}

  \inferrule{\typat {Γ} l p A \\ \typat {Δ} m p B}{\typat {Γ,Δ} {(l,m)} p {(A ⊗ B)}}{\codistributes p}

  \inferrule{\bla}{\typat {\emptyContext} {()} p {\Unit}}{\codistributes p}

  \inferrule{\typat {Γ} m {qp} A}{ \typat {Γ} {\TBOX p m} {q\,} {\BOX p A}}
\end{mathpar}

\caption{%
  Typing rules of \Lamp \agdaref[Typing rules]{Natural.Pats.Syntax}{}. The judgement \(Γ ⊢ t : A\) is for terms and \(\typat {Γ} m p A\) is for patterns. } %
\label{fig:typing-rules}
\end{figure}

The typing rules are shown in \cref{fig:typing-rules}. A number of
points are worth detailing.  First, while the eliminators of
\(A \With B\) are usual projections, the elimination of
\(A \Tensor B\), \(\Unit\) and \(\BOX q A \) is done by combining
$\ttlet$ with patterns (respectively \((l,m)\), \(()\) and
\(\TBOX q m\)). We also allow pattern binding in \(λ\)-abstraction
for convenience.
All listed patterns accommodate arbitrary modalities, though
$(l,m)$ and $()$ impose side conditions, yielding the restriction
of \ref{principle:oplax}. Meanwhile, the $\TBOX q m$ pattern
sustains the \ref{principle:product} in the
left-to-right direction.

Second, several patterns handle modalities specifically, instead of type constructors.
The $\pconv {φ} m$ pattern applies the morphism
$φ : p ≼ q$ to adjust the modality of a pattern, sustaining \ref{principle:convertibility}.
The $\pthrow$ pattern matches any
type with modality $0$ without binding variables, sustaining
\ref{principle:zero}. Finally, $\panalyse p l q m$ matches a type with
modality $p+q$ by decomposing it into $l$ (matching modality $p$) and
$m$ (matching modality $q$), thereby supporting \ref{principle:sum}.

Finally, $\rul{\BOXintro}$ requires either a singleton context $Γ$ or
an element of $\distributes p$. Accordingly,
\ref{principle:functoriality} has no side condition, while
\ref{principle:lax} has one. A detailed semantic analysis of these
conditions follows in the next section. \Cref{fig:principles} shows
the lambda terms for each principle, demonstrating the purpose of
patterns.

\begin{figure}
  \centering
  \begin{align*}
    \input{moldy/zero_principle_term.tex} & : \input{moldy/zero_principle_type.tex}  & \text{\ref{principle:zero}} \\
    \input{moldy/sum_principle_term.tex} & : \input{moldy/sum_principle_type.tex}  & \text{\ref{principle:sum}} \\
    \input{moldy/unit_principle_l_term.tex} & : \input{moldy/unit_principle_l_type.tex}  & \text{\ref{principle:unit}, left to right} \\
    \input{moldy/unit_principle_r_term.tex} & : \input{moldy/unit_principle_r_type.tex}  & \text{\ref{principle:unit}, right to left} \\
    \input{moldy/product_principle_l_term.tex} & : \input{moldy/product_principle_l_type.tex}  & \text{\ref{principle:product}, left to right} \\
    \input{moldy/product_principle_r_term.tex} & : \input{moldy/product_principle_r_type.tex}  & \text{\ref{principle:product}, right to left} \\
    \input{moldy/distr_l_1_term.tex} & : \input{moldy/distr_l_1_type.tex}  & \text{\ref{principle:oplax}, unit} \\
    \input{moldy/distr_r_1_term.tex} & : \input{moldy/distr_r_1_type.tex}  & \text{\ref{principle:lax}, unit} \\
    \input{moldy/distr_l_2_term.tex} & : \input{moldy/distr_l_2_type.tex}  & \text{\ref{principle:oplax}, pair, given \codistributes p} \\
    \input{moldy/distr_r_2_term.tex} & : \input{moldy/distr_r_2_type.tex}  & \text{\ref{principle:lax}, pair, given \distributes p} \\
    \input{moldy/conv_principle_term.tex} & : \input{moldy/conv_principle_type.tex}  & \text{\ref{principle:convertibility}} \\
    \input{moldy/functor_principle_term.tex} & : \input{moldy/functor_principle_type.tex}  & \text{\ref{principle:functoriality}}
  \end{align*}
  \caption{Implementation of the principles relating modality qualification and other type formers.}
  \label{fig:principles}
\end{figure}


\subsection{Issues with Operational Semantics}

We arrived at the above syntax by making minimal changes to earlier
calculi to support non-distributive modalities. Specifically, we
conditioned some rules with \(\distributive\) and \(\codistributive\)
as needed and extended the pattern syntax for explicit conversions and
explicit handling of modality sums. This change allows us to support
the desired programming principles (\cref{fig:principles}) in a convenient way.

However, \Lamp{} is not stable under computation: comprehensive
extensions to pattern bindings are required.  The usual approach is to
define how eliminator patterns reduce with the corresponding
construction. For instance, for products, the binding
\((m,l) = (t,u)\) reduces to two parallel bindings: \(m=t\) and
\(l=u\). For \Lamp{}, the binding \(\TBOX p m = \TBOX p t\) would
reduce to a \(p\)-way pattern \(m =_p t\). This already complicates
bindings beyond \Lamp{} syntax. We could extend the syntax to support
this form, but there are further complications because of the
modality-transformation patterns \(\panalyse p l q m\) and
\(\pconv {φ} m\), which do not interact with the type-structure. For
instance, we could have bindings such as
\(\pconv {φ} {\panalyse p {(l,m)} q x} =_r (t,u)\). It is clear that
the pattern syntax must be further extended to reduce this binding
form, but it is unclear at this stage how exactly.

Instead of trying to complete the calculus in a
\textit{ad-hoc} way, we will \emph{derive} the final form of
patterns, using semantic means (\cref{sec:normalisation}). The first
step is to define a categorical semantics for \(\Lamp{}\).

\section{Categorical Semantics}\label{sec:semantics}
\agdacomment{Section: Categorical Semantics}
\agdacomment{------------------------------} \agdablank

\providecommand\ourcategory{SMC${}^{{\funOp},{\With}}$}

Let us, for the sake of presentation, consider a simplified version of
$\Lamp$ without modalities—effectively assuming that every modality is
1. This yields a purely linear language. The semantics of this
restricted language naturally aligns with a symmetric monoidal closed
category equipped with finite products, which we denote as
\ourcategory{}.

\begin{definition}[\ourcategory]\label[definition]{def:bicat}\agdaref[Symmetric Monoidal Closed Category plus Standard Products]{Cats}{BiCat}
  A \ourcategory{} is a category \(\category C\) augmented with:
  \begin{itemize}
  \item A symmetric monoidal structure \((\tensorObj , \unitObj )\).
  \item Finite products \((\withObj,\topObj)\).
    In contrast to the \(\tensorObj\) product, it has projections \(A \withObj B \hom A\) and \(A \withObj B \hom B\) and pairing \((C \hom A) × (C \hom B) → (C \hom (A \withObj B))\).
        This product is in general distinct from the \(\tensorObj\) product.
      \item Internal hom-sets \((A \funObj B)\), such that
        \((A \tensorObj B \hom C) ≅ (A \hom (B \funObj C))\). Note
        that these are adjoint to the \(\tensorObj\) product,
        not the \(\withObj\) product.
      \end{itemize}
  Standard coherence conditions apply.
  The \(\tensorObj\) and \(\withObj\) products are independent, but
  if they coincide the category is Cartesian closed.  We
  overload the notation for types and their interpretation,
  using color to distinguish them.
\end{definition}
The notation that we use the morphisms of a \ourcategory{} is as follows.
For the \((\tensorObj , \unitObj )\) monoid structure:
\begin{itemize}
\item If \(f : A → B\) and \(g : C → D\) then \(f ⊗ g : A \tensorObj C → B \tensorObj D\)
\end{itemize}
For finite products \((\withObj,\topObj)\):
\begin{itemize}
\item \(π₁ : A \withObj B \hom A\) and \(π₂ : A \withObj B \hom B\)
\item if \(f : C \hom A\) and \(g : C \hom B\) then  \(f ▴ g : C \hom (A \withObj B)\).
  \item \(ε : A → \topObj\)
\end{itemize}
The isomorphism \((A \tensorObj B \hom C) ≅ (A \hom (B \funObj C))\) is witnessed by:
\begin{itemize}
\item \(\appl   : ((A \funObj B) \tensorObj A) → B\)
\item if \(f :((A \tensorObj B) → C)\) then \(\curry f : A → (B \funObj C)\)
\end{itemize}

Modalities are interpreted as endofunctors on $\category{C}$. Crucially, the collection of such endofunctors forms a \nameref{def:ringoid}.

\begin{definition}[endofunctor modalities]\label[definition]{def:functor-modalities}\agdaref[Functor modalities]{Sem}{endoMod}
  Let {\Funct C} be the modality structure defined by:
\begin{align*}
  M & := \category C → \category C & \text{endofunctors} \\
  p ≼ q & := p ⟶ q & \text {natural transformations between those}\\
  p · q & := p ∘ q & \text {corresponding functor composition}\\
  1 & := \mathsf{id} &  \text{identity functor} \\
  p + q & := p \tensorObj q & \text {pointwise monoidal product}\\
  0 & := \KObj(\unitObj) &  \text{constant functor (map every object to the unit)} \\
  \distributes p,\codistributes p  & := \text{witness that \(p\) is (op)lax}
    & \text{(see Def.~\ref{def:laxoplax})}
\end{align*}
\end{definition}

Remarkably, there is no restriction on the chosen
functors or the chosen natural transformations. Indeed, the modality
structure arises from already assumed structures:
\begin{itemize}
\item Its category structure is provided by the category
  structure of natural transformations.
\item The monoidal product \((·,1)\) is given by the
  monoid of endofunctor composition \((∘, id)\).
\item The monoidal product \((+,0)\) is given by the monoidal
  product structure \((\tensorObj , \unitObj )\).
\end{itemize}

As we detail in the rest of the section, the choice of
\(\distributes p\) and \(\codistributes p\) is exactly that which is
required to interpret \Lamp, but it is also one which corresponds to
distributivity laws. Let us first review the definitions.

\begin{definition}[lax and oplax functors]\label{def:laxoplax}
  A functor \(F\) is lax \agdaref[Lax functor]{Cats}[BiCat.Endo.Functor]{IsLax}
  if there are natural morphisms
  \(F A \tensorObj F B \hom F (A \tensorObj B)\) and
  \(\unitObj \hom F \unitObj\).
  Conversely, it is oplax \agdaref[Oplax functor]{Cats}[BiCat.Endo.Functor]{IsOplax}
  if we have natural morphisms
  \(F (A \tensorObj B) \hom F A \tensorObj F B\) and
  \(F \unitObj \hom \unitObj\). The definition of lax and oplax functors come with a number of additional coherence laws,
  but they are orthogonal to our development.
\end{definition}

\begin{lemma}\label[lemma]{lem:distr-lax}
In \Funct{C}, the sets \(\distributes p\) and
  \(\codistributes p\) (by definition witnessing laxity and
  op-laxity) are in bijection with (natural) distributivity homsets for modalities:
  \begin{align*}
    {0 ≼ p · 0} &↔  \unitObj \hom p \unitObj\\
    ∀ q r.\, {p·q + p·r ≼ p · (q + r)} &↔  ∀ A B.\, p A \tensorObj p B \hom p {(A \tensorObj B)}
      & \text{\agdaref[Distributive implies lax]{Sem}[Distr-Lax]{isLax} \agdaref[lax implies Distributive]{Sem}[Distr-Lax]{isDistr}} \\
    {p · 0 ≼ 0} &↔  p \unitObj \hom \unitObj \\
    ∀ q r.\, {p · (q + r) ≼ p·q + p·r} &↔  ∀ A B.\, p (A \tensorObj B) \hom p A \tensorObj p B
      & \text{\agdaref[Distributive implies oplax]{Sem}[Distr-Lax]{isOplax} \agdaref[oplax implies Distributive]{Sem}[Distr-Lax]{isDistr'}}
  \end{align*}
\end{lemma}
The above lemma is of practical interest to the \Lamp{} programmer,
because one can characterise \(\distributes p\) and
\(\codistributes p\) syntactically, in terms of the modality
structure, rather than semantically.

We have all pieces in place and thus can proceed with the
interpretation of \Lamp.

\begin{theorem}\label{thm:semantics1}
  For any \nameref{def:bicat} \category{C}
   and a map of base types \(\basetype \in \TyConst\) into objects of \category{C},
  there is an interpretation \(⟦\_⟧\) of
  \Lamp[{\Funct C}] into
  \(\category C\).\agdaref[Soundness of typing]{Natural.Pats.Sem1}{eval}
\end{theorem}
\begin{proof}
  The product \(\tensorObj\) is used to interpret \(\Tensor\) types,
  but also context concatenation.
  As usual, every type \(A\) is
  interpreted as an object \(⟦A⟧\) and every context \(Γ\) is
  interpreted as an object \(⟦Γ⟧\), relying on the monoidal structure
  to interpret concatenation of contexts:
  \(⟦Γ,Δ⟧ = ⟦Γ⟧ \tensorObj ⟦Δ⟧\) and \(⟦\emptyContext⟧ =
  1\). Modalities being endofunctors, one can interpret qualified
  types as \(⟦\BOX p A⟧ = p ⟦A⟧\) and modalities in contexts as
  \(⟦x : \moda p A⟧ = p ⟦A⟧\). The judgement \(Γ ⊢ t : A\) is
  interpreted as a morphism \(⟦Γ⟧ \hom ⟦A⟧\), and a pattern \(\typat {Γ} m p A\) is interpreted as a morphism \(p ⟦A⟧ \hom ⟦Γ⟧\).
  The object \(A \funObj B\) is used to interpret
  function types; and \(\withObj\) product objects interpret \(\With\)
  product types.

  The choice of the functorial interpretation of the monoid \(0,+\)  is just
  the one needed to make the interpretation of \(\pthrow\) and
  \(\panalyse p l q m\). The pattern \(\pconv {φ} m\) is interpreted to
  \(⟦m⟧ ∘ φ\).

  The interpretation of \(\BOXintro\) relies on using the functorial
  action of \(p\) on the interpretation of the premise. Additionally,
  the interpretations of the
  hypotheses \((p⟦A₁⟧ \tensorObj ⋯ \tensorObj p⟦Aₙ⟧)\) must be grouped
  together as \(p (⟦A₁⟧ \tensorObj ⋯ \tensorObj ⟦Aₙ⟧)\). This can be done
  thanks to \(p\) being lax, \textit{via} the
  condition \(\distributes p\) when \(Γ\) isn't a singleton; otherwise no operation is required.

  Additionally, pattern matching a \(p\)-qualified pair (resp. unit)
  is possible under the condition that \(p\) is oplax, and this is
  given exactly by the side conditions.
\end{proof}
We can summarise the semantics as follows. First, the semantics of terms is given by:
\begin{align*}
  ⟦ x ⟧ & = id \\
  ⟦ \lam m t ⟧ & = \curry {⟦t⟧ ∘ (⟦m⟧ ⊗ id)} \\
  ⟦ t u ⟧ & = \appl ∘ (⟦t⟧ ⊗ ⟦u⟧) \\
  ⟦ () ⟧ & = id \\
  ⟦ (t,u) ⟧ & = ⟦t⟧ ⊗ ⟦u⟧ \\
  ⟦ \ttop{Γ} ⟧ & = ε \\
  ⟦ \twith t u ⟧ & = ⟦t⟧ ▴ ⟦u⟧ \\
  ⟦ \tout i t ⟧ & = πᵢ ∘ ⟦t⟧ \\
  ⟦ \TBOX p t ⟧ & = p ⟦t⟧ ∘ ψ_{\lvert Γ \rvert} & \text {Where \(Γ\) is the context of the term.} \\
  ⟦ \wlet {m=t} u ⟧ & = ⟦u⟧ ∘ ((⟦m⟧ ∘ ⟦t⟧) ⊗ id)
\end{align*}
The (silent) context re-orderings are interpreted using the monoidal structure, and
the semantics of patterns is the following:
\begin{align*}
  ⟦ x ⟧ & = id \\
  ⟦\pconv {φ} m ⟧ & = ⟦m⟧ ∘ φ \\
  ⟦\panalyse p l q m⟧ & = ⟦l⟧ ⊗ ⟦m⟧ \\
  ⟦\pthrow⟧ & = id \\
  ⟦(l,m)⟧ & = (⟦l⟧ ⊗ ⟦m⟧) ∘ \bar {ψ}_2 \\
  ⟦()⟧ & = \bar {ψ}_0 \\
  ⟦\TBOX p m⟧ & = ⟦m⟧
\end{align*}
In the above if we let \(ψ₀ : \unitObj \hom F \unitObj\) and
\(ψ₂ : F A \tensorObj F B \hom F (A \tensorObj B)\) be the morphisms
witnessing the laxity of functors, and in general \(ψ_n\) collects
\(n\) objects under the functor. Then \(ψ₁\) is the identity. The
notation \(\bar {ψ}_n\) refers to oplaxity witnesses.

The semantics of modalities in terms of endofunctors is quite general:
one can pick any \nameref{def:bicat}, and then any endofunctor is available as a
modality. This is useful, but this level of generality restricts the
(meta)theoretical results that one can obtain about the
calculus. Therefore, most actual applications do not to expose
arbitrary functors as modalities, but only a substructure. Yet, our guiding
idea is that \emph{every useful modality structure can be mapped to
  the functorial one}. This idea is illustrated throughout the paper with a number of examples.
To formalise it, we first need the notion of
homomorphism between modality structures.

\begin{definition}[modality functor]\label[definition]{def:mod-interp}\agdaref[modality functor]{HomoSem}{ModHomo}
  A modality functor \(\mathbf h : \structure F → \structure G\) is a
  functor between the underlying category structures of
  \(\structure F\) and \(\structure G\) which
  \begin{itemize}
  \item weakly respects the \((+,0)\) monoid (oplax functor wrt. \((+,0)\))
  \item strongly respects the \((·,1)\) monoid (strong monoidal functor wrt. \((·,1)\))
  \item respects the \(\distributive\) and \(\codistributive\) families
  \end{itemize}
This implies that we have the following natural morphisms, natural isomorphisms and functions:
\[
 \begin{array}[t]{r@{\;}l}
      \mathbf h (p + q) &≼ \mathbf h (p) + \mathbf h (q) \\
      \mathbf h (0) &≼ 0 \\
    \end{array}
    \hspace{1cm}
    \begin{array}[t]{r@{\;}l}
      \mathbf h (p · q) &≅ \mathbf h (p) · \mathbf h (q) \\
      \mathbf h (1) &≅ 1 \\
    \end{array}
    \hspace{1cm}
    \begin{array}[t]{r@{\;}l}
      \distributes p & → \distributes {\mathbf h(p)} \\
      \codistributes p & → \codistributes {\mathbf h(p)}
    \end{array}
\]  Note that the operations and conditions on the right-hand sides refer to the modality structure {\structure G}.


\end{definition}

\begin{remark}\label{thm:syntax-mapping}\agdaref[Language morphism]{Natural.Pats.Map}{}
  Any \nameref{def:mod-interp} \(\mathbf h : \structure F → \structure G\) gives
  rise to a map from \Lamp[{F}] to
  \Lamp[{G}]. This map applies to
  types, contexts, terms and patterns; and preserves typing.
\end{remark}

\begin{corollary}[categorical semantics]\agdaref[Natural semantics]{Natural.Pats.Sem}{eval}\label{thm:semantics}
  Given a \nameref{def:mod-interp} \(\structure F → \Funct C\) and a map \(g : \TyConst → \mathsf{Obj}(\category C)\)
  we can interpret \Lamp[F] in a \nameref{def:bicat} \(\category C\).
\end{corollary}
\begin{proof}
  First we apply the modality functor on the terms of the calculus. Then we interpret terms using \cref{thm:semantics1}.
\end{proof}
It can be enlightening to consider what a homomorphism
\(\structure F → \Funct C\) entails. Note that there are two levels of functors here: endofunctors in \category C, and a functor from \structure F to those. So, we have a map \(h\) from
\(\modSet F\) to \category {C}-endofunctors, which respects
several conditions. First, we must have the following natural
transformations:
\[  \begin{array}[t]{r@{\;}l}
    h(0) &⟶ \KObj(\unitObj) \\
    h(p+q) &⟶ h(p) \tensorObj h(q) \\
\end{array}
\hspace{2cm}
  \begin{array}[t]{r@{\;}l}
    h (1) &⟷ id \\
    h (p · q) &⟷ h (p) ∘ h (q) \\
  \end{array}
\]
Second, from \(\distributes p\) (resp. \(\codistributes p\)) one must be able to construct a lax (resp. oplax) structure for \(h(p)\)
The rest of the paper
provides several applications, in the form of specialised
modality structures, and their morphisms to \(\Funct C\).

\subsection{Example: Forgetful Modality Functor}\label{ex:forget}
\agdacomment{Subsection: Forgetful Modality Functor}
\agdacomment{-------------------------------------------}
\agdablank

A simple but important \nameref{def:mod-interp} is the one which forgets
all modalities.
Consider the trivial modality structure \(\structure I\), given by chosing
\(M := \braces{◆}\) (a set with a single element, implying
\(0 = 1 = ◆\)). 
Then, \Lamp[L] is a usual lambda calculus with cartesian product types. In this
situation, the \nameref{thm:semantics} reduces to the usual
semantics of the lambda calculus in terms of cartesian-closed
categories.

\begin{theorem}[Semantics of {\Lamp[I]}]\agdaref[Forgetful modality functor]{Sem.Forgetful}{forgetful}
  If\/ \(\category C\) has cartesian products, there is a \nameref{def:mod-interp}
  \(\structure I → \Funct C\), given by mapping \(◆\) to the identity
  functor (\(h(◆) = id\)).  The additive laws (\(h (p + q) ≼ h (p) + h (q)\) and \(h (0) ≼ 0\))
   are satisfied, thanks to an
  argument similar to \cref{thm:cartesian-informational}.  The rest of the
  conditions demanded by \cref{def:mod-interp} are trivially
  satisfied.
\end{theorem}

\begin{remark} 
  For every \structure{F}, there is a homomorphism
  \(\structure F → \structure I\). Consequently, one can erase modalities
  from any \Lamp{} instance, and obtain a usual lambda calculus.
\end{remark}

\subsection{Application: Arrays}
\agdacomment{Subsection: Application: Arrays}
\agdacomment{-------------------------------}
\agdablank
A first non-trivial application is arrays with known sizes (natural numbers). We
let modalities be the sizes of arrays, such that \(\BOX p A\) is interpreted as an array of A of size \(p\).  In this application the \ref{principle:product},
\(\BOX p {(\BOX q A)} \funOpBidir \BOX {p·q} A\) corresponds to flattening and
unflattening of arrays.
\(\BOX {p+q} A \funOp (\BOX p A \Tensor \BOX q A)\) corresponds to array split.
In general, bare \Lamp{} is not sufficient to program all useful array programs. Primitives such as \(sum : \BOX{\mathit{p}}\mathit{ℕ}⊸\mathit{ℕ}\) or
concatenation are also needed. (For this last case, see also \cref{sec:merge}.)
However, \BOXintro{} and the corresponding pattern lets us map, zip and unzip arrays conveniently.
Here is for instance the dot product of two \(ℕ\) arrays of length \(p\), given \(mul : ℕ ⊸ ℕ ⊸ ℕ\):
\begin{align*}
  dot &: \BOX{\mathit{p}}\mathit{ℕ}⊸\BOX{\mathit{p}}\mathit{ℕ}⊸\mathit{ℕ} \\
  dot &= \lam{\TBOX{\mathit{p}}{\mathit{x}}}{\lam{\TBOX{\mathit{p}}{\mathit{y}}}{\mathit{sum}\app{\TBOX{\mathit{p}}{\mathit{mul}\app{\mathit{x}}\app{\mathit{y}}}}}}
\end{align*}

\paragraph{Semantics}
Thanks to the groundwork layed out above, to provide semantics we only
need to specify the modality structure \(\structure N\) and a \nameref{def:mod-interp} \(\structure N → \Funct C\).
\begin{definition}[\structure N]\agdaref[Array modalities]{Sem.Arrays}{modalities}
  We can enthuse the natural numbers with a modality structure by
  leveraging the usual the additive and multiplicative structures. The
  convertibility relation \((≼)\) is just equality. We assert that all
  arrays sizes can be (un)zipped by letting \(\distributes p\), \(\codistributes p\) be a unit set.
\[  \begin{array}[t]{c@{\hspace{2cm}}c@{\hspace{2cm}}c}
    \modSet N ≔ ℕ &
    p ≼ q ≔ p =_ℕ q &
    \distributes p, \codistributes p := \braces {◆}
  \end{array}
\]
\end{definition}
\begin{definition}\agdaref[Array definition]{Sem.Arrays}{copies}
  Let \(\mathbf{a}(n)\) be an endofunctor which maps \(A\) to the \(n\)-ary product of \(A\):
  \[\mathbf{a}(n)(A) = \underbrace{A \tensorObj ⋯ \tensorObj A}_{\text{\(n\) times}}\]
\end{definition}
\begin{theorem}
   \(\mathbf {a}\) is a \nameref{def:mod-interp} \(\structure N → \Funct C\).
   \agdaref[Array modality functor]{Sem.Arrays}{homo}
 \end{theorem}
 \begin{proof}
   All the cases (in particular, that \(\mathbf{a}(n)\) is both lax
   and oplax) are either immediate or can be proved by induction on
   \(n\).
 \end{proof}
 We stress that in \Lamp, copies of data must be made explicit.  In
 this instance, this means that the management of the memory
 associated with arrays is explicit.
 This instance of \Lamp~
 does not keep track of structured array shapes (for ``tensors'').
 To support such shapes, one could make modalities
 be expressions of a language with explicit sums and
 products (\(\times\)).
 Operations like transpositions would then be
 provided as modality morphisms, such as
 \(p × q ≼ q × p\).

\subsection{Application: Sensitivity Analysis}
\agdacomment{Subsection: Application: Sensitivity Analysis}
\agdacomment{---------------------------------------------}
\agdablank

Another application is differential privacy, where one is interested
in publishing statistically anonymised data without revealing
individual secrets. For this purpose, \citet{reed_distance_2010}
present a system which ensures that if a certain amount of noise is
introduced in the inputs of a program, then at least the same amount
is present in the outputs.
In their system, every type~$A$ comes with a \emph{metric}
$d_A : A × A → \Rinf$, where $\Rinf$ denotes the set of non-negative
reals augmented with positive infinity.
Then a function $f$ from $A$ to $B$ is defined to be $c$-sensitive if
it does not increase distances by a factor greater than $c$ as
defined by the metrics for $A$ and $B$:
$d_B(f(x),f(y)) ≤_{ℝ} c · d_A(x,y)$.  Consequently, under the
assumption that distance \(0\) means equality at all types,
$0$-sensitive functions are necessarily constant.  Conversely
$∞$-sensitive functions impose no restriction on the argument.
Because of the inequality, $c$-sensitivity is subject to subsumption:
if $c' ≥_{ℝ} c$ and $f$ is $c$-sensitive then $f$ is also
$c'$-sensitive.

We can frame this system as an instance of \Lamp, and by making
\(c\)-sensitivity an appropriate family of endofunctors. We first
define the underlying category, following
\citet{amorim_semantic_2017}.

\begin{definition}\agdaref[category of metric spaces]{Sem.Sensitivity}{cat}
  Let \(\category S\) be the \nameref{def:bicat} of metric spaces and \(1\)-sensitive functions.
  To be precise, each object \(A\) is a carrier set \(\carrier{A}\) associated with a
  metric \(d_A : \carrier A × \carrier A → \Rinf \) and
  the homset \(A \hom B\) is comprised of functions \(f : \carrier{A} → \carrier{B}\) together with a guarantee that \(d_B (f (x),f (y)) ≤ d_A (x,y)\) for all \(x, y \in \carrier{A}\).
  The main non-trivial choice is that of which metric to assign to function spaces.
  Following \citeauthor{reed_distance_2010},
  we let the distance between two functions
  \(f₁, f₂ : A → B\) be
  \(\max \braces{d_B(f_1(x),f_2(x)) ∣ x ∈ \carrier A}\).
  The carrier set of both \(A \tensorObj B\) and \(A \withObj B\) is
  the cartesian product of \(\carrier A\) and \(\carrier B\), and the associated distances are
  \begin{align*}
    d_{A \tensorObj B} (x₁,y₁) (x₂,y₂) & = d_A(x₁,x₂) + d_B(y₁,y₂) \\
    d_{A \withObj B} (x₁,y₁) (x₂,y₂) & = \max (d_A(x₁,x₂), d_B(y₁,y₂))
  \end{align*}
  For \(\withObj\) it is adequate to use the maximum (rather than the
  sum) because only one of the components of the pair can ever be
  used.
\end{definition}

\begin{definition}
  Let \(\structure S\) be the modality structure with
  $\modSet S ≔ \Rpos$, with the usual arithmetic operations. The modality structure is a preorder rather than a general category and the preorder relation
  is the opposite of the usual order on ℝ: $(≼) = (≥_{ℝ})$.  The sets
  \(\distributes p\) and \(\codistributes p\) are trivial, meaning that one
  can analyse \(\Tensor\) products and \(\BOX p\) regardless of modalities in the context.
\end{definition}

\begin{definition}[magnifying functor]
  \agdaref[Magnifying functor]{Sem.Sensitivity}{scale}
  Let \(\mathbf s(p)\) be a family of
  endofunctors over \(\category S\) which leave the underlying set
  unchanged and multiplies the metric by the factor \(p\),
  so \(d_{\mathbf s(p)(A)}(x,y) = p\cdot d_A(x,y)\).
\end{definition}
Note that \(\mathbf s(0)(A)\) ``squashes'' \(A\) by identifying all of its values.
\begin{theorem}
  \(\mathbf s\) is a \nameref{def:mod-interp} \(\structure S \to \Funct S\).
  \agdaref[Scaling is a modality functor]{Sem.Sensitivity}{homo}
\end{theorem}

Then \Lamp[S]{} is a calculus for sensitivity.
\Citeauthor{reed_distance_2010} use logical relations to show that their
functions satisfy the intended semantics, but we do not need this
additionally machinery for \Lamp[S].
Instead, the standard \nameref{thm:semantics} provides us with the witnesses that
any \Lamp[S] function of type \(\fun {\BOX p A} B\) inhabits the
interpretation of its type, which in turns guarantees
\(p\)-sensitivity.

In \Lamp[S], the principle \(\BOX {0} A \funOp 1\) means that
values of squashed types can be ignored, while \(\BOX {p+q} A \funOp \BOX p A \Tensor \BOX q A\)
lets one spread sensitivity across several copies of a
variable. The multiplicative and functorial principles are
unsurprising, but \ref{principle:convertibility} gives that whenever
\(p ≥_ℝ q\), we have \(\BOX p A \funOp \BOX q A\). This behaviour is
adequate to capture sensitivity as intended.



In the system of \citeauthor{reed_distance_2010}, there are native
coproducts, and one can analyse type variables of arbitrary non-zero
sensitivity levels, essentially sustaining the principle
\begin{equation*}
  \fun  {\BOX p (\Plus A B)}
  \fun {((\fun {\BOX p A} α) \With (\fun {\BOX p B} α))} α
\end{equation*}
assuming our encoding of coproducts, this is simply
\begin{equation}
  \fun  {\BOX p (\Plus A B)}
  {(\Plus {\BOX p A} {\BOX p B})} \label{eq:liberal-case-analysis}
\end{equation}

In general, \Lamp does \emph{not} admit this principle. However, sensitivity analysis is special,
because non-zero magnification commutes with object constructors:
\begin{lemma}
  Let \(p = s(x)\) for some \(x ≠ 0\). Then:
  \begin{enumerate}
  \item \(pA \withObj pB \homBidir p(A \withObj B) \)
    \agdaref[Scaling records]{Sem.Sensitivity}{scale-with}
  \item \(p(A \funObj B) \homBidir pA \funObj pB\)
    \agdaref[Scaling functions]{Sem.Sensitivity}{scale-fun}
  \end{enumerate}
\end{lemma}
\begin{corollary}
  Let \(p = \mathbf s(x)\) for some \(x ≠ 0\).
  Then we have a morphism \(p(\Plus[X] A B) \hom \Plus [pX] {pA} {pB}\).
\end{corollary}
Because we have not integrated polymorphism in \Lamp{}, we cannot
derive (\ref{eq:liberal-case-analysis}), but we can approximate it as in
the above corollary. If we had polymorphism, we could see that \(p\) also commutes
with universal quantification, and obtain
\[
       p(∀α. ((A \withObj B) \funObj α) \funObj α)
  \hom ∀α. ((pA \withObj pB) \funObj pα) \funObj pα.
\]
Because \(x ≠ 0\), there is an inverse magnifying functor \(q = \mathbf s(x^{-1})\),
and substituting \(α\) by \(qα\) provides a semantic justification
for (\ref{eq:liberal-case-analysis}) in the case of sensitivity analysis.

\subsection{Monadic Effects}\label{sec:monad}
\agdacomment{Subsection: Monadic Effects}
\agdacomment{---------------------------}
\agdablank
Because \Lamp[\Funct{C}] supports arbitrary endofunctors as
modalities, it can also support monads as such, and consequently it
can be used as a calculus for effects. In this section we consider how
using it in this way plays out. The monadic structure can be
captured directly in terms of conditions imposed on modalities:
\begin{definition}\label[definition]{def:monadic-modality}~
  \begin{itemize}
  \item A modality \(p\) is a monoid iff we have \(η : 1 ≼ p\) and \(μ : p·p ≼ p\) with the usual laws.
    \item A modality \(p\) is an effect if it is a monoid and we have \distributes p.
  \end{itemize}
\end{definition}
In usual functional programming languages, such monads (monoids in the
category of endofunctors) are strong, and thus are
also automatically lax monoidal functors---aka.\ applicative functors
or idioms~\citep{mcbride_applicative_2007}.  This lets one sequence
effects.

In our setting of monoidal categories, monads are rarely strong. This is why we require
\(\distributes p\) for any effect modalities, i.e., demand them to be lax functors at
the semantic level. (Laxity is just what is needed to make the monads strong.~\cite{kock1972strong}) There are some important monads in non-cartesian
categories. For instance, linear state \(S \hom (S \tensorObj -)\) is
a monad in every monoidal category, but not a strong one.\agdaref[State monad]{Cats}[BiCat.Endo]{state}

This sets up the familiar monadic situation: qualification by \(p\) means that \(p\)-effects are
associated with a value. As usual, to use a \(p\)-qualified
variable we must account for the effects associated with it, typically
by forwarding them to another \(p\)-qualified value; and this is
ensured by \Lamp{} rules.

Sequencing effects is achieved syntactically by \(\BOXintro\).
Its semantics is given by combining effects in the context,
using the fact that \(p\) interprets to a lax functor; thanks to the side condition.  Indeed the
associated morphism \(pA \tensorObj pB \hom p(A\tensorObj B)\)
sequences the effects associated with \(A\) and \(B\) into effects
associated with \(A\tensorObj B\).  Thus the \BOXintro{} rule is
similar in spirit to the idiom brackets of
\citet{mcbride_applicative_2007}, but the ordering of effects is not
determined by the box contents but rather the order of variables in
the context \(Γ\). While unusual, this strategy will not cause surprises to the programmer, because the semantics only use
\(A \tensorObj B \hom B \tensorObj A\) when the program explicitly
demands variable reordering. At a syntactic level, a substitution
procedure would never reorder variables in the
context.

\begin{remark}\label{rk:effects-are-relevant}
  Here modalities must form a proper category rather than just a
  poset, since \(p ≼ q\) contains computational content.  For
  instance, we need to distinguish the two principal ways to form the
  natural transformation \( p∘p∘p ⟶ p∘p \), which are $\mu$ and
  $p(\mu)$.  In particular for the list monad, where \(\mu\) is the
  list concatenation, those produce different results: the first outer
  concatenation, and the second inner concatenation.
\end{remark}

While it is clear that we have all the components necessary to
properly support monadic effects, how convenient it is to use
depends on the actual surface syntax. With just \Lamp, the
monadic bind is noisy, and explicitly references \(η\) and \(μ\):\agdaref[Monadic bind]{Natural.Pats.Examples}{bind}
\begin{align*}
  bind & : \BOX p A \funOp (A \funOp \BOX p B) \funOp \BOX p B \\
  bind & = \lam{\TBOX{\mathit{p}}{\mathit{x}}}{\lam{(\pconv{\mathit{η}}{\mathit{f}})}{\tlet{\TBOX{\mathit{p}}{\TBOX{\mathit{p}}{\pconv{\mathit{μ}}{\mathit{y}}}}}{\TBOX{\mathit{p}}{\mathit{f}\app{\mathit{x}}}}{\TBOX{\mathit{p}}{\mathit{y}}}}}
\end{align*}
The existence of \(bind\) means that \Lamp can be extended with a Haskell-like
do-notation. But one can also consider programming using \Lamp
syntax directly. In first approximation, one would replace each monadic bind
by a \(\ttlet\) binding like the above. For instance:
\[
\begin{array}{ll}
  \mathsf{do}
  & x₁ ← e₁ \\
  & x₂ ← e₂ \\
  & \mathsf{return} e₃
\end{array}
\quad::=\quad
\begin{array}{ll}
  & \tlet {\TBOX p {\TBOX p {\pconv {μ} {x₁}}}} {\TBOX p {e₁}}\\
  & \tlet {\TBOX p {\TBOX p {\pconv {μ} {x₂}}}} {\TBOX p {e₂}} \\
  & \TBOX p {e₃}
\end{array}
\]

\providecommand{\tget}{\ensuremath{\mathit{get}}}
\providecommand{\tput}{\ensuremath{\mathit{put}}}
\providecommand{\tinc}{\ensuremath{\mathit{inc}}}

However, there is a subtlety when effectful functions return the unit type.
For example, let us consider $p = ℕ → (ℕ × {-})$, the monad threading a linear state of type $ℕ$
and the constants $\tget : \BOX p ℕ$ and $\tput : ℕ → \BOX p 1$,
using increment $\tinc : ℕ → ℕ$ and constant $42 : ℕ$.
\agdaref[State monad: get-put example]{Natural.Pats.Examples}{get-put}
\[\tlet{\TBOX{\mathit{p}}{\mathit{y}}}{\mathit{get}}{\tlet{\TBOX{\mathit{p}}{\TBOX{\mathit{p}}{\pconv{\mathit{μ}}{\mathit{z}}}}}{\TBOX{\mathit{p}}{\mathit{put}\app{(\mathit{inc}\app{\mathit{y}})}}}{\TBOX{\mathit{p}}{\tlet{()}{\mathit{z}}{42}}}}\]
In the above, one must retain the effects of \(\mathit{put}\) in the
variable $z : \moda {p} 1$.  It must then be explicitly consumed before
returning the final value. This ensures that no effect is lost. Note
that the typesystem would not accept consuming \(u\) outside of a
\(\TBOX p {·}\) bracket, because we don't have \(\codistributes p\). Indeed, given
\(\codistributes p\) we could drop any \(p\)-effect, which is
unsound in general.





In closing, note that \(+\) is not all that useful on effect
modalities: it represents that either effect can be present (but not
both). Yet, knowing that a given effect distributes on addition of
\emph{other} modalities still corresponds to the effect being a lax
functor.

\subsection{Dual Modalities: Adjoint Functors}
\agdacomment{Subsection: Dual Modalities}
\agdacomment{------------------------------}
\agdablank

A common pattern are dual modalities ---not a surprising fact given
that they correspond to the common adjoint functor pattern. In the presence of dual
modalities, one can present some key \Lamp{} rules in a way which is
better from a type-checking point of view.

\begin{definition}
  The modality \(\bar p\) is left-dual to \(p\) iff we have
  morphisms \(ε : \bar p · p ≼ 1\) and \(η : 1 ≼ p · \bar p\) such
  that \((ε ⊗ id) ∘ (id ⊗ η) = id = (id ⊗ ε) ∘ (η ⊗ id)\).
\end{definition}

The following theorem is reminiscent of Fitch-style modalities
\citep{clouston_fitch_2018}.
However, since our modalities are not all (co)distributive,
we only recover the adjoint with side conditions:
either the context $\Gamma$ is a singleton or the modality comes with the required
(co)distributivity property.

\begin{theorem}
If \(p\) has a left-dual, then the following \BOXin\ rule
is equivalent to {\BOXintro}
and has an inverse \BOXel:
\[
  {\BOXin}\inferrule
    {\bar p Γ ⊢ t : A}
    {Γ ⊢ \TBOXalt p t : \BOX p A}
    {\distributesMaybe{\Gamma}{p}}
\qquad
\qquad
  {\BOXel}\inferrule
    {Γ ⊢ t : \BOX p A }
    {\bar p Γ ⊢ \TUNBOX p t : A}
    {\distributesMaybe \Gamma {\bar p}}
\]
\end{theorem}
\begin{proof}
  Interderivability of \BOXintro[p] and \BOXin{} is mediated by whole
  context conversions. In general if \(φ : p ≼ q\) then we can apply
  the pattern \(\pconv {φ} x\) to each variable in \(x\) in the
  context. Hence the following rule, for whole-context conversion, is admissible:
  \[{\rul{conv}}\inferrule{φ : p ≼ q\\qΓ ⊢ t : A} {pΓ ⊢ \pconv {φ} t : A}\]

  Then we can define \(\TBOXalt p t\) and \(\TUNBOX p t\) as follows:
  \begin{align*}
   \TBOXalt p t & = \pconv{\mathit{η}}{\TBOX{\mathit{p}}{\mathit{t}}} & \text{\agdaref[Dual box introduction]{Natural.Pats.Examples}{box-}} \\
   \TUNBOX p t & = \tlet{\TBOX{\bar{p}}{\TBOX{\mathit{p}}{\pconv{\mathit{ε}}{\mathit{y}}}}}{\TBOX{\bar{p}}{\mathit{t}}}{\mathit{y}} & \text{\agdaref[Dual box elimination]{Natural.Pats.Examples}{invert-box-}}
  \end{align*}
The \BOXintro[p] rule can be recovered this way:
\( \TBOX p t = \TBOXalt{p}\app{(\pconv{\mathit{ε}}{\mathit{t}})}.\)\agdaref[Derived box introduction]{Natural.Pats.Examples}{box-'}
  The typing derivations are found in \cref{app:adjoints}.
\end{proof}
Assuming that we have \distributes{p}, then
the rule \BOXin{} is invertible and does not constrain the form of \(Γ\) when read backwards.
Thus, it in such a situation it is safely applicable
when deriving
\(Γ ⊢ \BOX p A\).
So, if a left-dual is available and \distributes{p},
\(\BOXin\) should be exposed to the users of the language.

Dual modalities are commonplace: every pair of adjoint
functors \(F ⊣ G\) can work as dual modalities. However, they are
not often both lax and oplax functors---and this fact is one of the
reasons why we have explicit modality classification through
\distributive\ and \codistributive.
Yet, when \(G\) is lax,
\(F\) is oplax,~\agdaref[Left adjoint is oplax]{Cats}[BiCat.Endo]{left-adjoint-oplax} and
\textit{vice-versa}.~\agdaref[Right adjoint is lax]{Cats}[BiCat.Endo]{right-adjoint-lax}

There are several ways to gain intuitive understanding for dual
modalities.
One understanding stems from the observation that the
\(\bar p\)-qualified variables from the context are inaccessible, so
they might just as well be discarded \emph{unless} their type is
\(p\)-qualified.  That is to say, multiplication of a
context by \(\bar p\) filters out variables which are not
\(p\)-qualified.

In programming, the paradigmatic example of adjoints is that of
currying and uncurrying: the functor \(\bar p = (X\tensorObj-)\) is
left-adjoint to \(p = (X \funObj -)\).~\agdaref[Tell is left adjoint to ask]{Cats}[BiCat.Endo]{tell-ask-adjoint}
The effect \(p\) corresponds to using
\(X\), and \(\bar p\) corresponds to providing \(X\).  Furthermore, if
\(X\) is a comonoid (which it will always be if \category C has
cartesian products),
then \(p\) is lax~\agdaref[Ask is lax]{Cats}[BiCat.Endo]{ask-lax}
and \(\bar p\) is oplax;~\agdaref[Tell is oplax]{Cats}[BiCat.Endo]{tell-oplax}
thus, they are made available to the rules which require \(\distributes p\)
and  \(\codistributes {\bar p}\).
%

In general, if \(p\) is associated to an effect, then \(\bar p\) is
associated to the corresponding co-effect.  Then, the \BOXel{} rule
means that, instead of dealing with effect \(p\), one can instead
demand a context which provides the coeffect \(\bar p\).

\subsection{Coeffects}\label{sec:comonad}
One can use \Lamp{} to work with comonadic co-effects instead of monadic effects. As expected, the definitions are dual:
\begin{definition}\label[definition]{def:comonadic-modality}~
  \begin{itemize}
  \item A modality \(\bar p\) is a comonoid iff it has \(ε : \bar p ≼ 1\) and \(δ : \bar p ≼ \bar p·\bar p\), satisfying the usual laws.
  \item A modality \({\bar p}\) is a co-effect if it is a comonoid and we have \(\codistributes {\bar p}\).
  \end{itemize}
\end{definition}

We can adapt our example from \cref{sec:monad} to
co-effects. Throughout this subsection, we assume \(\bar p\) to be
left adjoint to the state monad.
Then it is also a comonad and an oplax
functor.
\footnote{This left-adjoint is not the costate monad, and does not exist
  in most categories. To be able to interpret the whole program, it must
  to be (mechanically) translated back to monadic style.
  \Citet{hinze_monads_manuscript} provides a pedagogical exposition.}
Hence one can use \BOXel{} to
obtain the following co-primitives:
\(coget : \BOX{\bar{p}}1⊸\mathit{ℕ}\) and
\(coput : \BOX{\bar{p}}\mathit{ℕ}⊸1\)

In this style, values of type \(\BOX {\bar p} α\) act as a
coeffect source (or equivalently an effect sink). Each effectful
function must take one such argument.
Thus the program that increments the state and then
returns $42$ has type \(\BOX{\bar{p}}1⊸\mathit{ℕ}\) and becomes:
\[\lam{\TBOX{\bar{p}}{\pconv{\mathit{δ}}{\mathit{u}}}}{\tlet{()}{\mathit{coput}\app{\TBOX{\bar{p}}{\mathit{inc}\app{(\mathit{coget}\app{\TBOX{\bar{p}}{\mathit{u}}})}}}}{42}}\]
This construction is more direct than the monadic equivalent, suggesting that the \Lamp{} syntax is better suited to this style. Effects from several functions can be sequenced by duplicating effect sinks, which is possible because \(\codistributes {\bar p}\) (costate is oplax):
\begin{align*}
  \lam{\TBOX{\bar{p}}{\mathit{x}}}{\text{\ref{principle:oplax}}\app{\TBOX{\bar{p}}{\tlet{()}{\mathit{x}}{((),())}}}} & : \BOX{\bar{p}}1⊸\BOX{\bar{p}}1⊗\BOX{\bar{p}}1
\end{align*}
The effects will be sequenced in the left-to-right order of the
tuple. This kind of control is afforded by linearity: without it,
effect sinks could be duplicated structurally, and the sequencing order
would have to be specified by \textit{ad-hoc} syntactic means.

\agdacomment{Section: Resource Semantics}
\agdacomment{---------------------------}
\agdablank
\providecommand\Needs[2]{#2 ⊩ #1}
\providecommand\Inside[2]{#2 ∈ #1}
\providecommand\suchThat[2]{\braces { #1 ∣ #2 }}

\section{Resource Semantics (Kripke Model and Phase Semantics)}\label{sec:forgetful-model}

In this subsection, we define a semantics for \Lamp which predicts when a given resource
produces (values of) a given type, as an instance of the \nameref{thm:semantics}. This semantics combines aspects of the
phase semantics of linear logic~\citep{girard_linear_1987},
and aspects of intuitionistic Kripke semantics.
The model is parameterised over an arbitrary structure of \emph{resources}.
What we call a resource corresponds to a \emph{world} in the Kripke
semantics terminology, and a \emph{phase} in the linear logic
literature.
\begin{definition}[resource structure]\label[definition]{def:resource-frame}
  \agdaref[Resource structure]{Sem.Kripke.Worlds}{Monoid} A resource structure is a set
  \(W\) equipped with a preorder \(≤\) and a commutative monoid \((0,+)\) which
  respects the preorder.
\end{definition}
The motivating reading is that the resource 0 corresponds to no
resource at all; (+) corresponds to resource concatenation; and \(w' ≤ w\)
means that \(w'\) can be transformed to \(w\) without leftover (so
both insufficient or excessive resources may prevent \(\Needs A w\)
from holding). The preorder can also sustain \(w ≤ 0\) to model that \(w\)
can be discarded, or \(w ≤ w+w \) to model that \(w\) can be
duplicated.

We then proceed and define a relation between resources and types:
\(\Needs A w\).  In plain Kripke semantics \(\Needs A w\) means that
\(A\) holds at world \(w\), but for us \(\Needs A w\) denotes that a
value of type \(A\) can be produced from the resource \(w\).

\newcommand\downset[1]{\category {\check{#1}}  }
Instead of defining the relation \(\Needs A w\) directly we will get
it as an instance of the \nameref{thm:semantics} for types, thanks to
an adequate \nameref{def:bicat} structure, called hereafter {\downset W}. We let objects be sets
of resources, as is usually done in phase semantics, and define
\(\Needs A w ≔ w ∈ ⟦A⟧_{\downset W}\).  Additionally, resource sets will always be
closed downwards, meaning that if \(w' ≤ w\) then
\(\Needs A w → \Needs A {w'}\). This property
ensures that any valid resource transformation preserves the
\(( ⊩ A)\) predicate.

\newcommand\JJ[1]{#1 ≤ 0}
\newcommand\PP[3]{#1 ≤ #2 + #3}

\begin{definition}[category of resource sets, \(\downset W\)]
  \agdaref[Category of resource sets]{Sem.Kripke}{cat}\label[definition]{def:resource-set-cat}
  Downward-closed resource sets form a \nameref{def:bicat}.  A morphism \(A \hom_{\downset W} B\) witnesses inclusion \(A ⊆ B\),
  or equivalently \(\Needs {(A \funObj B)} 0\).  Object formers are:
\begin{align}
  {\unitObj} & = \suchThat w {w ≤ 0} \label{eq:kripke-unit}\\
  {A \tensorObj B}  & = \suchThat w {∃\Inside A {w_a}, ∃\Inside B {w_b}. ~~ w ≤ w_a + w_b} \label{eq:kripke-product}\\
  {\topObj} & = W \label{eq:kripke-garbage} \\
  {A \withObj B} & = A ∩ B \label{eq:kripke-additive-product} \\
  {A \funObj B} & = \suchThat w {∀ \Inside A {w_a}. ~~ \Inside B  {w + w_a}} \label{eq:kripke-fun}
\end{align}
\end{definition}%
The unit object (\ref{eq:kripke-unit}) is constructed from nothing;
i.e. the resource \(0\), or any resource that can be transformed to \(0\).  The
product  \((A \tensorObj B)\) (\ref{eq:kripke-product}) breaks down the resources $w$ used
for the construction of the pair into resources $w_a$ for the first
component and $w_b$ for the second component.  This is expressed by
the condition $\PP w {w_a} {w_b}$.
The product \(A \withObj B\) can be constructed if both \(A\)
and \(B\) can be constructed. The resources associated with the
components (\(A\), \(B\)) are the same as for the whole, because one
can only extract either component, never both. The \(\topObj\) object can be
constructed from anything. Intuitively it acts like a garbage bin
where one can throw any resource.

The definition of the internal hom (\ref{eq:kripke-fun}) states that
a function can be constructed from resources $w$ if for any argument
that brings its own resources $w_a$ the function result can be
constructed with resources $w + w_a$. Read differently, the resources
for a function application is the composition of the resources for
both function and argument.
Functions stemming from closed terms do not need own resources, thus,
they start with $0$, but as a curried function is applied to its
arguments one after another, it accumulates the resources coming with
each argument to eventually construct a result from all the gathered
resources.

To represent modalities, we need a notion of resource
\emph{scaling}. In general, we want resources to form a module-like
structure with \emph{modalities as scalars}.
\begin{definition}[modality-resource module]\label[definition]{def:resource-module}\agdaref[Modality-resource module]{Sem.Kripke.Worlds}{ModuleFor}
  In addition to a \nameref{def:resource-frame}, a modality-resource module is given by
  \begin{itemize}
  \item An arbitrary modality structure \(\structure Q\),
  \item A resource-scaling operation \(((◂) : \modSet Q → W → W) \) such that the following laws hold:
  \begin{align}
    \text{if \(s ≼ s'\)  and \(w ≤ w'\)} & \text{ then \(s ◂ w ≤ s' ◂ w'\)} \label{law:module-mono} \\
    s ◂ (s' ◂ w)	& = (s · s') ◂ w \\
    1 ◂ w	& = w \\
    0	& = s ◂ 0 \\
    0 ◂ w	& ≤ 0 \\
    (s + s') ◂ w	& ≤ (s ◂ w) + (s' ◂ w) \\
    (s ◂ w) + (s ◂ w')	& ≤ s ◂ (w + w') & \text {provided \(\distributes s\)} \label{law:module-distribute} \\
    s ◂ (w + w')	& ≤ (s ◂ w) + (s ◂ w') & \text {provided \(\codistributes s\)} \label{law:module-codistribute}
  \end{align}
  \end{itemize}
\end{definition}
Then, the resource semantics of a modality \(p\) is given by the following homomorphism from
\(\structure Q\) to endofunctors in \(\downset W\).  The idea is
that, to be able to produce \(\BOX p A\) with resources \(w\), one divides
\(w\) in \(p\) equal parts such that each part produces \(A\).

\begin{definition}\label[definition]{def:kripke-mul-by}
  \agdaref[Resource scaling functor]{Sem.Kripke}{mulBy}
  For \(s ∈ \modSet Q\), we define the functor \(\mathbf q(s)\) acting
  on resource sets as follows:
  \[\mathbf q(s)(A) = \braces {w ∣ {w ≤ s ◂ {w'}} ~∧~ w' ∈ A}\]
  The action of \(\mathbf q(s)\) on morphisms maps \(\Needs {(A \funObj B)} w\) to \(\Needs {\mathbf q(s)(A) \funObj \mathbf q(s)(B)} {(s ◂ w)}\).
 \end{definition}

\begin{theorem}
  \(\mathbf q\) is a \nameref{def:mod-interp} \(\structure Q → \Funct {\downset W}\).
\end{theorem}
\begin{proof}
  The proof relies on the distributivity properties of \((◃)\).
  \agdaref[Resource scaling is modality functor]{Sem.Kripke}{homo}
\end{proof}
\begin{corollary}[soundness of resource semantics]\label{thm:soundness}
  \agdaref[Soundness of resource semantics]{Sem.Kripke.Properties}[Soundness]{sound}
  Assume \(g : \TyConst → \downset W\) a map from base types to the set of resources which produce it.
  Then for any modality structure \structure{Q}, and any term \(\emptyContext ⊢ t : A\) of \Lamp[Q],  we have \(\Needs A 0\).
\end{corollary}
\begin{proof}
  We use the \nameref{thm:semantics} on \(\downset W\) to interpret $t$, letting \(g\)
  interpret base types and \(\mathbf q\) interpret modalities,
  leveraging the fact that it is a \nameref{def:mod-interp}.
  We get \(\unitObj ⊆ ⟦A⟧_{\downset W}\), or equivalently \(∀w. w ⊩ \unitObj → w ⊩
  A\). We conclude by letting \(w\) be \(0\).
\end{proof}

\newcommand\mset{\ensuremath{\mathcal{M}}}

\paragraph{Example: multiset resources}
Resources can be generated freely from a set \(E\) of elementary
resources.  Then a resource $w$ is a multiset $w \in \mset(E)$,
addition $w + w'$ is the
multiset union and $0$ is empty multiset. Linearity is then modeled by
letting \((≤)\) be equality. To sustain affinity for a given elementary resource
\(e ∈ E\), one can add \(\braces e ≤ 0\).

Typically, modalities would be natural numbers, scaling each
elementary resource and preserving the \(≤\) relation. If one makes a (perhaps frugal) meal with 5
carrots (\(5 \mathsf{carrot} ≤ \mathsf{meal}\)), then one can make two meals with ten
carrots: \(10 \mathsf{carrot} ≤ 2 \mathsf{meal}\).

However, multisets are \(ℕ\)-linear combinations of elements of \(E\),
thus one can let modalities be \(ℕ\)-linear maps on such multisets
(\(W ⊸ W\)). Indeed, as scalars these maps satisfy the laws of a
\nameref{def:resource-module}.
Running with our vegetable garden example, one can define a
season-shifting modality \(Season\), as a linear operator which maps one seed to one
carrot.  \Cref{def:kripke-mul-by} effectively computes the inverse of such linear operators, so
\(\mathbf q(Season)(\mathsf{\braces {carrot}}) = \braces {\mathsf{seed}}\). Then one can make the same meal from just
seeds, but a season later: \(5 \mathsf{seed} ⊩ \BOX {Season} \mathsf{Meal}\).

(It is easy to see how such a setup can be applied to vector spaces
and linear transformations.)

\subsection{Informational modalities}\label{sec:informational-modality}

\agdacomment{Subsection: Informal modalities}
\agdacomment{-------------------------------}
\agdablank

Even though we have motivated our Kripke-style semantics by a resource
analogy, it also applies to resources which are both discardable and
duplicable. We begin by characterizing the situation in terms of modality conversions.

\begin{definition}[informational modality]\label[definition]{def:informational-modality}
  A modality \(p\) is:
  \begin{description}
  \item[discardable] iff we have \(p ≼ 0\).\label{mor:disc}
  \item[duplicable] iff we have \(p ≼ p+p\).
  \item[informational] iff it is both discardable and duplicable.
  \end{description}
\end{definition}
Thanks to the use of the \pthrow{} (resp.\ \(\panalyse p x p y\) ) pattern, one can discard (resp.
duplicate) a \(p\)-qualified variable \(x\) if \(p\) is discardable or
duplicable, and an informational variable can be used as many times as one
wants. Consequently \(p\) is not used to control \emph{how much}
variables are used, but rather \emph{in which context} they are used.


\begin{definition*}
  We say that the modality structure \structure{F}, and by extension
  \Lamp[F], is informational if every modality in \structure{F} is
  informational.
\end{definition*}

\begin{remark}\label{thm:cartesian-informational}
  If the monoids \(( \category C ,\withObj, \topObj)\) and
  \(( \category C , \tensorObj , \unitObj )\) coincide, then every endofunctor \(F\)
  in \Funct{C} is an informational modality in \Lamp[\Funct{C}], i.e.,
  discardable~\agdaref[Cartesian implies discardable]{Cats}[BiCat.Endo.Functor]{getε}
  and duplicable.~\agdaref[Cartesian implies duplicable]{Cats}[BiCat.Endo.Functor]{getδ}
  In other words, when constructing a \nameref{def:mod-interp}
  \(\structure{F} → \Funct{C}\), the \(p≼ 0\) and
  \(p ≼ p+p\) conversions in \structure{F} can always be
  realised, thanks to the cartesian structure of \(\category C\).
\end{remark}
\begin{remark}\label{thm:semi-lattice}
  In a informational system, the modality preorder \(≼\) is a semi-lattice, with \(0\) being the top \(⊤\) and \((+)\) the meet (\(∧\)).
  That is, for all \(p, q\),
  \begin{itemize}
    \item \(p ≼ 0\)
    \item \(p ≼q\) iff \(p = p + q\).
  \end{itemize}
\end{remark}

We can then define a \nameref{def:ringoid} from informational modalities.
\begin{definition}[semi-lattice resources]\label[definition]{def:lattice-worlds}
  \agdaref[Semi-lattice resources]{Sem.Kripke.Worlds}[Lattice]{monoid}
  \(
  \begin{array}[t]{r@{\;}l}
    W & := L  \\
    p + q & := p ∧ q  \\&
  \end{array}
  \hspace{1cm}
  \begin{array}[t]{r@{\;}l}
    p ≤ q & \mbox{ iff } p = p \wedge q \\
    0 & := ⊤\\
  \end{array}
  \)\end{definition} %
A consequence of this choice of resources is that the tensor and cartesian products coincide. Indeed,
 $ (A \tensorObj B)  =  \braces{ w ∣ ∃w_a∈ A, w_b ∈ B.  w ≤ w_a ∧ w_b}$
and because \(A\) and \(B\) are downwards closed,
\(A \tensorObj B  =  A ∩ B\).
With the same line of argumentation, we also find \(\unitObj = W\) and
\begin{align*}
  \Needs {A \funObj B} w  & \mbox{ iff } ∀u≤w. \Needs A u  → \Needs B u,
\end{align*}
thereby recovering the usual intuitionistic Kripke semantics


\begin{definition}[lattice modalities]\agdaref[Lattice modalities]{Sem.Kripke.Worlds}[Lattice]{modalities}\label[definition]{def:lattice-modalities}
  Given any lattice \((L, {∧}, ⊤, {∨}, ⊥)\), we can construct
  \(\structure L\), the structure of lattice modalities, as follows:
\[    \begin{array}{r@{\;}l}
      \modSet L     & := L  \\
      p + q & := p ∧ q  \\
      0     & := ⊤ \\
    \end{array}
\hspace{2cm}
    \begin{array}{r@{\;}l}
      p ≤ q & := p ≤ q \\
      p · q & := p ∨ q  \\
      1     & :=  ⊥ \\
    \end{array}
\]\end{definition}
Because \(L\) is a ring, it is clear that \(L\) is a module with
itself as scalars and (∨) as
operation.~\agdaref[Lattice modules]{Sem.Kripke.Worlds}[Lattice]{modul} As such, \(L\) is an
appropriate instance of a resource structure, and together with the associated modalities we have a \nameref{def:resource-module}.
Then, the main \Lamp{} principles become:
\begin{align*}
\BOX {⊤} A & → 1 & \text{\ref{principle:zero}} \\
\BOX {p∧q} A & → \BOX p A × \BOX q A & \text{\ref{principle:sum}} \\
\BOX {⊥} A & ↔ A  &\text{\ref{principle:unit}}  \\
\BOX p {(\BOX q A)} & ↔ \BOX {p ∨ q} A & \text{\ref{principle:product}}
\end{align*}
Namely, at resource \(p∧q\) one has access to both data available at
\(p\) and at \(q\); and at resource \(p∨q\) one only have access to data
which is accessible simultaneously at \(p\) and at \(q\).
As for every \Lamp{} instance, information flow is determined by
\((≼)\), according to \ref{principle:convertibility}.
The elements of the lattice can be enthused with various intuitions,
depending on the application. For instance, they can represent:
\begin{itemize}
\item Lamport clocks for distributed computing;
\item Agents, and \(\BOX p A\) means that agent \(p\) believes \(A\);
\item Security levels, and \(\BOX p A\) means that agent
  \(p\) can access \(A\).
\end{itemize}
We give a more detailed semantics for lattices with information flow in \cref{app:security}.

\section{Relational Semantics}\label{sec:relmodel}
\agdacomment{Section: Relational Semantics}
\agdacomment{-----------------------------}
\agdablank

\providecommand\Relation[2]{#1_{\R}^{#2}}
\providecommand\RelatedAt[4]{(#3,#4) ∈ \Relation{#1}{#2}}

While the Kripke-style semantics described in
\cref{sec:forgetful-model} is effective to determine if some
resource can produce some type, in programming languages applications
it is often useful to simultaneously establish relations between programs.

In this section we extend Kripke predicates to indexed logical
relations. This extension remains a instance of the generic category
semantics. The idea is to have two forgetful interpretations (as in
\cref{ex:forget}) in the category \category{Set}, and capture the
meaning of modalities by a resource-indexed relation between them.  Thus,
in this interpretation, an object is a triple of two sets \(A_1\),
\(A_2\), and a family of contravariant resource-indexed relations
\(\Relation A w\) in \(\Pot(A_1 × A_2)\).   More
explicitly, resources related by \((≤)\) yield (inversely) contained
relations: $w ≤ w' → \Relation A {w'} ⊆ \Relation A w$.

\begin{definition}[category of indexed relations]\label[definition]{def:rel-cat}
  \agdaref[Category of indexed relations]{Sem.IndexedRelations}{relcat}
Contravariant relations form a \nameref{def:bicat} \(\category R\).
The relational components of the object formers are defined as follows:
{\small
\begin{align}
  \RelatedAt {\unitObj} w {()} {()} & = {w ≤ 0} \label{eq:rel-unit}\\
  \RelatedAt {(A \tensorObj B)} w {(x₁,y₁)} {(x_2,y_2)}  & =  {∃w_a w_b. w ≤ (w_a + w_b) ∧  \RelatedAt A {w_a} {x_1} {x_2}  ∧  \RelatedAt B {w_b} {y_1} {y_2}} \label{eq:rel-product}\\
  \RelatedAt {\topObj} w {x₁} {x₂} & = ⊤ \label{eq:rel-garbage}\\
  \RelatedAt {(A \withObj B)} w  {(x₁,y₁)} {(x_2,y_2)} & = \RelatedAt A {w} {x_1} {x_2} ∧ \RelatedAt B {w} {y_1} {y_2} \label{eq:rel-additive-product}\\
  \RelatedAt {(A \funObj B)} w {f_1} {f_2} & = ∀ w_a x_1 x_2.  \RelatedAt A {w_a} {x_1} {x_2} → \RelatedAt B  {w + w_a}  {f_1 (x_1)} {f_2(x_2)} \label{eq:rel-fun}
\end{align}
}A morphism in the homset \(A \hom_{\category R} B\) is constituted of a
pair of functions \(f_1 : A_1 \hom_{\category {Set}} B_1\) and
\(f_2 : A_2 \hom_{\category {Set}} B_2\),
and a witness that they are related at resource \(0\):
\(\RelatedAt {(A \funObj B)} 0 {f_1} {f_2}\).
\end{definition}

The extension of the simple Kripke-style model to relations is oblivious to
the related elements.
%
We have asserted in \cref{sec:semantics} that modalities are
interpreted as endofunctors. Unfortunately, in the special case of
relations, \emph{arbitrary functors do not yield the desired models}.
Indeed, the index-wise product functor \(p \tensorObj q\) is problematic as an
interpretation of the modality addition \(p+q\). It comes down the
fact that the interpretation of the points must be pairs, but useful
relational models assume that when qualifying a type by \(p+q\), the
underlying values are not pairs (as they were in the arrays example),
but rather two aliases of the same object.

\providecommand\idendo{identity-endofunctor\xspace}
\providecommand\idendos{identity-endofunctors\xspace}
\providecommand\Idendos{Identity-endofunctors\xspace}
\providecommand\femb[1]{{#1}^{*}}

Consequently, we must construct a new modality structure
\structure{R}, restricted to relational endofunctors whose underlying
functors (in the \(\category{Set}\) category) are the identity
functors.~\agdaref[Modality functor]{Sem.IndexedRelations}{ModalityFunctor}
(Hereafter we refer to this kind of functors as \idendos.)
Such functors are
characterized by:
\begin{itemize}
\item A map \(F\) between contravariant relations. If the input relation
  family is noted \(\Relation R w\), then the constructed family is noted
  \( \Relation {F(R)} w\).
\item An action on morphims in \(\category R\).
\end{itemize}
The functorial action on these morphisms is a proof that functions
remain related under the \(F\) action on morphisms. Namely:
\(\Relation {(A \funObj B)} w ⊆ \Relation {(F(A) \funObj F(B))} w \).
\agdaref{Sem.IndexedRelations}{from⊆}
We proceed to define a new product of objects in \(\category R\)
whose purpose is to interpret modality addition:
\begin{align}
  \label{eq:intersection-unit}  \RelatedAt {\otop} w {x₁} {x_2}  & =  w ≤ 0
    & \text{\agdaref{Sem.IndexedRelations}{Top}} \\
  \label{eq:intersection-product} \RelatedAt {(A \oand B)} w {x₁} {x_2}  & =  ∃w_a w_b. w ≤ (w_a + w_b) ∧  \RelatedAt A {w_a} {x_1} {x_2}  ∧  \RelatedAt B {w_b} {x_1} {x_2}
    & \text{\agdaref{Sem.IndexedRelations}{_∧I_}}
\end{align}
In contrast to the \(\unitObj\) and \(\tensorObj\) relations formers, the \( \otop \) and
\(\oand\) relation formers are defined for \emph{any} underlying domain,
not just the product in \(\category {Set}\). Coincidentally, the
\emph{same} underlying values are related by both operands of
\(A \oand B\).  Defining pointwise \idendos{} for \((\withObj)\) is impossible,
but not a problem for (\( \oand \)).

\begin{definition}[\nameref{def:ringoid} of \idendos]\label[definition]{def:relation-modalities} This structure, called \structure{R}, is:\agdaref{Sem.IndexedRelations}{IdFunctorsMods}
\begin{align*}
  M     & := \category R → \category R & \text{\idendos} \\
  p ≼ q & := p ⟶ q & \text {natural transformations}\\
  p · q & := p ∘ q & \text {functor composition}\\
  1     & := \mathsf{id} &  \text{identity functor} \\
  p + q & := p \oand q & \text {pointwise application of \( \oand \)}\\
  0     & := \KObj(\otop)   \\
  \distributes p, \codistributes p & := \text{\(p\) is lax}, \text{\(p\) is oplax}
\end{align*}
\end{definition}
Note that the distributivity conditions continue to refer to (op) laxity wrt. to the \((\tensorObj,\unitObj)\) monoid, not \((\oand, \otop)\).
Indeed, \ref{principle:oplax} and \ref{principle:lax} are about splitting and joining data, whereas
(\(\oand\)) has the same data on both sides. Hence \cref{lem:distr-lax} does not apply for relational semantics. Regardless, \Lamp{} is usable on \idendos, because they embed into relational functors.
\begin{theorem}
  The embedding \(\structure R \hom \Funct R\) is a \nameref{def:mod-interp}.\agdaref{Sem.IndexedRelations}{idFunctSem}
\end{theorem}
\begin{proof}
  Most of the cases are immediate, but the structures differ for sums,
  and thus we have to show that we have natural transformations
  \(\KObj(\otop) ⟶ \KObj(\unitObj)\) and \(A \oand B ⟶ A \tensorObj B\). The underlying natural
  transformations are discard (\(ε\)) and duplication (\(δ\)), from the
  cartesian structure of \category{Set}. But in fact, at the
  relational level, there is by definition no difference between
  \(\KObj(\otop)\) and \(\KObj(\unitObj)\) (and respectively \(A \oand B\) and \(A \tensorObj B\))
  when the underlying transformations are \(ε\) (respectively \(δ\)).
\end{proof}
We are now done with our generic relational interpretation for
modalities in \(\structure R\).  Because modality functors
compose, it means that to instanciate this generic relational model
on any new modality structure \(\structure F\), we need only a \nameref{def:mod-interp}
\(\structure F → \structure R\). We show an example in \cref{app:security}.

With the \nameref{def:rel-cat}, the \nameref{thm:semantics} specialise to the
fundamental lemma of logical relations:
\begin{theorem}[Fundamental Lemma]\agdaref{Rel.Sem.Kripke.Properties}{fundamental-lemma}\label{thm:fundamental-lemma}
  Given \(Γ ⊢ t : A\), a resource \(w\) and two
  environments \(ρ=(x₁,\ldots,xₙ )\) and \(ρ'=(x'₁,\ldots,x'ₙ )\) related by
  \(\Relation {Γ} w\) we have
  \(\RelatedAt A w {⟦t⟧_{\category {Set}}(ρ)} {⟦t⟧_{\category
      {Set}}(ρ')}\).
\end{theorem}
To clarify, for every binding \(xᵢ : \moda p Aᵢ\) in context \(Γ\),
the components \(x,x'\) of the respective environments \(ρ,ρ'\) are related
when \(\RelatedAt {p{A_i}} w x {x'}\).

\subsection{Application: Information-Flow Security and Non-Interference}\label{app:security}
\agdacomment{Subsection: Information-Flow Security and Non-Interference}
\agdacomment{----------------------------------------------------------}
\agdablank

One application of modal type systems is to ensure that certain parts
of a program do not have access to private (high security)
information. Several type systems have been proposed to explicitly
support this feature, notably the seminal work of
\citet{abadi_core_1999}.

We have seen in \cref{sec:informational-modality} that lattice modalities
model a flow of information in
the \((≼)\) direction. In the information-flow community, this
property is expressed as non-interference, namely that the public
output of a program does not depend on secret inputs. To prove this
property for \Lamp{}, we need to instantiate the relational model
appropriately for informational modalities.

Before proceeding, we can spend some effort checking that
the \nameref{def:rel-cat} provides the right definition for our
application. Recalling that each resource $w:W$ stands for a security
level, the relation $\Relation A w$ will identify values of type \(A\)
that an agent of clearing level $w$ is not allowed to distinguish
(``see''). One extreme is the discrete relation that hides nothing and
allows one to distinguish everything (full information); the other
extreme is the full relation that identifies any two values and thus
hides everything (no information).

Because security resource-levels form a lattice, the product structures
(\(\withObj, \topObj\) and \(\tensorObj, \unitObj\)) coincide, and our
relational semantics reduces to the usual Kripke logical relations.

In particular, the empty tuple is indistinguishable from itself; and
accordingly the corresponding semantic relation may hold at all
security levels, as there is never a need to look into the empty
tuple, regardless the capabilities an agent is equipped with.
Likewise, the capability to access a pair should include the
capability to access both components.  Turning this statement around,
a pair is only indistinguishable from another pair if both respective
components are so. Because resource-levels form a lattice, $w_a + w_b$ in
\cref{eq:rel-product} is the meet $w_a ∧ w_b$, and the condition
$w ≤ w_a ∧ w_b$ breaks down into the pair of conditions $w ≤ w_a$ and
$w ≤ w_b$.


From the security perspective, access to a function and access to its
argument should be sufficient to get access to the result.  Thus in
\cref{eq:rel-fun} setting the resource of the result to ${w} + {w_a}$,
does the correct thing for access control to functions.

We now proceed to map \nameref{def:lattice-modalities} (\(\structure L\))
 to \idendos.  Following
\citet{algehed_simple_2019}, we
provide an \idendo{} which hides the information below a given level.

\begin{definition}[hiding functor]\agdaref{Sem.IndexedRelations.Security}{raise}\label[definition]{def:raise-functor}
  For each security level \(w_0\) we define the functor \(\mathbf s(w_0)\) whose action on contravariant
  relations is:
  \(\RelatedAt {\mathbf s(w_0) A} w {x_1} {x_2} := (w_0 ≤ w) ⟶ \RelatedAt A w
  {x_1} {x_2} \).
\end{definition}
We can see that if \(w < w_0\) the premise is false, and thus the whole
relation is trivial: \(x₁\) and \(x₂\) can no longer be distinguished,
or ``hidden''. Otherwise, \(w_0 ≤ w\) and the relation remains the
same. Taking the dual viewpoint, the higher the clearing level $w$,
the more one is allowed to see, thus, the fewer values become related
by the indistinguishability relations.  This matches the fact that
relations are contravariant: $w ≤ w'$ implies
$\Relation A {w'} \subseteq \Relation A {w}$. This observation is the key
justification for \(\mathbf s\) being a suitable \nameref{def:mod-interp}:
\begin{theorem} \(\mathbf s\) is a modality functor
  \(\structure L → \structure R\).\agdaref{Sem.IndexedRelations.Security}{secSpec}
\end{theorem}

We have now all the pieces to obtain non-interference for \Lamp[L],
which states formally that if a level \(l\) is more private than
\(l'\), then an input protected at \(l\) cannot influence an output at
\(l'\).
\begin{theorem}[non-interference]\label{thm:non-interference}\agdaref{Rel.Sem.Security.NonInterference}[NIWiProducts]{non-interference}
  If \(l \not ≼ l'\), and $f : \dent {\fun A \fun [l] B \BOX {l'} \Bool}$, and $⊢ u : A$, and
  $b_1,b_2 ∈ \dent B$ then $f {\dent u} b_1 = f {\dent u} b_2$.
\end{theorem}
\begin{proof}
  The proof is an application of the \nameref{thm:fundamental-lemma}
  at resource-level \(l'\), and essentially relies on the definition
  of the \nameref{def:mod-interp} \(\mathbf s\).  Because the inputs
  are at a higher level \(l\), they are not observable and thus always
  related. Because outputs are at the default level (that for which we
  use the \nameref{thm:fundamental-lemma}) the output relation is the
  standard one for booleans (simple equality).
\end{proof}

\section{Computationally-Tracking Resource Semantics (Presheaf Model)}\label{sec:computational-model}

The model of \cref{sec:forgetful-model} does not account for what
happens in the (co)effectful models, because it ignores the
computational content associated with resource conversion. In this
section we amend it to account for computational effects.

First, the resource structure must be a category rather than a set:
\begin{definition}[resource structure]\label[definition]{def:resource-frame-2}
  \agdaref[Resource structure]{Sem.Presheaf.Worlds}{Monoid}
  A resource structure is a category \(\category W\), with \(w ⇝ w'\) as hom-sets, equipped with the
  symmetric monoidal structure \((+,0)\).
\end{definition}
The homset \(w ⇝ w'\) is a proof-relevant version of the predicate \(w ≤ w'\) used in the plain Kripke semantics.
The laws required by the
\nameref{def:resource-module} structure are thus witnessed by morphisms in the category
of resources, which are also required to be natural.\agdaref[Modality-world module]{Sem.Presheaf.Worlds}{ModuleFor}

Second, the target category is one of presheaves rather than just
downward-closed sets. In fact a downward-closed set is exactly
a contravariant functor from the preordered set \(W\) to \category{Prop},
and a presheaf is a contravariant functor from the category
\(\category W\) to \category{Set}.  Thus for each type \(A\), we have a functor that
maps each resource \(w\) to a set \(\Needs A w\) of
\emph{values}. Each value is a witness that \(w\) produce \(A\). In particular, a value may contain explicit reference to resource transformations.
The presheaf action on a conversion \(f : w ⇝ w'\) turns any value \(u : \Needs A {w'}\) into
an element of \(\Needs A w\).
Following \citet{scott1980relating}, we write \(u↿f\) for this action.
As any functorial action, it must respect identity \(u ↿ \mathsf{id} = u\) and composition\((u↿(f∘g) = u↿f↿g)\).

\providecommand\ObjNeeds[2]{#2 \cobj{⊩} #1}

\begin{definition}[category of presheaves over resources]\label[definition]{def:world-presheaves-category}
  Presheaves over a resource category \category{W} form a \nameref{def:bicat}, hereafter called \(\category {\widehat W}\).
The action of object formers on resources are:
\begin{align*}
  \ObjNeeds {{\unitObj}} w & = w ⇝ 0 \\
  \ObjNeeds {({A \tensorObj B})} w  & = \int^{w_a, w_b} \ObjNeeds A {w_a} \quad\times\quad \ObjNeeds B {w_b} \quad\times\quad (w ⇝ w_a + w_b) \\
  \ObjNeeds {{\topObj}} w & = \top \\
  \ObjNeeds {({A \withObj B})} w & = \ObjNeeds A w \quad\times\quad \ObjNeeds B w \\
  \ObjNeeds {({A \funObj B})} w & = \prod_{w_a} \left(\ObjNeeds A {w_a} \quad\to\quad \ObjNeeds B {(w + w_a)}\right)
\end{align*}
\end{definition}

Most cases extend the non
computationally-tracking variant straightforwardly, but \({A \tensorObj B}\) deserves a
commentary. It is the Day convolution of the presheaves \(A\) and
\(B\). It is similar to the set
\(∃_{w_a, w_b} (\ObjNeeds A {w_a}) \times (\ObjNeeds B {w_b}) \times (w ⇝ w_a + w_b)\), but
with an additional naturality condition, expressed by using a coend
instead of an existential.  What is this additional condition? Consider
two triples \((u₁,v₁,h₁)\) and \((u₂,v₂,h₂)\) inhabiting the above
set.  We want to consider them equal when, for any conversions
\(f,g\) such that \(u₁↿f = u₂\) and \(v₁↿g=v₂\) we have \(h₁= (f⊗g)∘h₂\).
That is, it does not matter if we convert by \(f\) and \(g\)
separately in the components or by \(f⊗g\) as part of the resource
splitting.

In the \(\downset W\) category, morphisms are trivial, but not in \(\category {\widehat W}\). Hence we
must check that they satisfy the laws of \nameref{def:bicat}, and that
they respect naturality in resources: for every morphism \(h : \ObjNeeds A w → \ObjNeeds B w\) we have
\(h (u↿f) = h (u )↿f\). The proofs are sometimes tedious, but machine-checked.
  \agdaref[Category of presheaves over resources]{Sem.Presheaf}{cat}

Third, the resource-scaling functor family must be amended in a similar manner to the product.
\begin{definition}\label[definition]{def:presheaf-mul-by}
  \agdaref[Resource scaling functor]{Sem.Presheaf}{mulBy}
  For any modality \(s\) in \(\modSet Q\), \quad
  \(\ObjNeeds {\mathbf q(s)(A)} w = \int^{w'}  \ObjNeeds A {w'} ~×~ (w ⇝ s ◂ {w'})\).
   It is the left Kan extension of \(⟦A⟧_{\widehat W} \) along $s$.
\end{definition}

A payoff of this model is that
the \nameref{thm:soundness} now becomes an evaluator:
\begin{corollary}[evaluation]\label{thm:presheaf-soundness}
  \agdaref{Sem.Presheaf.Properties}[Soundness]{sound}
  Assume \(g : \TyConst → \widehat W\) a map from base types to resource presheaves.
  Then for any modality structure \structure{F}, we can evaluate a closed term \(\emptyContext ⊢ t : A\) of \Lamp[F],  to a value in \(\Needs A 0\).
\end{corollary}

\subsection{Example: Parallel Execution Time}\label{sec:parallel-timings}
\agdacomment{Example: Parallel Execution Time}
\agdacomment{--------------------------------}
\agdablank

Let us now considering \emph{time} as a resource.
We let the resource structure be the positive real numbers, and use \emph{maximum} as a monoid, yielding an \nameref{def:informational-modality} set:
\begin{definition}[time resources]
\(    \begin{array}[t]{r@{\;}l}
      W       & := ℝ^+ \\
      p ≤_W q & := p ≥_ℝ q
\end{array}
\hspace{1cm}
    \begin{array}[t]{r@{\;}l}
      0_W     & := 0_ℝ \\
      p +_W q & := \max(p,q)  \\
    \end{array}
\)\end{definition}
The model gives us \(\Needs A w\), the set of \(A\)-values that can be
produced in \(w\) units of time.  Even though the category {\category
  W} is a preorder in this case, and thus conversions have no useful content, the
presheaf model lets us associate arbitrary sets to base types. For
instance, assuming that we have the base type $b_{ℕ}$ of natural
numbers, we can let $w ⊩ b_{ℕ}$ be the subset of natural numbers which
can be generated in $w$ units of time or less.  For instance,
picturing a unary representation, we may have
$w ⊩ b_{ℕ} = \braces { n ∈ ℕ ∣ n ≤_{ℝ} w}$. For any \(u ∈ w ⊩ b_{ℕ}\) and any \(f:w'≤_W w\) we have \(u ↿f = u\). Indeed, the order is opposite to the usual
order on natural numbers, and \(f\) witnesses
\(w'≥_{ℝ}w\), and thus \(\Needs {b_{ℕ}} w ⊆ \Needs {b_{ℕ}} {w'}\). That is to
say, when getting more time, one can still produce (at least) the same
values.  In fact this inclusion is true for all types. In general, \(\Needs A 0\) corresponds to
values which are ready \emph{now}, and \(\Needs A w\) corresponds to
values which will be ready \(w\) units of time in the future.
\begin{remark}
  The category {\category W} is a semi-lattice (\cref{def:lattice-worlds}), and thus both products coincide:
  \begin{align}
      \ObjNeeds {\unitObj} w & = \top \\
      \ObjNeeds {({A \tensorObj B})} w & = \ObjNeeds A w \quad\times\quad \ObjNeeds B w \\
      \ObjNeeds {{A \funObj B}} w  & = \prod_{u≤_W w} ( \ObjNeeds A u \quad →\quad \ObjNeeds B u) \label{eq:presheaf-lattice-fun}
    \end{align}
\end{remark}
In timing terms, the pair \(A \tensorObj B\) is ready when both
components are ready; so with this setup the model corresponds to
maximal parallelism.  Additionally function application does not
take any time all (\ref{eq:presheaf-lattice-fun}): the result of a
function application is ready when both the function and the argument
are ready. In order ascribe time to primitive functions, we can annotate their
type with a modality, as explained below.

Modalities correspond to \emph{relative} timings, or delays, which we can choose
to be in \(ℝ⁺\). Their structure is borrowed from the so-called \emph{arctic
semiring}:
\begin{definition}[timing modalities]
  \(
    \begin{array}[t]{r@{\;}l}
      M     & := ℝ⁺  \\
      p + q & := \max(p,q)  \\
      0     & := 0_ℝ \\
    \end{array}
    \hspace{1cm}
    \begin{array}[t]{r@{\;}l}
      p ≤ q & := p ≥_ℝ q \\
      p · q & := p +_ℝ q \\
      1     & := 0_ℝ \\
    \end{array}
  \)\end{definition}
  \begin{remark}
    The above structure forms a ring and is an extension of the
    monoidal structure of time as resource, so it is clear that \((M, W, ·)\)
    forms a \nameref{def:resource-module}.
  \end{remark}

  \Cref{def:presheaf-mul-by} simplifies to:
  \begin{align*}
    \ObjNeeds {\mathbf q(p)(A)} w
    & = \int^{w'}  \ObjNeeds A {w'} ~×~ {w ≥_{ℝ} p + w'} & \text{by definition} \\
    & = \ObjNeeds A {(w - p)}  & \text{by closure of \(A\)}
  \end{align*}
  Hence, qualification by \(p\) adds an artificial waiting time of
  \(p\). Equivalently, it lets us consider the values
  that were ready \(p\)-time in the past. Thus \(\fun {\BOX p A} B\) is
  interpreted as a function converting an \(A\)-value that was
  ready \(p\)-time in the past to a \(B\)-value ready now. In more
  direct terms, this type refers to a function taking \(p\) time to perform its
  work.  The \Lamp principles match well with this
  meaning of modalities.
  For instance, \ref{principle:convertibility} lets us add
  arbitrary waiting times, and the \ref{principle:product} becomes
  \(\BOX p {(\BOX q A)} ↔ \BOX {p +_ℝ q} A\) and corresponds to
  compounding delays.





\subsection{Completeness}
\agdacomment{Subsection: Completeness}
\agdacomment{------------------------}
\agdablank

It may be still unclear if the resource model can account for
arbitrary functors, such as (co)monadic effects. It can: the key is to use a
\nameref{def:resource-module} which captures \emph{exactly} how a type can be
constructed, by letting resources be objects and modalities be functors.
By fully reflecting the \nameref{def:bicat}
structure, this instance of the resource model is complete wrt. the
chosen \nameref{def:bicat}.

\paragraph{Complete model}
Let us pick an arbitrary \nameref{def:bicat} \category{C}.
Then let the \nameref{def:resource-frame-2} arise from the \((\Tensor,\unitObj)\) monoid
from \category C (\(+ ≔ \Tensor\) ; \(0 ≔ \unitObj\) ; \(⇝~≔~→\)
) %
\agdaref[Worlds for completeness]{Sem.Presheaf.Properties}[Completeness]{worlds} and the
associated scalars be the
\nameref{def:functor-modalities}: \(\structure Q := \Funct{C}\). \agdaref[Modalities for
completeness]{Sem.Presheaf.Properties}[Completeness]{functorModul} Then we have a \nameref{def:resource-module}.
The \((◃)\) operation is the action of those functors. The module properties are immediate consequences of the
definitions. In particular, the distributivity properties are satisfied thanks to \cref{lem:distr-lax}.
In the arising presheaf model, a value \(u ∈ w ⊩ A\) contains \(\category C\) morphisms.

\begin{theorem}[completeness]\label{thm:completeness}
  Assume arbitrary base types \(\TyConst\) and arbitrary modalities \(\structure F\).
  Let \(g : \TyConst → \mathsf{Obj}(\category C)\), and \(\mathbf h : \structure F → \Funct C\).
  Then any type \(A\) can be interpreted to an object \(⟦A⟧\) in \category{C}, and:
  \begin{center}
    Given \(v ∈ (\Needs A 0)\) coming from the complete model, we can construct a morphism \(\unitObj \hom_{\category C} ⟦A⟧\).\agdaref[Completeness of resource-module semantics]{Sem.Presheaf.Properties}[Completeness]{complete}
  \end{center}
\end{theorem}
\newcommand\reify{\mathord{{↓}}}
\newcommand\reflect{\mathord{{↑}}}
\begin{proof}
  We prove the following two lemmas by mutual induction on the structure of the type \(A\):
  \begin{description}

  \item[Reification] If an object \(X\) produces \(A\) in the
    sense of the complete model, then it also produces its
    interpretation, \textit{via} a morphism.  That is, given an element \(u ∈ \Needs A X\), we can construct a morphism
    \(\reify^Au : X \hom_{\category C} ⟦A⟧\). Its definition is:\agdaref[Reification]{Sem.Presheaf.Properties}[Completeness]{reify}
  \[\begin{array}[t]{r@{\;}l}
    \reify^{A⊗B} (u,v,f) & = (\reify^A u ~⊗~ \reify^B v) ∘ f \\
    \reify^{A \With B} (u,v) & = (\reify^A u) ~▴~ (\reify^B v) \\
    \reify^{A⊸B} u & = \curry {\reify^B (u(\reflect_A))} \\
  \end{array}\hspace{1cm}
  \begin{array}[t]{r@{\;}l}
    \reify^{1} u & = u \\
    \reify^{\Ttop} u & = ε \\
    \reify^{\BOX p A} (u , f) & = \mathbf h (p) (\reify^A u) ∘ f \\
  \end{array}\]

    Furthermore, reification is natural: \(\reify^A {(u↿f)} = \reify^Au ∘ f \).

  \item[Reflection] The interpretation of a type as an object produces it in the
    sense of the resource semantics: for every type \(A\) there is an element \((\reflect_A) ∈ \Needs A {⟦A⟧_{\category C}}\).
    \agdaref[Reflection]{Sem.Presheaf.Properties}[Completeness]{reflect}
    The definition is:
    \[    \begin{array}[t]{r@{\;}l}
  \reflect^{A⊗B} & = ((\reflect^A) , (\reflect^B) , id) \\
  \reflect^{A \With B} & = ((\reflect^A) ↿ π₁ , (\reflect^B) ↿ π₂) \\
  (\reflect^{A⊸B})(u) & = (\reflect^B) ↿ (\appl ∘ (id ⊗ (\reify^A u) )) \\
\end{array}\hspace{1cm}
\begin{array}[t]{r@{\;}l}
  \reflect^{1} & = id \\
  \reflect^{\Ttop} & = ◆ \\
  \reflect^{\BOX p A} & = ((\reflect^A), \mathbf h (p) id) \\
\end{array}\]

  \end{description}
\end{proof}


\section{Normal Syntax}\label{sec:normalisation}
\agdacomment{Section: Normalisation}
\agdacomment{----------------------}
\agdablank

We now have the groundwork laid out to derive a normalising
syntax. Observe that \cref{thm:presheaf-soundness} can be used to map
\Lamp{} to a \nameref{def:world-presheaves-category}. The morphisms in
the presheaf category effectively perform the computations in
\category{Set} (the set of values). We also have
\nameref{thm:completeness}, which maps {\category {\widehat C}} back
to \category{C} via \nameref{def:bicat}. Thus, the composition of
these two results provides a \emph{normalisation by evaluation} (NbE)
procedure~\citep{coquand1991algorithm}. Much of the literature on NbE
is concerned with the computation of normal forms \textit{per se}, but
here we will also use it to derive the syntax of normal forms. We
could let \category{C} be the free \nameref{def:bicat}, but it is
preferable to target a language which excludes redexes \emph{by
  construction}.  This property is ensured by dividing the language
into two separate syntactic categories, one for constructors and one
for eliminators.  Additionally, we'd rather not target a category
(which would yield a point-free language), but rather a lambda
calculus with usual named variables.  We write \(Γ ⊨ A\) for the typing judgement associated with
constructors. Strikingly, the contexts for this language cannot be
\Lamp{} contexts.  Instead, we must use the \emph{free}
\nameref{def:resource-module} generated by syntactic modalities (for
scalars) and types (for resources).\footnote{Indeed, in \Lamp{}
  context scaling distributes on contexts \emph{by definition}, and
  consequently it would not be possible to reflect \(\BOX p\) types.}
In practice, this means two things. First, the syntax for contexts is
extended with explicit scaling:
\providecommand\singletonContext[1]{(#1)}
\providecommand\pairContext[2]{(#1, #2)}
\begin{align*}
  Γ,Δ,Ξ,Φ ::= \emptyContext ∣ Γ, x : A ∣ p ◃ Γ
\end{align*}
Second, patterns are extended with the means to manipulate such
scaling. In fact, scaling patterns is not sufficient, one must scale
whole \emph{pattern bindings}. The resulting syntax is shown in
\cref{fig:final-patterns}.  These pattern bindings correspond to the
laws of a \nameref{def:resource-module}.

The associated presheaf semantics arising from the resource structure
of \((⇝)\) provides a notion of \emph{value}, \(\Needs A {Γ}\).
\Cref{thm:presheaf-soundness} gives us that any \Lamp{} term \(Γ ⊢ A\)
evaluates to an element of \(\Needs A {Γ}\). We can proceed as in the
proof of \nameref{thm:completeness}.  The following two lemmas are
then proven by mutual induction on the type \(A\):
\begin{description}
  \item[Reification] Given a value \(v ∈ \Needs A {Γ}\), we can construct a normal form \(Γ ⊨ (\reify^A{v}) : A\). \agdaref{Normal.Free}{reify}
  \item[Reflection] For any type \(A\), there is a value \((\reflect_A) ∈ (\Needs A {\singletonContext {x:A}})\)\agdaref{Normal.Free}{reflect}
  \end{description}
  
We discuss normal forms below, in conjunction with the reification
lemma. But let us focus on reflection first. Its proof follows the
same structure as the corresponding lemma in the
\nameref{thm:completeness} theorem. Whenever a morphism is used in the
completeness proof, we add it to the generators of the pattern binding
syntax \(⇝\) as an explicit eliminator, one for each type
(\cref{fig:final-patterns}). These eliminators work on variables and
bind their results to variables directly, rather than acting on
arbitrary terms as in \Lamp{}. Notably, the application rule applies a
function (variable \(x\)) to a normal form (\(t\)). In sum, the
\((⇝)\) sub-language corresponds to generalized eliminators.

  \begin{figure}
    \figsepar{Pattern binding composition (Category structure)}
    \begin{mathpar}
      \inferrule{}{id : Γ ⇝ Γ}
      \and
    \inferrule{f : Γ ⇝ Δ  \\ g : Δ ⇝ Ξ}{f ; g : Γ ⇝ Ξ}
    \end{mathpar}
    \figsepar{Context reordering (Symmetric monoidal structure)}
  \begin{mathpar} %
    \inferrule{f : Γ_1 ⇝ Δ_1  \\ g : Γ_2 ⇝ Δ_2 }{f,g ~:~ Γ_1,Γ_2 ⇝ Δ_1,Δ_2}
      \and
    \inferrule{}{σ : Γ,Δ ⇝ Δ,Γ}
  \end{mathpar} %
    \figsepar{Context scaling manipulation (functorial properties thereof)}
  \begin{mathpar} %
    \inferrule{}{π : (p+q)◃Γ ⇝ p◃Γ, q◃Γ}
      \and
    \inferrule{}{ζ : 0◃Γ ⇝ \emptyContext}
      \and
    \inferrule{}{α : pq ◃ Γ ↭ p◃(q◃Γ)}
      \and
    \inferrule{}{ρ : 1◃Γ ↭ Γ}
      \and
    \inferrule{φ : p≼q \\ f : Γ ⇝ Δ}{φ(f) ~:~ p◃Γ ⇝ q◃Δ}
      \and
    \inferrule{\distributes{p}}{δ : p◃(Γ,Δ) ⇝ p◃Γ,p◃Δ}
      \and
    \inferrule{\codistributes{p}}{δ' : p◃Γ,p◃Δ ⇝ p◃(Γ,Δ)}
      \and
    \inferrule{\distributes{p}}{ε : p ◃ \emptyContext ⇝ \emptyContext}
      \and
    \inferrule{\codistributes{p}}{η : \emptyContext ⇝ p ◃ \emptyContext}
  \end{mathpar}
    \figsepar{Eliminators}
  \begin{mathpar}
    \inferrule{}{{\TBOX p y} = x ~:~  \singletonContext {x:\BOX p A} ⇝ p ◃ \singletonContext {y : A}}
    \and
    \inferrule{}{(x,y)=z~:~\singletonContext {z:A \Tensor B} ⇝ \pairContext {x{:}A} {y{:}B}}
    \and
    \inferrule{}{()=x~:~\singletonContext {x{:}\Unit} ⇝ \emptyContext}
    \and
    \inferrule{}{y = \tout i x~:~\singletonContext {x : A_1 \With A_2} ⇝ \singletonContext {y : A_i} }
    \and
    \inferrule{Γ ⊨ t : A}{y = x(t)~:~Γ, (x : \fun A B) ⇝ \singletonContext {y : B} }{\rul{app}}
  \end{mathpar}
  \figsepar{Term binding}
  \[\inferrule {Γ ⊨ t : A} {x=t ~:~ Γ ⇝ \singletonContext{x : A}}\]
  \figsepar{Normal forms (constructors)}
\begin{mathpar}
  \inferrule{s ~:~ Γ ⇝ \singletonContext {x{:}b} }{Γ ⊨ \wlet s x : b}
  \and
  \inferrule{s ~:~ Γ ⇝ \emptyContext} {Γ ⊨ \wlet s () : \Unit}
  \and
  \inferrule{s : Ξ ⇝ Γ,Δ\\Γ ⊨ t : A\\Δ ⊨ u : B}{Ξ ⊨ \wlet s (t,u) : A \Tensor B}
  \and
  \inferrule{\bla}{Γ ⊨ \ttop{Γ} : \Ttop}
  \and
  \inferrule{\forall i.~ Γ ⊨ tᵢ : A_i}{Γ ⊨ \twith {t₁} {t₂} : A_1 \With A_2}
  \and
  \inferrule{Γ,x:A ⊨ t : B}{Γ ⊨ \lam x t : \fun A B}
  \and
  \inferrule{Γ ⊨ t : A \\ s ~:~ Δ ⇝ p◃Γ }{Δ ⊨ \wlet s \TBOX p t : \BOX p A}
\end{mathpar}
\caption{Syntax of \PLamp{}. \agdaref{Normal.Syntax}{_≼_}.  We resort
  to a point-free syntax for scaling pattern bindings,
  because such scaling operations do not touch variable names or types; only
  modalities.  The symbol \(↭\) means that one can convert in both
  directions. For instance, \(ρ : 1◃Γ ↭ Γ\) stands for
  \(ρ⁺ : 1◃Γ ⇝ Γ\) and \(ρ⁻: Γ ⇝ 1◃Γ\). The associators of the symmetric monoidal structure are silent.}
    \label{fig:final-patterns}
  \end{figure}


  It remains to spell out the reification lemma.  In the \cref{thm:completeness},
  the corresponding lemma builds a morphism in the chosen
  category, which in this section would be
  \(Γ ⇝ \singletonContext {x{:}A}\). Applying the recipe literally would
  mean adding constructors to that category, mixing them with
  eliminators. But then it would not be obvious that terms are
  normalised. So instead we define a new syntactic category of normal
  forms, \(Γ ⊨ t:A\). To find out its inhabitants, we simply follow the
  structure of \cref{thm:completeness} and assert the pieces of
  syntax which we need, one for each type.

  What we have obtained is a \emph{polarised} syntax for our semantics. That is,
every intermediate result is bound to a variable.\agdaref{Normal.Syntax}{_⊨_}
A property of the constructor forms is that they can always be prefaced by a pattern binding. That is, the following rule is admissible:
\(\inferrule{Δ ⊨ t : A \\ f : Γ ⇝ Δ}{Γ ⊨ \wlet f t : A}\)\agdaref{Normal.Syntax}{_↿_}. The idea is
that the bindings (\(f\)) can be pushed down the syntax tree until
reaching one of the constructions with native binding (variable, \((t,u)\) or \(\TBOX p {t}\)) and inserted there.
Note that this procedure never performs a computation: it consists in
prepending a context transformation before constructors, not the other way around. This operation the restriction map for \(⊨ A\), as a presheaf.

The equational theory of the pattern-binding fragment is given by the laws of a \nameref{def:resource-module} (\cref{fig:eq-neg}).
Those correspond to commuting conversions in sequent calculi presentations.
Equality for \((⊨)\) just a congruence for each constructor. However, congruence for \((t,u)\) and \(\TBOX p {t}\) is unusual:
they assert that bindings can equally be applied to
the compound or to the constituents, as in the presheaf semantics.

\section{Completing the calculus}
Unlike \Lamp{}, the language of normal forms presented above has a (well-behaved) equational theory.
However, by definition, it contains only normal forms and thus cannot represent computation. This
shortcoming is addressed in this section. The resulting language
is called \PLamp{}. We prove that \PLamp{} retains all the
desirable theoretical properties of a core calculus. One top of normal forms, we only
add the binding of an arbitrary term: \agdaref{Normal.Syntax}{cut}
\[\inferrule {Γ ⊨ t : A} {x=t ~:~ Γ ⇝ A}\]

\begin{figure}
  \figsepar{Equivalence}
  \begin{mathpar}
    {f ≅ f} \and 
    \inferrule{f ≅ g\\g ≅ h}{f ≅ h} \and 
    \inferrule{f ≅ g}{g ≅ f} \and 
  \end{mathpar}

  \figsepar{Laws of a \nameref{def:resource-module} (naturality of context scaling manipulation)}
  \begin{mathpar}
    id(id) ≅ id  \and
    φ(f) ; ψ(g) ≅ (φ ∘ ψ)(f ; g)  \and
    δ'; (id(f), id(g)) ≅ id(f,g) ; δ'  \and
    δ; id(f,g) ≅ (id(f), id(g)) ; δ  \and
    π ; id(g) ≅ (id(g), id(g)) ; π  \and
    ζ ; id(f) ≅ ζ  \and
    x=z(t);(id, f) ≅ x=z(\wlet f t)  \and
    ρ⁺; f ≅ id(f) ; ρ⁺  \and
    ρ⁻;f  ≅  ρ⁻;id(f)  \and
    α⁺ ; (φ ∘ ψ)(f) ≅ φ(ψ(f)) ; α⁺ \and
    (φ ∘ ψ)(f) ; α⁻ ≅ α⁻ ; φ(ψ(f)) \and
  \end{mathpar}


   \figsepar{Administration of let bindings}
  \begin{mathpar}
  ((x=t) ; f)   \; ≅ \;   (x = \wlet f t) 
    \and
    {(x=(\wlet f x))\; ≅ \; f} 
    \and
    (\wlet {x=t} x)  ≅  t 
\end{mathpar}
  \figsepar{Computational rules}
  \begin{mathpar}
    (y=x(u) ; x = λz.t)  ≅ (y=t ; z=u)  \and 
    ((x,y)=z ;  z = (t,u)) ≅ (x=t; y=u) \and
    (\TBOX p t=z ; x = \TBOX p t) ≅ id_p(x=t) \and
    ()=z ; z= () ≅ id \and
    x = \tout i z ; z = \twith {t₁} {t₂} ≅ x = t_i \and
  \end{mathpar}
  \figsepar{Congruence for constructors}
  \begin{mathpar}
    \inferrule{t ≅ t' \\ u ≅ u'}{\twith t u ≅ \twith {t'} {u'}} \and
    \inferrule{t ≅ t'}{\lam x t ≅ \lam x t'} \and

  \inferrule {\wlet f t ≅ t' \\ \wlet g u ≅ u'}{\wlet {f,g} (t,u) ≅ (t',u')}
  \and
  \inferrule {\wlet f t ≅ t'}{\wlet {id_p(f)} \TBOX p t ≅ \TBOX p t'}
  \end{mathpar}
  \caption{Equational theory for \PLamp{}. We omit the usual SMC laws
    (equivalence, naturality) as well as congruence rules for
    pattern bindings: they are unsurprising.  \agdaref{Normal.Equations}{_≋_}
    \agdaref{Normal.Equations}{_≅_}
  }
  \label{fig:eq-neg}
\end{figure}

The equational theory is extended with administrative rules for term bindings, naturality laws of for the manipulation of context scaling,
and computational rules for each type (\cref{fig:eq-neg}).
This completes the definition of \PLamp{}. It admits the same
models as \Lamp{}: given a \nameref{def:mod-interp} into \nameref{def:functor-modalities}, terms can be evaluated into any
\nameref{def:bicat} \(\category C\), in a straightforward
manner. \agdaref{Normal.Sem}{eval} This interpretation,
\(⟦·⟧^*_{\category C}\), is complete with respect to term equality: if
\(t ≅ u\) then \(⟦t⟧^*_{\category C} = ⟦u⟧^*_{\category
  C}\). \agdaref{Normal.Sem1}{eval≅}

Because both \PLamp{} and \Lamp{} are interpreted into the same
categorical structure, \PLamp{} is subject to the same normalisation
by evaluation procedure as \Lamp{}. To reiterate, one first evaluates
in the presheaf category over eliminators, then reifies into normal
forms: \(\nf(t) ≜ \reify{⟦t⟧^*}\). \agdaref{Normal.Renormalisation}{renorm}
Since \Lamp{} lacks an equational theory, we could not prove, or even
state, any useful property of its normalisation procedure. However, we can do so
for \PLamp{}:
\begin{theorem}[Soundness of normalisation for \PLamp{}]
  \(\nf(t) ≅ t\) for any \(Γ ⊨ t : A\).
  \agdaref{Normal.Renormalisation}{renorm-idem}
\end{theorem}
\begin{proof}
  The proof follows the standard pattern. A logical relation \((∼)\) is
  defined between terms and values. \agdaref{Normal.Relations}{_∼_} The result falls out of the
  combination of two lemmas: the fact that values are related to their
  reification \((v ∼ \reify{v})\) \agdaref{Normal.Relations}{Sim-reify} and the fundamental lemma of logical relations.\agdaref{Normal.EvalRel}{fundamental}
\end{proof}
As corollaries, we have:
  \begin{itemize}
  \item stability: \(\nf(\nf(t)) ≅ \nf(t)\) and
  \item preservation of meaning: \(⟦t⟧^* = ⟦\nf(t)⟧^*\).
  \end{itemize}

  Finally, we can show that the semantics of \Lamp{} and \PLamp{} agree. To do this, we construct the syntactic category of
  \PLamp{} terms, hereafter called \(\category {Syn}\). \agdaref{Normal.Cat}{scat}
  \(\category{Syn}\) is a \nameref{def:bicat} where objects are types, and any morphism \(A → B\) is a term \(A ⊨ t : B\).
  \begin{theorem}
     For any \nameref{def:bicat} {\category C} and any \(Γ ⊢ t : A\), we have  \(⟦⟦t⟧_{\category {Syn}}⟧^*_{\category C} = ⟦t⟧_{\category C}\). \agdaref{Normal.CatOK}{thm}
   \end{theorem}
   As a corollary, the normal forms obtained directly from \Lamp{}
   terms are the same as those that one would obtain by first
   translating them to \PLamp{}.  The following diagram depicts the
   possible interpretations, as well as reification (\(\reify\))
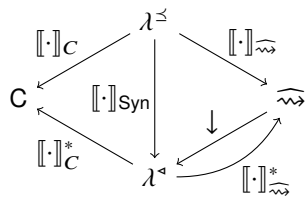
   \begin{center}
     \begin{tikzpicture}
       [post/.style={semithick,auto,inner sep=1pt}]
       \node (source) at (120:1.2)  {\Lamp{}};
       \node (target) at (240:1.2)  {\PLamp{}};
       \node (values) at (0:1.2)  {$\widehat {⇝}$};
       \node (arbit) at (-2.4,0)  {$\category C$};

       \draw [->] (source) to node[post,swap] {$⟦·⟧_{\category {Syn}}$} (target);
       \draw [->] (source) to node[post] {$⟦·⟧_{\widehat {⇝}}$} (values);
       \draw [->] (target) to[bend right] node[post,swap] {$⟦·⟧^*_{\widehat {⇝}}$} (values);
       \draw [->] (values) to node[post,swap] {$\reify$} (target);
       \draw [->] (source) to node[post,swap] {$⟦·⟧_{C}$} (arbit);
       \draw [->] (target) to node[post] {$⟦·⟧^*_{C}$} (arbit);
     \end{tikzpicture}
   \end{center}
   In the above, \(\widehat {⇝}\) refers to the presheaf category over the \(⇝\) category.

   In summary \Lamp{} has a well-defined semantics, but reasoning
   about \Lamp{} programs cannot be done using its syntax. If
   syntactic reasoning is needed, it one can safely translate first
   into \PLamp{}.

\section{Related work}\label{sec:related-work}

At its core, our framework is syntactically and semantically based on
linear logic~\citep{girard_linear_1987}.  The generalisation of
modalities to ring-like structures is due to
\citet{girard_bounded_1992}, and further refined, \textit{inter alia},
by \citet{dal_lago_bounded_2009} or \citet{ghica_bounded_2014}.

The idea of interpreting modalities as functors can be traced at least
to \citet{ghilardi1988modal}, and more explicitly
\citet{bierman2000intuitionistic}.
\Citet{brunel_core_2014} lays the concrete groundwork which we
build upon. The principal generalisations that we offer are those
listed in the introduction: we don't require distributivity (monoidal functors), and we
explicitly allow computational content in the \nameref{def:ringoid}, but
we also unify the treatment of effects and coeffects, as well as quantitative modalities.
(The requirement on distributivity and monoidal functor pervades literature which extends modality applications in several ways, i.a. the work of \citet{gratzer_multimodal_2021}.)
Another difference is that \citeauthor{brunel_core_2014} propose a
realisability semantics which resembles our resource semantics, but it
features (double) orthogonalisation, and as such cannot be made an
instance of the general categorical semantics. Additionally, we
provide a completeness theorem and support $(\With)$ products.


\label{sec:merge}
Finally, \citeauthor{brunel_core_2014} sustain the converse of the \ref{principle:sum}:
\(\BOX p A \Tensor \BOX q A \funOp \BOX {p+q} A\). For arrays, such a
principle would correspond to concatenation, and it is central in the fractional ownership
system of \citet{makwana_numlin_2019}, so it is clearly useful for some
applications. The main reason for us leaving it out
is that we could not make sense of it with regard to
substitution. Indeed, it demands that two substitutions \(Γ₁ ↦ Δ\) and
\(Γ₂ ↦ Δ\) can be merged, which appears impossible in general.


\Citet{abel_unified_2020} have proposed a unified framework
for modalities in type-systems similar to ours, and have provided the relational
semantics combining parametricity interpretations
\citep{reynolds_types_1983}, classified sets
\citep{kavvos:popl19,abadi_core_1999}, and resource indexing
\citep{atkeyWood:tyde18,brunel_core_2014}. We have made this semantics
an instance of the more general categorical framework, hoping to make
it more approachable. Additionally, we support non-distributive
modalities and our semantics is proven complete.
\Citeauthor{abel_unified_2020} natively support coproducts,
whereas we encode them using inversion of control. It leads them to
propose that sum types can be analysed iff they are qualified with a
modality \(p ≼ 1\). According to our semantic interpretation, the
condition on \(p\) is the existence of a morphism
\(p(A \sumObj B) \funObj pA \sumObj pB\). In sum, our framework
provides a semantic reason for the fact that sensitivity analysis
treats coproducts differently from other systems. We propose that supporting sum requires shifting to a semantics based on sheaves (instead of presheaves).
The treatment of \citet{DBLP:conf/lics/AltenkirchDHS01} should be adaptable to our setting, but we leave it as future work.

\paragraph{Modalities as 2-categories and arbitrary functorial interpretations}
In a recent line of work,
\citet{kavvos_two-dimensional-I_2024,kavvos_two-dimensional-II_2024}
has proposed to interpret modalities as
functors. \Citeauthor{kavvos_two-dimensional-I_2024} does not consider
generic modality structures as we do, but concentrate on the case of
dual modalities and an underlying cartesian category where (\(\With\))
and (\(\Tensor\)) coincide.
However \citeauthor{kavvos_two-dimensional-I_2024}, similarly to \citet{Orchard_2020},  offers another interesting direction of
generalisation: for him modalities are not a set, but rather
morphisms in a category. The objects of this category are called
\emph{modes}. It appears that \Lamp{} can be generalised in this
direction: the relation \((≼)\) becomes the hom-set of a 2-category
and the typing relation is indexed by a mode \(m\), as in
\(Γ ⊢ t : A @ m\). The semantic interpretation would target a
\emph{family} of monoidal categories \(\category C(m)\), and
modalities would be interpreted as functors between those. Working out
the details, in particular the generalised resource-semantics and its
implications, is left as future work.

To our knowledge, the present work is the only one to offer a fully
mechanised formalisation.


\paragraph{Semantics}

The resource semantics can be considered a reformulation of that of
\citet{abel_unified_2020} in a presheaf
framework. \citet{ohearn_logic_1999} defined the same semantics much
earlier, but for a language lacking modalities altogether. The
presheaf semantics of tensor products have been studied by
\citet{day_closed_1970}. The presheaf semantics for functors can be
traced to \citet{day_kan_1995}. They are primarily interested in
monoidal functors, but what we use (the left Kan extension) carries
over to our setting.

\section{Discussion}
\label{sec:discussion}

\paragraph{Type-checking}\label{sec:typechecking}
We have focused primarily on \Lamp{} semantics, but any concrete
language based on it should feature a typechecker. As presented, the
\rul{\Tensor-intro} and \rul{app} rules split the context
non-deterministically, and it may not be clear how this split should
be implemented in a typechecker. One solution is to let
the typing rules to return a context of unused
variables~\citep{DBLP:conf/types/Allais17}.
Additionally, for type checking to be
possible, the set of modalities need a decidable equality in
addition to the structure necessary for semantic purposes.

\paragraph{Type-Polymorphism}\label{sec:polymorphism}
We have chosen to leave out polymorphism from our development, for the following reasons:
1. Polymorphism is mostly orthogonal from our concerns.
The only interaction happens in the encoding of sums, but
that is invisible except for the discussion~(\ref{eq:liberal-case-analysis}) of case analysis for sensitivity.
2. A relational treatment
of polymorphism does not fit the simple categorical framework. One
would have to bring out at least the concept of display category~\citep{Ahrens_2019}.
Even though we have not presented this version in the
paper, some results based on this approach are available in our
formalisation. \agdaref{Rel.Sem.Kripke}{cat}

\paragraph{Even fewer structural rules}
We have considered modalities which limit duplication and discard of
variables, but one can also use them to, for instance, control movement
(associativity and commutativity) of variables in the
context~\citep{moortgat_multimodal_1996}. Support for this even more
general case could work in a similar way as our treatment of
(optional) distributivity over products, but is left for future work.

\paragraph{Poset modalities}\label{sec:poset-modalities}

We have considered modalities which form a category, which is a
computationally-relevant version of a poset. This lets us support
arbitrary effects (see \cref{rk:effects-are-relevant}).  This means
that in general modality conversions must be explicit in syntax.  In
contrast, if one assumes irrelevant conversions, one can on the
contrary let the type-checker infer such conversions, as
\citet{orchard_quantitative_2019} do. But, one can also compute whether
there is a single conversion, and use it if it is unique or report an
error otherwise.

\section{Conclusion}\label{sec:conclusion}

We have shown that general functorial programming can be expressed
using the syntax typically reserved for modalities or grade, albeit
with some tweaks.  Certain syntactic constructs are reserved to
(op)lax functors, and when multiple natural transformations exist
between functors, conversions between modalities must be made
explicit.

Our language unifies what is traditionally seen as modalities
(quantitative applications) and as functors (data structures, effects,
coeffects). The functional-programming community has proposed a
variety of surface syntax for each domain—do-notations, idiom
brackets, parallel list comprehensions, etc. Our proposal subsumes
these alternatives, even when all these applications are mixed in a
single program. It does so with only a handful of syntactic
elements. Although there is some syntactic overhead (particularly for
effects), it can be mitigated with syntactic sugar. In cases where
modality conversions can be inferred, the syntax can even be
lightweight.

Given that modalities are, in fact, functors, why not program directly
with those functors? Indeed, functorial and (co)monadic programming is
already possible today, but it requires juggling with a zoo of
combinators. We argue that a direct syntax—much like lambda notation’s
advantage over SKI combinators—offers a clearer, more ergonomic
approach.





\bibliography{local,jp,confs}


\begin{thebibliography}{38}


\ifx \showCODEN    \undefined \def \showCODEN     #1{\unskip}     \fi
\ifx \showDOI      \undefined \def \showDOI       #1{#1}\fi
\ifx \showISBNx    \undefined \def \showISBNx     #1{\unskip}     \fi
\ifx \showISBNxiii \undefined \def \showISBNxiii  #1{\unskip}     \fi
\ifx \showISSN     \undefined \def \showISSN      #1{\unskip}     \fi
\ifx \showLCCN     \undefined \def \showLCCN      #1{\unskip}     \fi
\ifx \shownote     \undefined \def \shownote      #1{#1}          \fi
\ifx \showarticletitle \undefined \def \showarticletitle #1{#1}   \fi
\ifx \showURL      \undefined \def \showURL       {\relax}        \fi
\providecommand\bibfield[2]{#2}
\providecommand\bibinfo[2]{#2}
\providecommand\natexlab[1]{#1}
\providecommand\showeprint[2][]{arXiv:#2}

\bibitem[\protect\citeauthoryear{Abadi, Banerjee, Heintze, and Riecke}{Abadi
  et~al\mbox{.}}{1999}]%
        {abadi_core_1999}
\bibfield{author}{\bibinfo{person}{M. Abadi}, \bibinfo{person}{A. Banerjee},
  \bibinfo{person}{N. Heintze}, {and} \bibinfo{person}{J.G. Riecke}.}
  \bibinfo{year}{1999}\natexlab{}.
\newblock \showarticletitle{A core calculus of dependency}. In
  \bibinfo{booktitle}{\emph{Proceedings of the 26th {ACM} {SIGPLAN-SIGACT}
  symposium on Principles of programming languages}}. ACM,
  \bibinfo{pages}{147--160}.
\newblock


\bibitem[\protect\citeauthoryear{Abel and Bernardy}{Abel and Bernardy}{2020}]%
        {abel_unified_2020}
\bibfield{author}{\bibinfo{person}{Andreas Abel} {and}
  \bibinfo{person}{Jean-Philippe Bernardy}.} \bibinfo{year}{2020}\natexlab{}.
\newblock \showarticletitle{A unified view of modalities in type systems}.
\newblock \bibinfo{journal}{\emph{Proceedings of the ACM on Programming
  Languages}} \bibinfo{volume}{4}, \bibinfo{number}{ICFP}
  (\bibinfo{year}{2020}).
\newblock


\bibitem[\protect\citeauthoryear{Ahrens and Lumsdaine}{Ahrens and
  Lumsdaine}{2019}]%
        {Ahrens_2019}
\bibfield{author}{\bibinfo{person}{Benedikt Ahrens} {and}
  \bibinfo{person}{Peter~LeFanu Lumsdaine}.} \bibinfo{year}{2019}\natexlab{}.
\newblock \showarticletitle{Displayed Categories}.
\newblock \bibinfo{journal}{\emph{Logical Methods in Computer Science}}
  \bibinfo{volume}{Volume 15, Issue 1} (\bibinfo{date}{March}
  \bibinfo{year}{2019}).
\newblock
\showISSN{1860-5974}
\urldef\tempurl%
\url{https://doi.org/10.23638/lmcs-15(1:20)2019}
\showDOI{\tempurl}


\bibitem[\protect\citeauthoryear{Algehed and Bernardy}{Algehed and
  Bernardy}{2019}]%
        {algehed_simple_2019}
\bibfield{author}{\bibinfo{person}{Maximilian Algehed} {and}
  \bibinfo{person}{Jean-Philippe Bernardy}.} \bibinfo{year}{2019}\natexlab{}.
\newblock \showarticletitle{Simple noninterference from parametricity}.
\newblock \bibinfo{journal}{\emph{Proceedings of the ACM on Programming
  Languages}} \bibinfo{volume}{3}, \bibinfo{number}{ICFP}
  (\bibinfo{year}{2019}), \bibinfo{pages}{89}.
\newblock


\bibitem[\protect\citeauthoryear{Allais}{Allais}{2017}]%
        {DBLP:conf/types/Allais17}
\bibfield{author}{\bibinfo{person}{Guillaume Allais}.}
  \bibinfo{year}{2017}\natexlab{}.
\newblock \showarticletitle{Typing with Leftovers - {A} mechanization of
  Intuitionistic Multiplicative-Additive Linear Logic}. In
  \bibinfo{booktitle}{\emph{23rd International Conference on Types for Proofs
  and Programs, {TYPES} 2017, Budapest, Hungary, May 29 - June 1, 2017}}
  \emph{(\bibinfo{series}{LIPIcs})}, \bibfield{editor}{\bibinfo{person}{Andreas
  Abel}, \bibinfo{person}{Fredrik~Nordvall Forsberg}, {and}
  \bibinfo{person}{Ambrus Kaposi}} (Eds.), Vol.~\bibinfo{volume}{104}.
  \bibinfo{publisher}{Schloss Dagstuhl - Leibniz-Zentrum f{\"{u}}r Informatik},
  \bibinfo{pages}{1:1--1:22}.
\newblock
\urldef\tempurl%
\url{https://doi.org/10.4230/LIPICS.TYPES.2017.1}
\showDOI{\tempurl}


\bibitem[\protect\citeauthoryear{Altenkirch, Dybjer, Hofmann, and
  Scott}{Altenkirch et~al\mbox{.}}{2001}]%
        {DBLP:conf/lics/AltenkirchDHS01}
\bibfield{author}{\bibinfo{person}{Thorsten Altenkirch}, \bibinfo{person}{Peter
  Dybjer}, \bibinfo{person}{Martin Hofmann}, {and} \bibinfo{person}{Philip~J.
  Scott}.} \bibinfo{year}{2001}\natexlab{}.
\newblock \showarticletitle{Normalization by Evaluation for Typed Lambda
  Calculus with Coproducts}. In \bibinfo{booktitle}{\emph{16th Annual {IEEE}
  Symposium on Logic in Computer Science, Boston, Massachusetts, USA, June
  16-19, 2001, Proceedings}}. \bibinfo{publisher}{{IEEE} Computer Society},
  \bibinfo{pages}{303--310}.
\newblock
\urldef\tempurl%
\url{https://doi.org/10.1109/LICS.2001.932506}
\showDOI{\tempurl}


\bibitem[\protect\citeauthoryear{Atkey and Wood}{Atkey and Wood}{2018}]%
        {atkeyWood:tyde18}
\bibfield{author}{\bibinfo{person}{Robert Atkey} {and} \bibinfo{person}{James
  Wood}.} \bibinfo{year}{2018}\natexlab{}.
\newblock \showarticletitle{Context constrained computation}. In
  \bibinfo{booktitle}{\emph{Third Workshop on Type-Driven Development
  (TyDe'18)}}.
\newblock
\urldef\tempurl%
\url{https://personal.cis.strath.ac.uk/james.wood.100/pub/context-constrained.pdf}
\showURL{%
\tempurl}


\bibitem[\protect\citeauthoryear{Bierman and de~Paiva}{Bierman and
  de~Paiva}{2000}]%
        {bierman2000intuitionistic}
\bibfield{author}{\bibinfo{person}{Gavin~M. Bierman} {and}
  \bibinfo{person}{Valeria~CV de Paiva}.} \bibinfo{year}{2000}\natexlab{}.
\newblock \showarticletitle{On an intuitionistic modal logic}.
\newblock \bibinfo{journal}{\emph{Studia Logica}} \bibinfo{volume}{65},
  \bibinfo{number}{3} (\bibinfo{year}{2000}), \bibinfo{pages}{383--416}.
\newblock


\bibitem[\protect\citeauthoryear{Brunel, Gaboardi, Mazza, and Zdancewic}{Brunel
  et~al\mbox{.}}{2014}]%
        {brunel_core_2014}
\bibfield{author}{\bibinfo{person}{Aloïs Brunel}, \bibinfo{person}{Marco
  Gaboardi}, \bibinfo{person}{Damiano Mazza}, {and} \bibinfo{person}{Steve
  Zdancewic}.} \bibinfo{year}{2014}\natexlab{}.
\newblock \showarticletitle{A Core Quantitative Coeffect Calculus}. In
  \bibinfo{booktitle}{\emph{Programming Languages and Systems - 23rd European
  Symposium on Programming, {ESOP} 2014, Held as Part of the European Joint
  Conferences on Theory and Practice of Software, {ETAPS} 2014, Grenoble,
  France, April 5-13, 2014, Proceedings}} \emph{(\bibinfo{series}{Lecture Notes
  in Computer Science})}, \bibfield{editor}{\bibinfo{person}{Zhong Shao}}
  (Ed.), Vol.~\bibinfo{volume}{8410}. \bibinfo{publisher}{Springer},
  \bibinfo{pages}{351--370}.
\newblock
\urldef\tempurl%
\url{https://doi.org/10.1007/978-3-642-54833-8_19}
\showDOI{\tempurl}


\bibitem[\protect\citeauthoryear{Clouston}{Clouston}{2018}]%
        {clouston_fitch_2018}
\bibfield{author}{\bibinfo{person}{Ranald Clouston}.}
  \bibinfo{year}{2018}\natexlab{}.
\newblock \showarticletitle{Fitch-style modal lambda calculi}. In
  \bibinfo{booktitle}{\emph{International Conference on Foundations of Software
  Science and Computation Structures}}. Springer, \bibinfo{pages}{258--275}.
\newblock


\bibitem[\protect\citeauthoryear{Coquand}{Coquand}{1991}]%
        {coquand1991algorithm}
\bibfield{author}{\bibinfo{person}{Thierry Coquand}.}
  \bibinfo{year}{1991}\natexlab{}.
\newblock \showarticletitle{An algorithm for testing conversion in type
  theory}.
\newblock \bibinfo{journal}{\emph{Logical frameworks}}  \bibinfo{volume}{1}
  (\bibinfo{year}{1991}), \bibinfo{pages}{255--279}.
\newblock


\bibitem[\protect\citeauthoryear{{Dal Lago} and Hofmann}{{Dal Lago} and
  Hofmann}{2009}]%
        {dal_lago_bounded_2009}
\bibfield{author}{\bibinfo{person}{Ugo {Dal Lago}} {and}
  \bibinfo{person}{Martin Hofmann}.} \bibinfo{year}{2009}\natexlab{}.
\newblock \showarticletitle{Bounded linear logic, revisited}.
\newblock In \bibinfo{booktitle}{\emph{Typed Lambda Calculi and Applications}}.
  \bibinfo{publisher}{Springer}, \bibinfo{pages}{80--94}.
\newblock


\bibitem[\protect\citeauthoryear{Day}{Day}{1970}]%
        {day_closed_1970}
\bibfield{author}{\bibinfo{person}{Brian Day}.}
  \bibinfo{year}{1970}\natexlab{}.
\newblock \showarticletitle{On closed categories of functors}. In
  \bibinfo{booktitle}{\emph{Reports of the Midwest Category Seminar IV}}.
  Springer-Verlag, \bibinfo{pages}{1--38}.
\newblock


\bibitem[\protect\citeauthoryear{Day and Street}{Day and Street}{1995}]%
        {day_kan_1995}
\bibfield{author}{\bibinfo{person}{Brian Day} {and} \bibinfo{person}{Ross
  Street}.} \bibinfo{year}{1995}\natexlab{}.
\newblock \showarticletitle{Kan extensions along promonoidal functors}.
\newblock \bibinfo{journal}{\emph{Theory and applications of categories}}
  \bibinfo{volume}{1}, \bibinfo{number}{4} (\bibinfo{year}{1995}),
  \bibinfo{pages}{72--77}.
\newblock


\bibitem[\protect\citeauthoryear{de~Amorim, Gaboardi, Hsu, Katsumata, and
  Cherigui}{de~Amorim et~al\mbox{.}}{2017}]%
        {amorim_semantic_2017}
\bibfield{author}{\bibinfo{person}{Arthur~Azevedo de Amorim},
  \bibinfo{person}{Marco Gaboardi}, \bibinfo{person}{Justin Hsu},
  \bibinfo{person}{Shin{-}ya Katsumata}, {and} \bibinfo{person}{Ikram
  Cherigui}.} \bibinfo{year}{2017}\natexlab{}.
\newblock \showarticletitle{A semantic account of metric preservation}. In
  \bibinfo{booktitle}{\emph{Proceedings of the 44th {ACM} {SIGPLAN} Symposium
  on Principles of Programming Languages, {POPL} 2017, Paris, France, January
  18-20, 2017}}. \bibinfo{pages}{545--556}.
\newblock
\urldef\tempurl%
\url{http://dl.acm.org/citation.cfm?id=3009890}
\showURL{%
\tempurl}


\bibitem[\protect\citeauthoryear{Ghica and Smith}{Ghica and Smith}{2014}]%
        {ghica_bounded_2014}
\bibfield{author}{\bibinfo{person}{Dan~R. Ghica} {and} \bibinfo{person}{Alex~I.
  Smith}.} \bibinfo{year}{2014}\natexlab{}.
\newblock \showarticletitle{Bounded Linear Types in a Resource Semiring}. In
  \bibinfo{booktitle}{\emph{Programming Languages and Systems - 23rd European
  Symposium on Programming, {ESOP} 2014, Held as Part of the European Joint
  Conferences on Theory and Practice of Software, {ETAPS} 2014, Grenoble,
  France, April 5-13, 2014, Proceedings}}. \bibinfo{pages}{331--350}.
\newblock
\urldef\tempurl%
\url{https://doi.org/10.1007/978-3-642-54833-8_18}
\showDOI{\tempurl}


\bibitem[\protect\citeauthoryear{Ghilardi and Meloni}{Ghilardi and
  Meloni}{2006}]%
        {ghilardi1988modal}
\bibfield{author}{\bibinfo{person}{Silvio Ghilardi} {and} \bibinfo{person}{GC
  Meloni}.} \bibinfo{year}{2006}\natexlab{}.
\newblock \showarticletitle{Modal and tense predicate logic: models in
  presheaves and categorical conceptualization}. In
  \bibinfo{booktitle}{\emph{Categorical Algebra and its Applications:
  Proceedings of a Conference, held in Louvain-La-Neuve, Belgium, July
  26--August 1, 1987}}. Springer, \bibinfo{pages}{130--142}.
\newblock


\bibitem[\protect\citeauthoryear{Girard}{Girard}{1987}]%
        {girard_linear_1987}
\bibfield{author}{\bibinfo{person}{Jean-Yves Girard}.}
  \bibinfo{year}{1987}\natexlab{}.
\newblock \showarticletitle{Linear logic}.
\newblock \bibinfo{journal}{\emph{Theoretical Computer Science}}
  \bibinfo{volume}{50}, \bibinfo{number}{1} (\bibinfo{year}{1987}),
  \bibinfo{pages}{1--101}.
\newblock


\bibitem[\protect\citeauthoryear{Girard}{Girard}{1991}]%
        {girard_new_1991}
\bibfield{author}{\bibinfo{person}{Jean-Yves Girard}.}
  \bibinfo{year}{1991}\natexlab{}.
\newblock \showarticletitle{A new constructive logic: classical logic}.
\newblock \bibinfo{journal}{\emph{Mathematical Structures in Computer Science}}
  \bibinfo{volume}{1}, \bibinfo{number}{3} (\bibinfo{year}{1991}),
  \bibinfo{pages}{255--296}.
\newblock


\bibitem[\protect\citeauthoryear{Girard, Scedrov, and Scott}{Girard
  et~al\mbox{.}}{1992}]%
        {girard_bounded_1992}
\bibfield{author}{\bibinfo{person}{Jean-Yves Girard}, \bibinfo{person}{Andre
  Scedrov}, {and} \bibinfo{person}{Philip~J. Scott}.}
  \bibinfo{year}{1992}\natexlab{}.
\newblock \showarticletitle{Bounded linear logic: a modular approach to
  polynomial-time computability}.
\newblock \bibinfo{journal}{\emph{Theoretical Computer Science}}
  \bibinfo{volume}{97}, \bibinfo{number}{1} (\bibinfo{year}{1992}),
  \bibinfo{pages}{1 -- 66}.
\newblock
\showISSN{0304-3975}
\urldef\tempurl%
\url{https://doi.org/10.1016/0304-3975(92)90386-T}
\showDOI{\tempurl}


\bibitem[\protect\citeauthoryear{Gratzer, Kavvos, Nuyts, and Birkedal}{Gratzer
  et~al\mbox{.}}{2021}]%
        {gratzer_multimodal_2021}
\bibfield{author}{\bibinfo{person}{Daniel Gratzer}, \bibinfo{person}{G.~A.
  Kavvos}, \bibinfo{person}{Andreas Nuyts}, {and} \bibinfo{person}{Lars
  Birkedal}.} \bibinfo{year}{2021}\natexlab{}.
\newblock \showarticletitle{Multimodal Dependent Type Theory}.
\newblock \bibinfo{journal}{\emph{Log. Methods Comput. Sci.}}
  \bibinfo{volume}{17}, \bibinfo{number}{3} (\bibinfo{year}{2021}).
\newblock
\urldef\tempurl%
\url{https://doi.org/10.46298/LMCS-17(3:11)2021}
\showDOI{\tempurl}


\bibitem[\protect\citeauthoryear{Hinze}{Hinze}{[n.d.]}]%
        {hinze_monads_manuscript}
\bibfield{author}{\bibinfo{person}{Ralf Hinze}.}
  \bibinfo{year}{[n.d.]}\natexlab{}.
\newblock \bibinfo{title}{Monads from Comonads, Comonads from Monads}.
\newblock
\newblock
\urldef\tempurl%
\url{https://www.cs.ox.ac.uk/ralf.hinze/WG2.8/28/slides/Comonad.pdf}
\showURL{%
\tempurl}


\bibitem[\protect\citeauthoryear{Kavvos}{Kavvos}{2019}]%
        {kavvos:popl19}
\bibfield{author}{\bibinfo{person}{Georgios~Alexandros Kavvos}.}
  \bibinfo{year}{2019}\natexlab{}.
\newblock \showarticletitle{Modalities, cohesion, and information flow}.
\newblock \bibinfo{journal}{\emph{Proceedings of the ACM on Programming
  Languages}} \bibinfo{volume}{3}, \bibinfo{number}{{POPL}}
  (\bibinfo{year}{2019}), \bibinfo{pages}{20:1--20:29}.
\newblock
\urldef\tempurl%
\url{https://doi.org/10.1145/3290333}
\showDOI{\tempurl}


\bibitem[\protect\citeauthoryear{Kavvos}{Kavvos}{2024a}]%
        {kavvos_two-dimensional-I_2024}
\bibfield{author}{\bibinfo{person}{Georgios~Alexandros Kavvos}.}
  \bibinfo{year}{2024}\natexlab{a}.
\newblock \showarticletitle{Two-Dimensional Kripke Semantics {I:} Presheaves}.
  In \bibinfo{booktitle}{\emph{9th International Conference on Formal
  Structures for Computation and Deduction, {FSCD} 2024, July 10-13, 2024,
  Tallinn, Estonia}} \emph{(\bibinfo{series}{LIPIcs})},
  \bibfield{editor}{\bibinfo{person}{Jakob Rehof}} (Ed.),
  Vol.~\bibinfo{volume}{299}. \bibinfo{publisher}{Schloss Dagstuhl -
  Leibniz-Zentrum f{\"{u}}r Informatik}, \bibinfo{pages}{14:1--14:23}.
\newblock
\urldef\tempurl%
\url{https://doi.org/10.4230/LIPICS.FSCD.2024.14}
\showDOI{\tempurl}


\bibitem[\protect\citeauthoryear{Kavvos}{Kavvos}{2024b}]%
        {kavvos_two-dimensional-II_2024}
\bibfield{author}{\bibinfo{person}{G.~A. Kavvos}.}
  \bibinfo{year}{2024}\natexlab{b}.
\newblock \showarticletitle{Two-dimensional Kripke Semantics {II:} Stability
  and Completeness}.
\newblock \bibinfo{journal}{\emph{CoRR}}  \bibinfo{volume}{abs/2406.03578}
  (\bibinfo{year}{2024}).
\newblock
\urldef\tempurl%
\url{https://doi.org/10.48550/ARXIV.2406.03578}
\showDOI{\tempurl}
\showeprint[arXiv]{2406.03578}


\bibitem[\protect\citeauthoryear{Kock}{Kock}{1972}]%
        {kock1972strong}
\bibfield{author}{\bibinfo{person}{Anders Kock}.}
  \bibinfo{year}{1972}\natexlab{}.
\newblock \showarticletitle{Strong functors and monoidal monads}.
\newblock \bibinfo{journal}{\emph{Archiv der Mathematik}} \bibinfo{volume}{23},
  \bibinfo{number}{1} (\bibinfo{year}{1972}), \bibinfo{pages}{113--120}.
\newblock


\bibitem[\protect\citeauthoryear{Makwana and Krishnaswami}{Makwana and
  Krishnaswami}{2019}]%
        {makwana_numlin_2019}
\bibfield{author}{\bibinfo{person}{Dhruv~C. Makwana} {and}
  \bibinfo{person}{Neelakantan~R. Krishnaswami}.}
  \bibinfo{year}{2019}\natexlab{}.
\newblock \showarticletitle{NumLin: Linear Types for Linear Algebra}. In
  \bibinfo{booktitle}{\emph{33rd European Conference on Object-Oriented
  Programming, {ECOOP} 2019, July 15-19, 2019, London, United Kingdom}}
  \emph{(\bibinfo{series}{LIPIcs})},
  \bibfield{editor}{\bibinfo{person}{Alastair~F. Donaldson}} (Ed.),
  Vol.~\bibinfo{volume}{134}. \bibinfo{publisher}{Schloss Dagstuhl -
  Leibniz-Zentrum f{\"{u}}r Informatik}, \bibinfo{pages}{14:1--14:25}.
\newblock
\showISBNx{978-3-95977-111-5}
\urldef\tempurl%
\url{https://doi.org/10.4230/LIPIcs.ECOOP.2019.14}
\showDOI{\tempurl}


\bibitem[\protect\citeauthoryear{{McBride} and Paterson}{{McBride} and
  Paterson}{2007}]%
        {mcbride_applicative_2007}
\bibfield{author}{\bibinfo{person}{Conor {McBride}} {and} \bibinfo{person}{Ross
  Paterson}.} \bibinfo{year}{2007}\natexlab{}.
\newblock \showarticletitle{Applicative programming with effects}.
\newblock \bibinfo{journal}{\emph{Journal of Functional Programming}}
  \bibinfo{volume}{18}, \bibinfo{number}{01} (\bibinfo{year}{2007}),
  \bibinfo{pages}{1--13}.
\newblock
\urldef\tempurl%
\url{https://doi.org/10.1017/S0956796807006326}
\showDOI{\tempurl}


\bibitem[\protect\citeauthoryear{Moortgat}{Moortgat}{1996}]%
        {moortgat_multimodal_1996}
\bibfield{author}{\bibinfo{person}{Michael Moortgat}.}
  \bibinfo{year}{1996}\natexlab{}.
\newblock \showarticletitle{Multimodal Linguistic Inference}.
\newblock \bibinfo{journal}{\emph{J. Log. Lang. Inf.}} \bibinfo{volume}{5},
  \bibinfo{number}{3/4} (\bibinfo{year}{1996}), \bibinfo{pages}{349--385}.
\newblock
\urldef\tempurl%
\url{https://doi.org/10.1007/BF00159344}
\showDOI{\tempurl}


\bibitem[\protect\citeauthoryear{Murphy, Crary, and Harper}{Murphy
  et~al\mbox{.}}{2005}]%
        {murphy_distributed_2005}
\bibfield{author}{\bibinfo{person}{Tom Murphy}, \bibinfo{person}{Karl Crary},
  {and} \bibinfo{person}{Robert Harper}.} \bibinfo{year}{2005}\natexlab{}.
\newblock \showarticletitle{Distributed Control Flow with Classical Modal
  Logic}. In \bibinfo{booktitle}{\emph{Computer Science Logic, 19th
  International Workshop, {CSL} 2005, 14th Annual Conference of the EACSL,
  Oxford, UK, August 22-25, 2005, Proceedings}} \emph{(\bibinfo{series}{Lecture
  Notes in Computer Science})}, \bibfield{editor}{\bibinfo{person}{C.{-}H.~Luke
  Ong}} (Ed.), Vol.~\bibinfo{volume}{3634}. \bibinfo{publisher}{Springer},
  \bibinfo{pages}{51--69}.
\newblock
\showISBNx{3-540-28231-9}
\urldef\tempurl%
\url{https://doi.org/10.1007/11538363\_6}
\showDOI{\tempurl}


\bibitem[\protect\citeauthoryear{O'Hearn and Pym}{O'Hearn and Pym}{1999}]%
        {ohearn_logic_1999}
\bibfield{author}{\bibinfo{person}{Peter~W O'Hearn} {and}
  \bibinfo{person}{David~J Pym}.} \bibinfo{year}{1999}\natexlab{}.
\newblock \showarticletitle{The logic of bunched implications}.
\newblock \bibinfo{journal}{\emph{Bulletin of Symbolic Logic}}
  \bibinfo{volume}{5}, \bibinfo{number}{2} (\bibinfo{year}{1999}),
  \bibinfo{pages}{215--244}.
\newblock


\bibitem[\protect\citeauthoryear{Orchard, Liepelt, and III}{Orchard
  et~al\mbox{.}}{2019}]%
        {orchard_quantitative_2019}
\bibfield{author}{\bibinfo{person}{Dominic Orchard},
  \bibinfo{person}{Vilem{-}Benjamin Liepelt}, {and}
  \bibinfo{person}{Harley~Eades III}.} \bibinfo{year}{2019}\natexlab{}.
\newblock \showarticletitle{Quantitative program reasoning with graded modal
  types}.
\newblock \bibinfo{journal}{\emph{{PACMPL}}} \bibinfo{volume}{3},
  \bibinfo{number}{{ICFP}} (\bibinfo{year}{2019}),
  \bibinfo{pages}{110:1--110:30}.
\newblock
\urldef\tempurl%
\url{https://doi.org/10.1145/3341714}
\showDOI{\tempurl}


\bibitem[\protect\citeauthoryear{Orchard, Wadler, and Eades}{Orchard
  et~al\mbox{.}}{2020}]%
        {Orchard_2020}
\bibfield{author}{\bibinfo{person}{Dominic Orchard}, \bibinfo{person}{Philip
  Wadler}, {and} \bibinfo{person}{Harley Eades}.}
  \bibinfo{year}{2020}\natexlab{}.
\newblock \showarticletitle{Unifying graded and parameterised monads}.
\newblock \bibinfo{journal}{\emph{Electronic Proceedings in Theoretical
  Computer Science}}  \bibinfo{volume}{317} (\bibinfo{date}{May}
  \bibinfo{year}{2020}), \bibinfo{pages}{18–38}.
\newblock
\showISSN{2075-2180}
\urldef\tempurl%
\url{https://doi.org/10.4204/eptcs.317.2}
\showDOI{\tempurl}


\bibitem[\protect\citeauthoryear{Petricek, Orchard, and Mycroft}{Petricek
  et~al\mbox{.}}{2014}]%
        {petricek_coeffects_2014}
\bibfield{author}{\bibinfo{person}{Tomas Petricek}, \bibinfo{person}{Dominic~A.
  Orchard}, {and} \bibinfo{person}{Alan Mycroft}.}
  \bibinfo{year}{2014}\natexlab{}.
\newblock \showarticletitle{Coeffects: a calculus of context-dependent
  computation}. In \bibinfo{booktitle}{\emph{Proceedings of the 19th {ACM}
  {SIGPLAN} international conference on Functional programming, Gothenburg,
  Sweden, September 1-3, 2014}}, \bibfield{editor}{\bibinfo{person}{Johan
  Jeuring} {and} \bibinfo{person}{Manuel M.~T. Chakravarty}} (Eds.).
  \bibinfo{publisher}{{ACM}}, \bibinfo{pages}{123--135}.
\newblock
\showISBNx{978-1-4503-2873-9}
\urldef\tempurl%
\url{https://doi.org/10.1145/2628136.2628160}
\showDOI{\tempurl}


\bibitem[\protect\citeauthoryear{Reed and Pierce}{Reed and Pierce}{2010}]%
        {reed_distance_2010}
\bibfield{author}{\bibinfo{person}{Jason Reed} {and}
  \bibinfo{person}{Benjamin~C. Pierce}.} \bibinfo{year}{2010}\natexlab{}.
\newblock \showarticletitle{Distance makes the types grow stronger: a calculus
  for differential privacy}. In \bibinfo{booktitle}{\emph{Proceeding of the
  15th {ACM} {SIGPLAN} international conference on Functional programming,
  {ICFP} 2010, Baltimore, Maryland, USA, September 27-29, 2010}},
  \bibfield{editor}{\bibinfo{person}{Paul Hudak} {and}
  \bibinfo{person}{Stephanie Weirich}} (Eds.). \bibinfo{publisher}{ACM Press},
  \bibinfo{pages}{157--168}.
\newblock
\urldef\tempurl%
\url{https://doi.org/10.1145/1863543.1863568}
\showDOI{\tempurl}


\bibitem[\protect\citeauthoryear{Reynolds}{Reynolds}{1983}]%
        {reynolds_types_1983}
\bibfield{author}{\bibinfo{person}{John~C. Reynolds}.}
  \bibinfo{year}{1983}\natexlab{}.
\newblock \showarticletitle{Types, abstraction and parametric polymorphism}.
\newblock \bibinfo{journal}{\emph{Information processing}}
  \bibinfo{volume}{83}, \bibinfo{number}{1} (\bibinfo{year}{1983}),
  \bibinfo{pages}{513--523}.
\newblock


\bibitem[\protect\citeauthoryear{Scott}{Scott}{1980}]%
        {scott1980relating}
\bibfield{author}{\bibinfo{person}{Dana~S Scott}.}
  \bibinfo{year}{1980}\natexlab{}.
\newblock \showarticletitle{Relating theories of the lambda calculus}.
\newblock \bibinfo{journal}{\emph{To HB Curry: Essays on combinatory logic,
  lambda calculus and formalism}} (\bibinfo{year}{1980}),
  \bibinfo{pages}{403--450}.
\newblock


\bibitem[\protect\citeauthoryear{Walker}{Walker}{2005}]%
        {walker:substructural}
\bibfield{author}{\bibinfo{person}{David Walker}.}
  \bibinfo{year}{2005}\natexlab{}.
\newblock \showarticletitle{Substructural Type Systems}.
\newblock In \bibinfo{booktitle}{\emph{Advanced Topics in Types and Programming
  Languages}}, \bibfield{editor}{\bibinfo{person}{Benjamin~C. Pierce}} (Ed.).
  \bibinfo{publisher}{MIT Press}, Chapter~1.
\newblock


\end{thebibliography}



\newpage
\appendix

\section{Typing derivations}

\providecommand\fullinfo[1]{\[\input{moldy/#1_term.tex} : \input{moldy/#1_type.tex}\] %
  \begin{center}%
    \input{moldy/#1_derivation.tex}%
  \end{center}%
}
\providecommand\fullinfolong[2]{\begin{align*} #2 & : \input{moldy/#1_type.tex}\\ #2 &= \input{moldy/#1_term.tex} \end{align*} %
  \begin{center}%
    \input{moldy/#1_derivation.tex}%
  \end{center}%
}

\subsection{Duals}\label[appendix]{app:adjoints}
\paragraph{Admissibility of \(\BOXin\)} \[\pconv{\mathit{η}}{\TBOX{\mathit{p}}{\mathit{t}}} : \mathit{A}⊸\BOX{\mathit{p}}\mathit{B}\text{, ~if~}t : \BOX{\mathit{q}}\mathit{A}⊸\mathit{B}\text{, ~if~}\mathit{η} : 1\convertible{}\mathit{p}\mathit{q}\] %
  \begin{center}%
    {\begin{tikzpicture}\path[-,line width=0.4pt,line cap=butt,line join=miter,dash pattern=](-66.1496pt,19.47pt)--(-24.5388pt,19.47pt)--(-24.5388pt,10.625pt)--(-66.1496pt,10.625pt)--cycle;
\node[anchor=north west,inner sep=0] at (-66.1496pt,19.47pt){\savebox{\marxupbox}{{\(\mathit{η}:1\convertible{}\mathit{p}\bar{p}\)}}\immediate\write\boxesfile{70}\immediate\write\boxesfile{\number\wd\marxupbox}\immediate\write\boxesfile{\number\ht\marxupbox}\immediate\write\boxesfile{\number\dp\marxupbox}\box\marxupbox};
\path[-,line width=0.4pt,line cap=butt,line join=miter,dash pattern=](-67.6496pt,20.97pt)--(-23.0388pt,20.97pt)--(-23.0388pt,20.97pt)--(-67.6496pt,20.97pt)--cycle;
\path[-,line width=0.4pt,line cap=butt,line join=miter,dash pattern=](-67.6496pt,20.97pt)--(-23.0388pt,20.97pt)--(-23.0388pt,20.97pt)--(-67.6496pt,20.97pt)--cycle;
\path[-,draw=white,line width=0.4pt,line cap=butt,line join=miter,dash pattern=](-67.6496pt,20.97pt)--(-23.0388pt,20.97pt);
\path[-,line width=0.4pt,line cap=butt,line join=miter,dash pattern=](45.8804pt,48.5781pt)--(45.8804pt,48.5781pt)--(45.8804pt,35.9581pt)--(45.8804pt,35.9581pt)--cycle;
\path[-,line width=0.4pt,line cap=butt,line join=miter,dash pattern=](41.4204pt,45.0781pt)--(44.3804pt,45.0781pt)--(44.3804pt,39.4581pt)--(41.4204pt,39.4581pt)--cycle;
\node[anchor=north west,inner sep=0] at (41.4204pt,45.0781pt){\savebox{\marxupbox}{{\(\mathit{t}\)}}\immediate\write\boxesfile{71}\immediate\write\boxesfile{\number\wd\marxupbox}\immediate\write\boxesfile{\number\ht\marxupbox}\immediate\write\boxesfile{\number\dp\marxupbox}\box\marxupbox};
\path[-,line width=0.4pt,line cap=butt,line join=miter,dash pattern=](45.8804pt,50.0781pt)--(45.8804pt,50.0781pt)--(45.8804pt,50.0781pt)--(45.8804pt,50.0781pt)--cycle;
\node[anchor=north west,inner sep=0] at (45.8804pt,50.0781pt){\savebox{\marxupbox}{{}}\immediate\write\boxesfile{72}\immediate\write\boxesfile{\number\wd\marxupbox}\immediate\write\boxesfile{\number\ht\marxupbox}\immediate\write\boxesfile{\number\dp\marxupbox}\box\marxupbox};
\path[-,line width=0.4pt,line cap=butt,line join=miter,dash pattern=](44.3804pt,51.5781pt)--(47.3804pt,51.5781pt)--(47.3804pt,51.5781pt)--(44.3804pt,51.5781pt)--cycle;
\path[-,line width=0.4pt,line cap=butt,line join=miter,dash pattern=](44.3804pt,51.5781pt)--(47.3804pt,51.5781pt)--(47.3804pt,51.5781pt)--(44.3804pt,51.5781pt)--cycle;
\path[-,draw=white,line width=0.4pt,line cap=butt,line join=miter,dash pattern=](44.3804pt,51.5781pt)--(47.3804pt,51.5781pt);
\path[-,draw=black,line width=0.4pt,line cap=butt,line join=miter,dash pattern=on 0.4pt off 1pt](45.8804pt,35.9581pt)--(45.8804pt,48.5781pt);
\path[-,line width=0.4pt,line cap=butt,line join=miter,dash pattern=](32.1361pt,34.4581pt)--(59.6246pt,34.4581pt)--(59.6246pt,23.8541pt)--(32.1361pt,23.8541pt)--cycle;
\node[anchor=north west,inner sep=0] at (32.1361pt,34.4581pt){\savebox{\marxupbox}{{\(\moda{\bar{p}}\mathit{Γ} ⊢ \mathit{A}\)}}\immediate\write\boxesfile{73}\immediate\write\boxesfile{\number\wd\marxupbox}\immediate\write\boxesfile{\number\ht\marxupbox}\immediate\write\boxesfile{\number\dp\marxupbox}\box\marxupbox};
\path[-,line width=0.4pt,line cap=butt,line join=miter,dash pattern=](45.8804pt,48.5781pt)--(45.8804pt,48.5781pt)--(45.8804pt,35.9581pt)--(45.8804pt,35.9581pt)--cycle;
\path[-,line width=0.4pt,line cap=butt,line join=miter,dash pattern=](30.6361pt,35.9581pt)--(61.1246pt,35.9581pt)--(61.1246pt,35.9581pt)--(30.6361pt,35.9581pt)--cycle;
\path[-,draw=white,line width=0.4pt,line cap=butt,line join=miter,dash pattern=](30.6361pt,35.9581pt)--(61.1246pt,35.9581pt);
\path[-,line width=0.4pt,line cap=butt,line join=miter,dash pattern=](25.6111pt,20.8541pt)--(66.1496pt,20.8541pt)--(66.1496pt,8.46pt)--(25.6111pt,8.46pt)--cycle;
\node[anchor=north west,inner sep=0] at (25.6111pt,20.8541pt){\savebox{\marxupbox}{{\(\moda{\mathit{p}\bar{p}}\mathit{Γ} ⊢ \BOX{\mathit{p}}\mathit{A}\)}}\immediate\write\boxesfile{74}\immediate\write\boxesfile{\number\wd\marxupbox}\immediate\write\boxesfile{\number\ht\marxupbox}\immediate\write\boxesfile{\number\dp\marxupbox}\box\marxupbox};
\path[-,line width=0.4pt,line cap=butt,line join=miter,dash pattern=](30.6361pt,35.9581pt)--(61.1246pt,35.9581pt)--(61.1246pt,22.3541pt)--(30.6361pt,22.3541pt)--cycle;
\path[-,line width=0.4pt,line cap=butt,line join=miter,dash pattern=](24.1111pt,22.3541pt)--(67.6496pt,22.3541pt)--(67.6496pt,22.3541pt)--(24.1111pt,22.3541pt)--cycle;
\path[-,line width=0.4pt,line cap=butt,line join=miter,dash pattern=](70.6496pt,27.429pt)--(136.6092pt,27.429pt)--(136.6092pt,18.2941pt)--(70.6496pt,18.2941pt)--cycle;
\node[anchor=north west,inner sep=0] at (70.6496pt,27.429pt){\savebox{\marxupbox}{{\(\distributesMaybe{\mathit{Γ}}{\mathit{p}}\)}}\immediate\write\boxesfile{75}\immediate\write\boxesfile{\number\wd\marxupbox}\immediate\write\boxesfile{\number\ht\marxupbox}\immediate\write\boxesfile{\number\dp\marxupbox}\box\marxupbox};
\path[-,line width=0.4pt,line cap=butt,line join=miter,dash pattern=](-15.0388pt,26.9841pt)--(21.1111pt,26.9841pt)--(21.1111pt,16.0941pt)--(-15.0388pt,16.0941pt)--cycle;
\node[anchor=north west,inner sep=0] at (-15.0388pt,26.9841pt){\savebox{\marxupbox}{{\rul{\BOXintro[\mathit{p}]}}}\immediate\write\boxesfile{76}\immediate\write\boxesfile{\number\wd\marxupbox}\immediate\write\boxesfile{\number\ht\marxupbox}\immediate\write\boxesfile{\number\dp\marxupbox}\box\marxupbox};
\path[-,draw=black,line width=0.4pt,line cap=butt,line join=miter,dash pattern=](24.1111pt,22.3541pt)--(67.6496pt,22.3541pt);
\path[-,line width=0.4pt,line cap=butt,line join=miter,dash pattern=](-16.0827pt,5.46pt)--(16.0827pt,5.46pt)--(16.0827pt,-5.46pt)--(-16.0827pt,-5.46pt)--cycle;
\node[anchor=north west,inner sep=0] at (-16.0827pt,5.46pt){\savebox{\marxupbox}{{\(\mathit{Γ} ⊢ \BOX{\mathit{p}}\mathit{A}\)}}\immediate\write\boxesfile{77}\immediate\write\boxesfile{\number\wd\marxupbox}\immediate\write\boxesfile{\number\ht\marxupbox}\immediate\write\boxesfile{\number\dp\marxupbox}\box\marxupbox};
\path[-,line width=0.4pt,line cap=butt,line join=miter,dash pattern=](-67.6496pt,22.3541pt)--(67.6496pt,22.3541pt)--(67.6496pt,6.96pt)--(-67.6496pt,6.96pt)--cycle;
\path[-,line width=0.4pt,line cap=butt,line join=miter,dash pattern=](-67.6496pt,6.96pt)--(67.6496pt,6.96pt)--(67.6496pt,6.96pt)--(-67.6496pt,6.96pt)--cycle;
\path[-,line width=0.4pt,line cap=butt,line join=miter,dash pattern=](-92.4296pt,9.53pt)--(-70.6496pt,9.53pt)--(-70.6496pt,4.875pt)--(-92.4296pt,4.875pt)--cycle;
\node[anchor=north west,inner sep=0] at (-92.4296pt,9.53pt){\savebox{\marxupbox}{{\rul{conv}}}\immediate\write\boxesfile{78}\immediate\write\boxesfile{\number\wd\marxupbox}\immediate\write\boxesfile{\number\ht\marxupbox}\immediate\write\boxesfile{\number\dp\marxupbox}\box\marxupbox};
\path[-,draw=black,line width=0.4pt,line cap=butt,line join=miter,dash pattern=](-67.6496pt,6.96pt)--(67.6496pt,6.96pt);
\end{tikzpicture}}%
  \end{center}%

\paragraph{Admissibility of \(\BOXel\)} \[\tlet{\TBOX{\bar{p}}{\TBOX{\mathit{p}}{\pconv{\mathit{ε}}{\mathit{y}}}}}{\TBOX{\bar{p}}{\mathit{t}}}{\mathit{y}} : \BOX{\mathit{q}}\mathit{A}⊸\mathit{B}\text{, ~if~}t : \mathit{A}⊸\BOX{\mathit{p}}\mathit{B}\text{, ~if~}\mathit{ε} : \mathit{q}\mathit{p}\convertible{}1\] %
  \begin{center}%
    \input{moldy/bar_elim_derivation.tex}%
  \end{center}%

\paragraph{Rule \(\BOXintro[p]\) from \(\BOXin\)}\[\TBOXalt{p}\app{(\pconv{\mathit{ε}}{\mathit{t}})} : \BOX{\mathit{p}}\mathit{A}⊸\BOX{\mathit{p}}\mathit{B}\text{, ~if~}t : \mathit{A}⊸\mathit{B}\text{, ~if~}\mathit{ε} : \mathit{q}\mathit{p}\convertible{}1\text{, ~if~}\mathit{η} : 1\convertible{}\mathit{p}\mathit{q}\] %
  \begin{center}%
    {\begin{tikzpicture}\path[-,line width=0.4pt,line cap=butt,line join=miter,dash pattern=](-38.6981pt,31.011pt)--(2.9327pt,31.011pt)--(2.9327pt,22.166pt)--(-38.6981pt,22.166pt)--cycle;
\node[anchor=north west,inner sep=0] at (-38.6981pt,31.011pt){\savebox{\marxupbox}{{\(\mathit{ε}:\bar{p}\mathit{p}\convertible{}1\)}}\immediate\write\boxesfile{96}\immediate\write\boxesfile{\number\wd\marxupbox}\immediate\write\boxesfile{\number\ht\marxupbox}\immediate\write\boxesfile{\number\dp\marxupbox}\box\marxupbox};
\path[-,line width=0.4pt,line cap=butt,line join=miter,dash pattern=](-40.1981pt,32.511pt)--(4.4327pt,32.511pt)--(4.4327pt,32.511pt)--(-40.1981pt,32.511pt)--cycle;
\path[-,line width=0.4pt,line cap=butt,line join=miter,dash pattern=](-40.1981pt,32.511pt)--(4.4327pt,32.511pt)--(4.4327pt,32.511pt)--(-40.1981pt,32.511pt)--cycle;
\path[-,draw=white,line width=0.4pt,line cap=butt,line join=miter,dash pattern=](-40.1981pt,32.511pt)--(4.4327pt,32.511pt);
\path[-,line width=0.4pt,line cap=butt,line join=miter,dash pattern=](27.3154pt,45.041pt)--(27.3154pt,45.041pt)--(27.3154pt,32.421pt)--(27.3154pt,32.421pt)--cycle;
\path[-,line width=0.4pt,line cap=butt,line join=miter,dash pattern=](22.8554pt,41.541pt)--(25.8154pt,41.541pt)--(25.8154pt,35.921pt)--(22.8554pt,35.921pt)--cycle;
\node[anchor=north west,inner sep=0] at (22.8554pt,41.541pt){\savebox{\marxupbox}{{\(\mathit{t}\)}}\immediate\write\boxesfile{97}\immediate\write\boxesfile{\number\wd\marxupbox}\immediate\write\boxesfile{\number\ht\marxupbox}\immediate\write\boxesfile{\number\dp\marxupbox}\box\marxupbox};
\path[-,line width=0.4pt,line cap=butt,line join=miter,dash pattern=](27.3154pt,46.541pt)--(27.3154pt,46.541pt)--(27.3154pt,46.541pt)--(27.3154pt,46.541pt)--cycle;
\node[anchor=north west,inner sep=0] at (27.3154pt,46.541pt){\savebox{\marxupbox}{{}}\immediate\write\boxesfile{98}\immediate\write\boxesfile{\number\wd\marxupbox}\immediate\write\boxesfile{\number\ht\marxupbox}\immediate\write\boxesfile{\number\dp\marxupbox}\box\marxupbox};
\path[-,line width=0.4pt,line cap=butt,line join=miter,dash pattern=](25.8154pt,48.041pt)--(28.8154pt,48.041pt)--(28.8154pt,48.041pt)--(25.8154pt,48.041pt)--cycle;
\path[-,line width=0.4pt,line cap=butt,line join=miter,dash pattern=](25.8154pt,48.041pt)--(28.8154pt,48.041pt)--(28.8154pt,48.041pt)--(25.8154pt,48.041pt)--cycle;
\path[-,draw=white,line width=0.4pt,line cap=butt,line join=miter,dash pattern=](25.8154pt,48.041pt)--(28.8154pt,48.041pt);
\path[-,draw=black,line width=0.4pt,line cap=butt,line join=miter,dash pattern=on 0.4pt off 1pt](27.3154pt,32.421pt)--(27.3154pt,45.041pt);
\path[-,line width=0.4pt,line cap=butt,line join=miter,dash pattern=](15.9327pt,30.921pt)--(38.6981pt,30.921pt)--(38.6981pt,24.261pt)--(15.9327pt,24.261pt)--cycle;
\node[anchor=north west,inner sep=0] at (15.9327pt,30.921pt){\savebox{\marxupbox}{{\(\mathit{Γ} ⊢ \mathit{A}\)}}\immediate\write\boxesfile{99}\immediate\write\boxesfile{\number\wd\marxupbox}\immediate\write\boxesfile{\number\ht\marxupbox}\immediate\write\boxesfile{\number\dp\marxupbox}\box\marxupbox};
\path[-,line width=0.4pt,line cap=butt,line join=miter,dash pattern=](27.3154pt,45.041pt)--(27.3154pt,45.041pt)--(27.3154pt,32.421pt)--(27.3154pt,32.421pt)--cycle;
\path[-,line width=0.4pt,line cap=butt,line join=miter,dash pattern=](14.4327pt,32.421pt)--(40.1981pt,32.421pt)--(40.1981pt,32.421pt)--(14.4327pt,32.421pt)--cycle;
\path[-,draw=white,line width=0.4pt,line cap=butt,line join=miter,dash pattern=](14.4327pt,32.421pt)--(40.1981pt,32.421pt);
\path[-,line width=0.4pt,line cap=butt,line join=miter,dash pattern=](-15.5693pt,19.166pt)--(15.5693pt,19.166pt)--(15.5693pt,8.5619pt)--(-15.5693pt,8.5619pt)--cycle;
\node[anchor=north west,inner sep=0] at (-15.5693pt,19.166pt){\savebox{\marxupbox}{{\(\moda{\bar{p}\mathit{p}}\mathit{Γ} ⊢ \mathit{A}\)}}\immediate\write\boxesfile{100}\immediate\write\boxesfile{\number\wd\marxupbox}\immediate\write\boxesfile{\number\ht\marxupbox}\immediate\write\boxesfile{\number\dp\marxupbox}\box\marxupbox};
\path[-,line width=0.4pt,line cap=butt,line join=miter,dash pattern=](-40.1981pt,32.511pt)--(40.1981pt,32.511pt)--(40.1981pt,20.666pt)--(-40.1981pt,20.666pt)--cycle;
\path[-,line width=0.4pt,line cap=butt,line join=miter,dash pattern=](-40.1981pt,20.666pt)--(40.1981pt,20.666pt)--(40.1981pt,20.666pt)--(-40.1981pt,20.666pt)--cycle;
\path[-,line width=0.4pt,line cap=butt,line join=miter,dash pattern=](-64.9781pt,23.236pt)--(-43.1981pt,23.236pt)--(-43.1981pt,18.581pt)--(-64.9781pt,18.581pt)--cycle;
\node[anchor=north west,inner sep=0] at (-64.9781pt,23.236pt){\savebox{\marxupbox}{{\rul{conv}}}\immediate\write\boxesfile{101}\immediate\write\boxesfile{\number\wd\marxupbox}\immediate\write\boxesfile{\number\ht\marxupbox}\immediate\write\boxesfile{\number\dp\marxupbox}\box\marxupbox};
\path[-,draw=black,line width=0.4pt,line cap=butt,line join=miter,dash pattern=](-40.1981pt,20.666pt)--(40.1981pt,20.666pt);
\path[-,line width=0.4pt,line cap=butt,line join=miter,dash pattern=](-17.6577pt,5.5619pt)--(17.6577pt,5.5619pt)--(17.6577pt,-5.5619pt)--(-17.6577pt,-5.5619pt)--cycle;
\node[anchor=north west,inner sep=0] at (-17.6577pt,5.5619pt){\savebox{\marxupbox}{{\(\moda{\mathit{p}}\mathit{Γ} ⊢ \BOX{\mathit{p}}\mathit{A}\)}}\immediate\write\boxesfile{102}\immediate\write\boxesfile{\number\wd\marxupbox}\immediate\write\boxesfile{\number\ht\marxupbox}\immediate\write\boxesfile{\number\dp\marxupbox}\box\marxupbox};
\path[-,line width=0.4pt,line cap=butt,line join=miter,dash pattern=](-17.0693pt,20.666pt)--(17.0693pt,20.666pt)--(17.0693pt,7.0619pt)--(-17.0693pt,7.0619pt)--cycle;
\path[-,line width=0.4pt,line cap=butt,line join=miter,dash pattern=](-19.1577pt,7.0619pt)--(19.1577pt,7.0619pt)--(19.1577pt,7.0619pt)--(-19.1577pt,7.0619pt)--cycle;
\path[-,line width=0.4pt,line cap=butt,line join=miter,dash pattern=](-59.4376pt,12.4319pt)--(-22.1577pt,12.4319pt)--(-22.1577pt,0.767pt)--(-59.4376pt,0.767pt)--cycle;
\node[anchor=north west,inner sep=0] at (-59.4376pt,12.4319pt){\savebox{\marxupbox}{{\(\BOXin{}\)}}\immediate\write\boxesfile{103}\immediate\write\boxesfile{\number\wd\marxupbox}\immediate\write\boxesfile{\number\ht\marxupbox}\immediate\write\boxesfile{\number\dp\marxupbox}\box\marxupbox};
\path[-,draw=black,line width=0.4pt,line cap=butt,line join=miter,dash pattern=](-19.1577pt,7.0619pt)--(19.1577pt,7.0619pt);
\end{tikzpicture}}%
  \end{center}%

\subsection{Effects}
\paragraph{bind}\begin{align*} bind & : \BOX{\mathit{p}}\mathit{A}⊸(\mathit{A}⊸\BOX{\mathit{p}}\mathit{B})⊸\BOX{\mathit{p}}\mathit{B}\text{, ~if~}\mathit{μ} : \mathit{p}\mathit{p}\convertible{}\mathit{p}\text{, ~if~}\mathit{η} : 1\convertible{}\mathit{p}\\ bind &= \lam{\TBOX{\mathit{p}}{\mathit{x}}}{\lam{(\pconv{\mathit{η}}{\mathit{f}})}{\tlet{\TBOX{\mathit{p}}{\TBOX{\mathit{p}}{\pconv{\mathit{μ}}{\mathit{y}}}}}{\TBOX{\mathit{p}}{\mathit{f}\app{\mathit{x}}}}{\TBOX{\mathit{p}}{\mathit{y}}}}} \end{align*} %
  \begin{center}%
    \input{moldy/bind_derivation.tex}%
  \end{center}%

\paragraph{example program (increment with state effect)} \[\tlet{\TBOX{\mathit{p}}{\mathit{y}}}{\mathit{get}}{\tlet{\TBOX{\mathit{p}}{\TBOX{\mathit{p}}{\pconv{\mathit{μ}}{\mathit{z}}}}}{\TBOX{\mathit{p}}{\mathit{put}\app{(\mathit{inc}\app{\mathit{y}})}}}{\TBOX{\mathit{p}}{\tlet{()}{\mathit{z}}{42}}}} : \BOX{\mathit{p}}\mathit{ℕ}\] %
  \begin{center}%
    \input{moldy/state_example_derivation.tex}%
  \end{center}%

\subsection{Co-effects}
\paragraph{example program (increment with coeffect)} \[\lam{\TBOX{\bar{p}}{\pconv{\mathit{δ}}{\mathit{u}}}}{\tlet{()}{\mathit{coput}\app{\TBOX{\bar{p}}{\mathit{inc}\app{(\mathit{coget}\app{\TBOX{\bar{p}}{\mathit{u}}})}}}}{42}} : \BOX{\bar{p}}1⊸\mathit{ℕ}\] %
  \begin{center}%
    \input{moldy/costate_example_derivation.tex}%
  \end{center}%

\paragraph{duplication of coeffect sources} \[\lam{\TBOX{\bar{p}}{\mathit{x}}}{\text{\ref{principle:oplax}}\app{\TBOX{\bar{p}}{\tlet{()}{\mathit{x}}{((),())}}}} : \BOX{\bar{p}}1⊸\BOX{\bar{p}}1⊗\BOX{\bar{p}}1\] %
  \begin{center}%
    \input{moldy/dupSinks_derivation.tex}%
  \end{center}%

\end{document}